\documentclass[journal, 10pt]{IEEEtran}
\IEEEoverridecommandlockouts
\normalsize
\bstctlcite{IEEEexample:BSTcontrol}
\usepackage{cite}
\usepackage{url}
\usepackage{hyperref}
\hypersetup{colorlinks=true, linkcolor=blue, citecolor=red, urlcolor=blue} 
\usepackage{amssymb}
\usepackage{amsmath}
\usepackage{array}
\usepackage{amsthm}
\allowdisplaybreaks 
\usepackage{autobreak}
\usepackage{mathptmx}
\usepackage{mathtools, cuted}
\usepackage{algpseudocode}
\usepackage[ruled, vlined, linesnumbered]{algorithm2e}
\usepackage{graphicx}
\usepackage{bbm}
\usepackage{tablefootnote}
\usepackage{tabu}
\usepackage{multirow}
\usepackage{array,booktabs}
\usepackage[table, dvipsnames]{xcolor}

\allowdisplaybreaks

\newtheorem{Theorem}{Theorem}
\newtheorem{Lemma}{Lemma}

\newtheorem{Remark}{Remark}

\newcommand{\rs}{\!\!}
\newcolumntype{C}[1]{>{\centering \arraybackslash}p{#1}}

\newcommand{\bviol}{\textcolor{violet}}

\title{Resource Constrained Vehicular Edge Federated Learning with Highly Mobile Connected Vehicles}  
\author{Md Ferdous Pervej, \IEEEmembership{Graduate Student Member, IEEE}, Richeng Jin, \IEEEmembership{Member, IEEE}, and Huaiyu Dai, \IEEEmembership{Fellow, IEEE}
\thanks{
This research was supported in part by the Zhejiang Provincial Natural Science Foundation of China under Grant No. LQ23F010021, in part by the Ng Teng Fong Charitable Foundation in the form of ZJU-SUTD IDEA Grant under Grant No. 188170-11102, in part by the National Key Research and Development Program of China under Grant 2018YFB1801104, and in part by the US National Science Foundation under grants CNS-1824518 and ECCS-2203214. (\textit{Corresponding author: Richeng Jin}.)}
\thanks{M. F. Pervej and H. Dai are with the Department of Electrical and Computer Engineering, NC State University, Raleigh, NC 27695, USA (e-mails: \{mpervej, hdai\}@ncsu.edu).}
\thanks{R. Jin is with the Zhejiang\text{–}Singapore Innovation and AI Joint Research Lab, the Department of Information and Communication Engineering, Zhejiang University, Hangzhou, China, 310007, and also with Zhejiang Provincial Key Lab of Information Processing, Communication, and Networking (IPCAN), Hangzhou, China, 310007 (e-mail: richengjin@zju.edu.cn).}
\vspace{-0.35 in}
}

\begin{document}
\maketitle
\IEEEpeerreviewmaketitle

\begin{abstract}
This paper proposes a vehicular edge federated learning (VEFL) solution, where an edge server leverages highly mobile connected vehicles' (CVs') onboard central processing units (CPUs) and local datasets to train a global model. 
Convergence analysis reveals that the VEFL training loss depends on the successful receptions of the CVs' trained models over the intermittent vehicle-to-infrastructure (V$2$I) wireless links.
Owing to high mobility, in the full device participation case (FDPC), the edge server aggregates client model parameters based on a weighted combination according to
the CVs' dataset sizes and sojourn periods, while it selects a subset of CVs in the partial device participation case (PDPC).
We then devise joint VEFL and radio access technology (RAT) parameters optimization problems under delay, energy and cost constraints to maximize the probability of successful reception of the locally trained models.
Considering that the optimization problem is NP-hard, we decompose it into a VEFL parameter optimization sub-problem, given the estimated worst-case sojourn period, delay and energy expense, and an online RAT parameter optimization sub-problem.
Finally, extensive simulations are conducted to validate the effectiveness of the proposed solutions with a practical $5$G new radio ($5$G-NR) RAT under a realistic microscopic mobility model.
\end{abstract}


\section{Introduction}
\noindent
\IEEEPARstart{W}{hile} modern connected vehicles (CVs) are an essential part of an intelligent transportation system (ITS), higher automation on the road demands more exploration.
One way to achieve higher automation is to put more sensors on the onboard units of these CVs to facilitate real-time sensing and onboard computing \cite{Qayyum2020Securing}.
Machine learning (ML) has shown its potential in various ITS applications, such as object detection, traffic sign classification, congestion prediction, velocity/acceleration prediction, etc., to name a few \cite{ye2018machine}.
However, the sensing capabilities and onboard computation powers of CVs are still limited.
Moreover, offloading raw data to an edge server raises immense privacy risks and requires humongous bandwidth. 
Therefore, a privacy-preserving distributed ML solution is urgently needed for modern vehicular edge networks (VENs) to ensure higher automation levels on the road where the moving CVs must make operational decisions swiftly.

With its privacy-preserving and distributed learning abilities, federated learning (FL) \cite{mcmahan2017communication} is, thus, an ideal solution for VENs.
Note that FL follows the parameter server paradigm, where the server distributes a global ML model to the clients, who then perform local model training in parallel on their devices and send their locally trained model parameters to the server \cite{richeng2022communication}.
Thus, the CVs do not need to share their raw data, i.e., data remains private. 
Besides, system and data heterogeneity of the CVs can be handled by carefully designing model aggregation rules and local training loss functions.

Unlike traditional stationary clients, however, devising a vehicular edge FL (VEFL) framework is challenging for multiple reasons.
Firstly, limited radio coverage makes the sojourn periods of the highly mobile CVs very short.
Therefore, the CVs can perform local model training only for a few iterations before moving out of the coverage area.
Secondly, modern CVs' onboard central processing units (CPUs) are responsible for many operational computations. 
Besides, the CVs are owned by different clients who may not readily join the FL process.
Therefore, a service level agreement (SLA) between a CV that wishes to utilize its limited resource for FL model training and the edge server should exist.
Note that an SLA is a commitment between the server and the CV that both parties agree to uphold.
Thirdly, a proper radio access technology (RAT) solution is required since the server can aggregate trained models only if these models are successfully received at the aggregation time.
However, the high mobility of the CVs makes communication over the intermittent wireless vehicle-to-infrastructure (V$2$I) links even more challenging.
As such, we shall carefully orchestrate the interplay between the server and the RAT solution to perform VEFL.
Moreover, the underlying RAT requires mandatory resource management.
Finally, system and data heterogeneity among the CVs is a norm in VENs since automotive makers produce products with different features.

\subsection{Related Work}
\noindent
We have seen many remarkable contributions to joint FL and wireless network parameter optimizations \cite{richeng2022communication,Amiri2020Federated,salehi2021Federated,Chen2021Ajoint,Yang2021}. 
However, these studies did not consider the fundamental constraint in VEN, i.e., client's high mobility, which results in a very short sojourn period.
Some recent works \cite{Zeng2021MultiTask,pervej2022Mobility,8917592,9492053,9424984,9416835,9181482,WANG2022199, Sun2022Edge, prathiba2021Federated} also considered FL for different tasks in VENs.
However, only a handful of studies \cite{zeng2022federated,Xiao2021Vehicle,taik2022clustered,liu2022FedCPF} addressed the constraints present in VENs.
Zeng \textit{et al.} proposed a dynamic federated proximal algorithm to design a controller for autonomous vehicles in \cite{zeng2022federated}.
The authors considered moving connected and autonomous vehicles as FL clients and devised their algorithm accounting for the communication and computation delay constraints.
The communication delay was derived using a simplistic channel model with one channel realization. 
Each vehicle had a unique orthogonal resource block to offload its trained model to the server.

Xiao \textit{et al.} jointly consider vehicular client selection, transmission power selection, CPU frequency selection and local model accuracy optimization under delay and energy constraints in \cite{Xiao2021Vehicle}.
More specifically, the authors assumed data quality is known, and optimized local model precision before the server performs global aggregation.
Besides, the authors considered a transmission control protocol/internet protocol based channel model with a unique radio resource for each client for offloading its trained model.
Taik \textit{et al.} considered a clustered vehicular FL in \cite{taik2022clustered}.
The authors used the traditional federated averaging (FedAvg) algorithm, where the vehicular cluster head performed model aggregation from the cluster members and then forwarded the aggregated model to the server.
Liu \textit{et al.} used a proximal FL algorithm, which is very similar to the widely used FedProx algorithm \cite{MLSYS2020_38af8613}, for vehicular edge computing in \cite{liu2022FedCPF}. 
The impact of mobility and wireless links was not considered in \cite{liu2022FedCPF}.

Asynchronous communication and model aggregation mechanisms were also proposed in some recent works \cite{liang2022Semi,yang2022efficient,chen2021semi}.
More specifically, \cite{liang2022Semi} considered a semi-synchronous FL for the Internet of vehicles, where the authors dynamically adjusted the server's waiting time between two global rounds in proportion to the total participating clients.
A hierarchical asynchronous FL was considered in \cite{yang2022efficient}. 
A semi-asynchronous hierarchical FL for transportation system was proposed in \cite{chen2021semi}. 
Particularly, \cite{chen2021semi} assumed a synchronous model aggregation for the local-edge level and a semi-synchronous model aggregation for the edge-cloud level.

\subsection{Motivations and Our Contributions}
\noindent
While \cite{Zeng2021MultiTask,pervej2022Mobility,8917592,9492053,9424984,9416835,9181482,WANG2022199, Sun2022Edge, prathiba2021Federated} showed the efficacy of FL in different vehicular applications and \cite{zeng2022federated,Xiao2021Vehicle,taik2022clustered,liu2022FedCPF} addressed some typical resource-constraints in VENs, these studies had their own limits.
Particularly, due to the intermittent wireless V2I links, VEFL is not as straightforward as broadcasting the global model and then aggregating locally trained model parameters under perfect wireless communication links between the server and clients.
A practical RAT, such as the $5$G new radio ($5$G-NR), is required for the parameter server to broadcast the global model in the downlink and then receive the model from the vehicular clients in the uplink.
Moreover, the parameter server must devise the VEFL strategy to accommodate the underlying RAT's characteristics.
As such, a joint study should address the constraints of the parameter server, mobile clients and the underlying RAT.
It is worth pointing out that $3$rd generation partnership project ($3$GPP) release $18$ will include different artificial intelligence and ML solutions for its data-driven network applications \cite{lin2022anOverview}.
Besides, different work groups within $3$GPP are working actively to include ML in the next-generation standard.
Moreover, ML-application-based RAT design is also a part of standardization for release $18$ \cite{3gppSA23_700}.

In this work, we, therefore, present a VEFL framework with a joint study of the impact of the mobility of the clients, i.e., the CVs, with a practical 5G-NR-based RAT solution and under strict delay, energy, computation resource, radio resource and cost constraints.
More specifically, our contributions are summarized as follows:
\begin{itemize}
    \item Leveraging $5$G-NR RAT, we propose a VEFL framework where an edge server utilizes a fixed bandwidth part (BWP) \cite{3gpp38_211} and an uplink heavy frame structure to receive the locally trained ML models over the intermittent V$2$I links from highly mobile CVs which participate in the model training and charge the server based on SLAs.
    \item We consider a full device participation case (FDPC) and a more practical partial device participation case (PDPC), where all CVs and only a subset of CVs participate in the model training, respectively. As FDPC is less flexible, to combat high mobility, i.e., short sojourn period, the server aggregates local model parameters based on a weighted combination reflecting the CVs' expected sojourn periods and dataset sizes. 
    \item In both cases, corresponding joint VEFL and RAT parameter optimization problems are formulated to maximize the probability of successful trained models reception at the server under strict delay, energy and cost constraints. 
    Since channel state information (CSI) can vary in each slot and is unknown beforehand, the original joint problem is decomposed into a VEFL parameter optimization sub-problem, given the upper bounds of the communication delay, energy expense and cost, and a RAT parameter optimization sub-problem that aims to maximize long-term energy-efficiency (EE).
    \item The non-convex VEFL parameter optimization sub-problems are solved using standard relaxations and the difference between convex (DC) approach.
    The fractional non-convex long-term EE optimization problem is first transformed into a tractable form using the Dinkelbach method, which is further converted into a per-slot online optimization problem leveraging Lyapunov drift-plus-penalty-based stochastic optimization.
    \item Finally, using simulation of urban mobility (SUMO) \cite{sumo2018}, we simulate a microscopic mobility scenario in Downtown Raleigh, NC, USA, and use four popular ML datasets to show the effectiveness of our proposed solutions.
\end{itemize}

The rest of the paper is organized as follows: Section \ref{sysModelSection} introduces our proposed VEFL system model. 
Section \ref{convergenceProbFormSection} provides the convergence analysis and our joint problem formulation.  
Section \ref{PDPC_Transformation_Solution_Section} presents the solution to the problem.
We discuss our simulation results in Section \ref{simulationResults_Section}. 
Finally, Section \ref{conclusion} concludes the paper.

\begin{figure}[!t] \vspace{-0.05in}
    \centering
	\includegraphics[trim={170 10 280 1}, clip, width=0.43\textwidth, height=0.24\textheight]{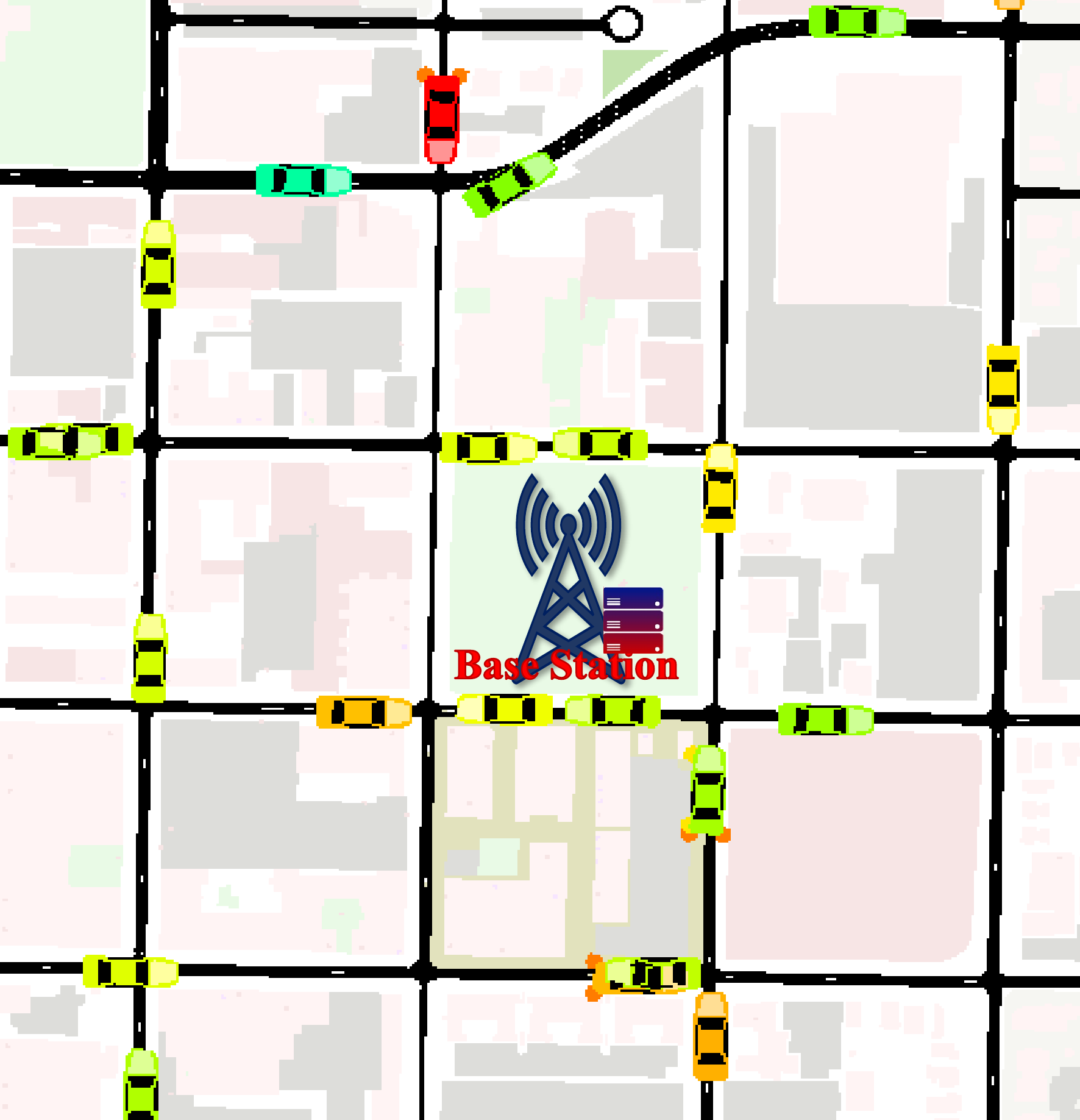} \vspace{-0.02 in}
    \caption{Vehicular FL system model}
    \label{system_model_VEFL}
\end{figure}

\section{VEFL System Model} 
\label{sysModelSection}

\noindent
We consider a VEN confined within a region of interest (RoI), as shown in Fig. \ref{system_model_VEFL}.
The edge server---embedded into the next generation Node B (gNB)---wishes to perform a distributed ML task leveraging the moving CVs' onboard CPUs and local datasets.
Similar to \cite{kim2023modems,Guo2019Enabling}, SLAs between the CVs and the edge server, which require the edge server to pay the CVs for contributing to the VEFL task, are assumed\footnote{However, the implementation of SLAs in a practical VEN implicates the core network and the user plane and control plane protocol stacks\cite{3gppTS23_501, 3gppTS23_502}, which are beyond the scope of this paper.}.
Note that the terms server and gNB are used interchangeably when there is no ambiguity.
We consider a general learning task, which can be object detection, traffic sign classification/detection, velocity/acceleration prediction, traffic congestion prediction, travel time prediction, fuel consumption prediction, etc., for our VEFL.
Moreover, the VEN operates in a discrete time-slotted manner.
The slots are denoted by $\mathcal{T}=\{t\}_{t=1}^{|\mathcal{T}|}$.
Particularly, the gNB has a fixed BWP for the VEFL to provide radio connectivity to the moving client CVs. 
The server can only leverage the trained ML model on the CVs' onboard CPUs within the communication range of the gNB.
Denote the communication radius of the gNB by $\mathrm{r}$.

The CVs enter and leave the RoI following some distributions, i.e., the CV set may not be the same in all time slots due to high mobility. 
Denote the CV set during time slot $t$ by $\mathcal{V}_t=\{v\}_{v=1}^{V_t}$.
Moreover, the server knows the maximum possible velocity on the roads inside this RoI.
Denote the maximum possible velocity of the RoI by $\mathrm{u}^{\mathrm{max}}$.
Assuming the gNB is located at the center of the coordinate system of the RoI, its coverage circle is given by
\begin{equation}
\label{eqnOfCircle}
    x^2+y^2 = \mathrm{r}^2,
\end{equation}
where $x$ and $y$ are the horizontal and vertical coordinates, respectively.

\subsection{CV Mobility Model}
\label{CVs_Mobility_Model}
\noindent
Denote the coordinate of $v\in \mathcal{V}_t$, during slot $t$, by $\left(\mathrm{x}^{\text{loc}}_{v}(t), \mathrm{y}^{\text{loc}}_{v}(t) \right)$ as shown in Fig. \ref{boundaryPoints}.
Also, let us denote CV $v$'s velocity and acceleration during time $t$ by $u_{v}(t)$ and $\dot{u}_{v}(t)$, respectively.
We consider the widely used car-following mobility (CFP) model, known as the intelligent driver model (IDM) mobility model, to model microscopic mobility for the CVs \cite{roy2010handbook}. 
In IDM, the mobility is controlled by the following instantaneous acceleration equation \cite{roy2010handbook}:
\begin{equation}
\label{idmAcce}
    \dot{u}_{v}(t) =  \bar{\mathrm{u}}_1\left(1 - [u_{v}(t) / \mathrm{u}^{\mathrm{max}}]^{4}\right) - \bar{\mathrm{u}}_1[s_v^{*}/\Delta d_{v} (t) ]^2,
\end{equation}
where $\bar{\mathrm{u}}_1$ is the maximum acceleration or a constant that depends on the design. 
$\Delta d_{v} (t)$ is the front bumper to the back bumper distance of CV $v$ and the vehicle directly in front of it, during time $t$. 
Furthermore, $s_v^{*}$ is the desired dynamical distance and is calculated as follows:
\begin{equation}
\label{idmDynDist}
   s_v^{*} = s^{\mathrm{saf}} + u_{v}(t) \cdot t^{\mathrm{dst}} + [u_{v}(t) \cdot \Delta u_v (t)]/(2\sqrt{\bar{\mathrm{u}}_1 \bar{\mathrm{u}}_2}), 
\end{equation}
where $s^{\mathrm{saf}}$ is the safety distance between $v$ and the CV directly in front of it, $t^{\mathrm{dst}}$ is the desired time headway that gives the minimum possible time to the CV directly in front of $v$, $\Delta u_{v}(t)$ is the velocity difference between $v$ and the vehicle in front of it, and $\bar{\mathrm{u}}_2$ is a positive number that defines the comfortable braking deceleration. 
Note that, in (\ref{idmAcce}), the first term, i.e., $\bar{\mathrm{u}}_1\left(1 - [u_{v}(t) / \mathrm{u}^{\mathrm{max}}]^{4}\right)$, is the instantaneous acceleration of CV $v$, which essentially is the desired acceleration on a free road.
The second term is the deceleration induced by the CV in front of it \cite{roy2010handbook}.

Moreover, we assume the server does not need to know the entire trajectory of the vehicle. 
Therefore, the server will only estimate the guaranteed sojourn period $\mathrm{t}_{v,t}^{\text{soj}}$.
To find $\mathrm{t}_{v}^{\mathrm{soj}} (t)$, first, we write the horizontal and vertical lines that intersect the gNB' radio coverage circle boundary as follows:\\
\begin{minipage}{0.23\textwidth}
    \begin{align}
    y &= \mathrm{y}^{\text{loc}}_{v}(t), \label{horizontalLine}
\end{align} 
\end{minipage}
\begin{minipage}{0.23\textwidth}
    \begin{align}
    x &= \mathrm{x}^{\text{loc}}_{v}(t). \label{verticalLine}
\end{align} 
\end{minipage}
Note that (\ref{horizontalLine}) and (\ref{verticalLine}) are represented by the purple and orange color chords, respectively, in Fig. \ref{boundaryPoints}.
We can then find the $x$-coordinates of the gNB coverage boundary by solving (\ref{eqnOfCircle}) and (\ref{horizontalLine}) as $x_{v}^{\mathrm{bnd},1}(t) \rs = \rs \sqrt{\mathrm{r}^2 - {\mathrm{y}_{v}^{\mathrm{loc}}(t) }^2 }$ and $x_{v}^{\mathrm{bnd},2}(t) \rs = \rs -\sqrt{\mathrm{r}^2 - {\mathrm{y}_{v}^{\mathrm{loc}}(t) }^2 }$.
As such, we can find the horizontal distances of these boundary points and the current location of the CV $(\mathrm{x}^{\text{loc}}_{v}(t), \mathrm{y}^{\text{loc}}_{v}(t) )$ as $\mathrm{d}_{x_{1}}^y=\left|\mathrm{x}^{\text{loc}}_{v}(t) - \mathrm{x}^{\mathrm{bnd},1}_{v} (t)\right|$ and $\mathrm{d}_{x_{2}}^y=\left|\mathrm{x}^{\text{loc}}_{v}(t) - \mathrm{x}^{\mathrm{bnd},2}_{v} (t)\right|$.
Similarly, we can find the $y$-coordinates of the gNB coverage boundary by solving (\ref{eqnOfCircle}) and (\ref{verticalLine}) as $y_{v}^{\mathrm{bnd},1}(t) \rs = \rs  \sqrt{\mathrm{r}^2 - {\mathrm{x}_{v}^{\mathrm{loc}}(t) }^2}$ and $y_{v}^{\mathrm{bnd},2}(t) \rs = \rs -\sqrt{\mathrm{r}^2 - {\mathrm{x}_{v}^{\mathrm{loc}}(t) }^2}$.
Then, the corresponding vertical distances of these boundary points and the current location of the CV $(\mathrm{x}^{\text{loc}}_{v}(t), \mathrm{y}^{\text{loc}}_{v}(t) )$ are calculated as $\mathrm{d}_{x}^{y_{1}} = \big|\mathrm{y}^{\text{loc}}_{v}(t) - \mathrm{y}^{\mathrm{bnd},1}_{v} (t) \big|$ and $\mathrm{d}_{x}^{y_{2}} = \big|\mathrm{y}^{\text{loc}}_{v}(t) - \mathrm{y}^{\mathrm{bnd},2}_{v} (t) \big|$.

As such, from time $t$, the minimum expected sojourn period of $v$ under the gNB's coverage is bounded below by
\begin{equation}
\label{sojournPeriod}
    \mathrm{t}_{v}^{\mathrm{soj}} (t) \geq  \mathrm{min} \big\{ \mathrm{d}_{x_{1}}^y, \mathrm{d}_{x_{2}}^y, \mathrm{d}_x^{y^{1}}, \mathrm{d}_x^{y^{2}} \big\}/ \mathrm{u}^{\mathrm{max}}.
\end{equation}
Note that (\ref{sojournPeriod}) is based on a linear trajectory and uniform velocity $\mathrm{u}^{\mathrm{max}}$, which is the worst-case estimation of the sojourn period of the CV. 

\begin{figure*}[!t] \vspace{-0.05in}
\begin{minipage}{0.35\textwidth}
    \centering
    \includegraphics[trim=150 52 260 5, clip, width=0.95\textwidth, height=0.2\textheight]{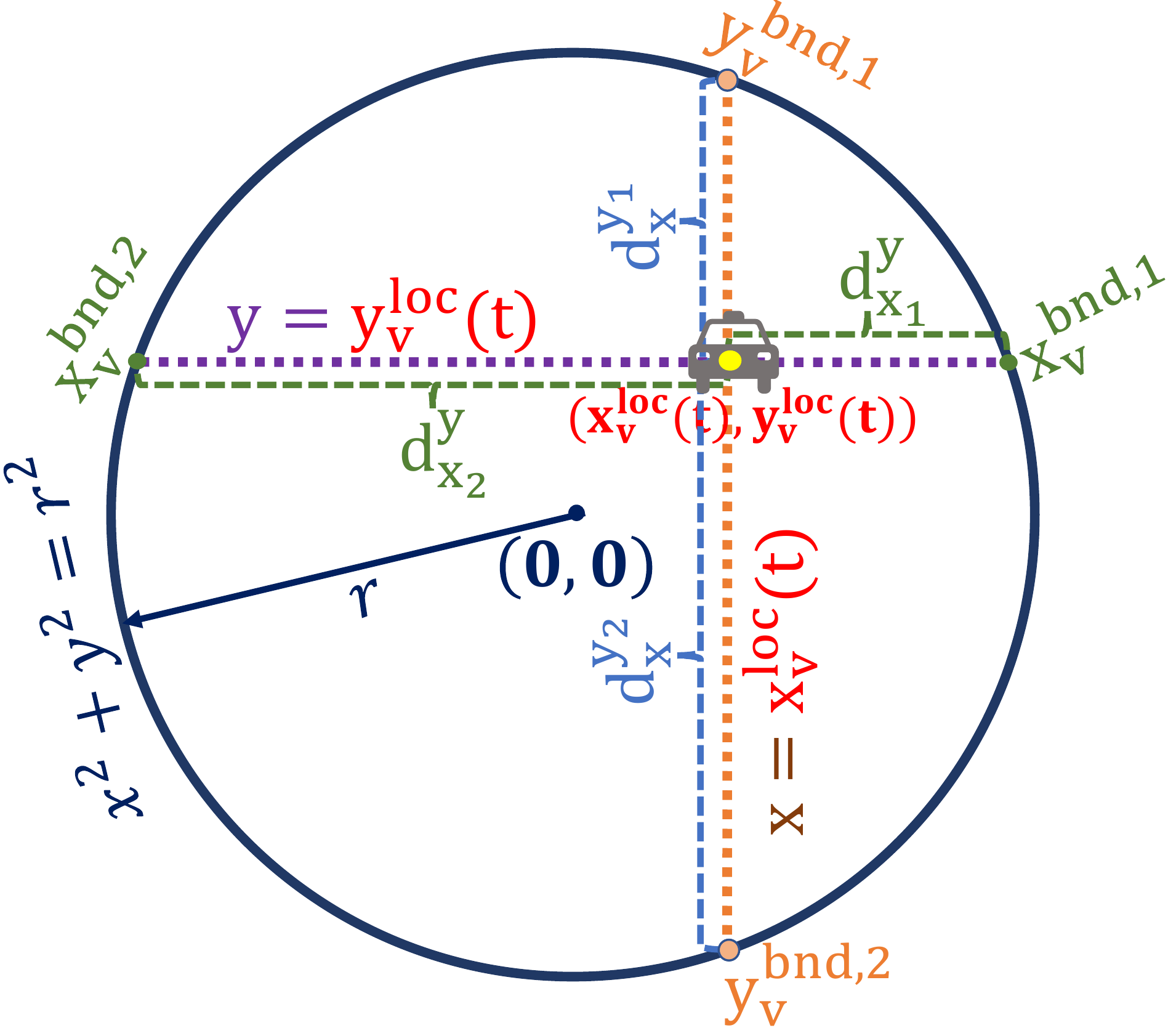} \vspace{-0.05 in}
    \caption{Finding coverage boundary points}
    \label{boundaryPoints}
\end{minipage} \hspace{-0.18in}
\begin{minipage}{0.64\textwidth}
    \centering
    \includegraphics[trim=50 700 50 50, clip, width=\textwidth, height=0.05\textheight]{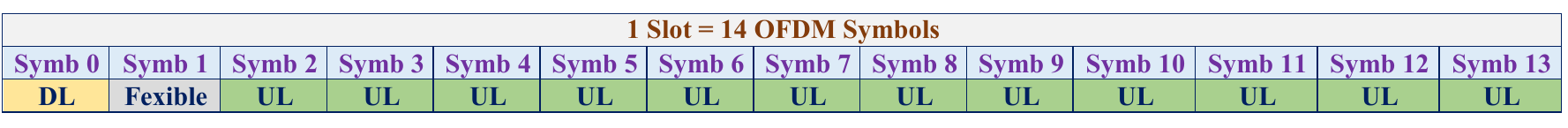} \vspace{-0.2 in}
    \caption{OFDM Symbols within a slot \cite{3gpp38_213}}
    \label{slotSybolStructure}
    \centering
    \includegraphics[trim=5 342 5 12, clip, width=1.05\textwidth, height=0.1\textheight]{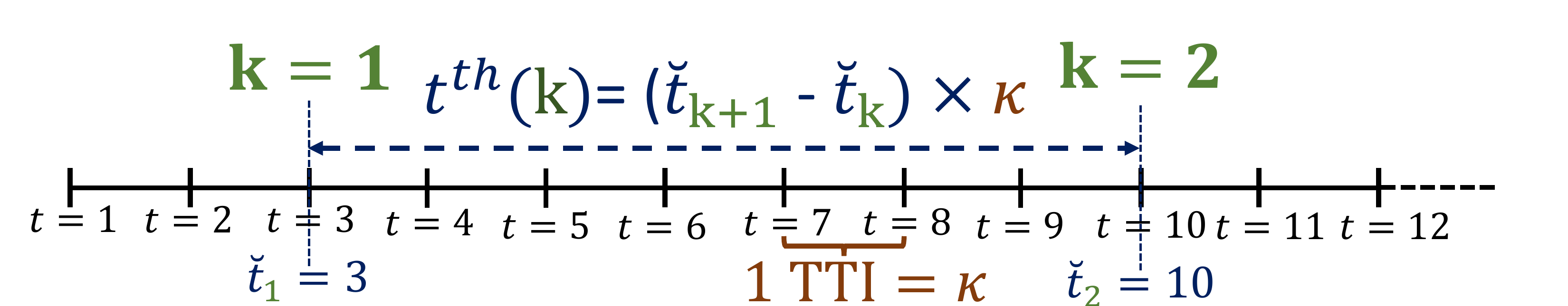} \vspace{-0.07 in}
    \caption{Time slot orchestration in the proposed VEFL}
    \label{slotStructure}
\end{minipage} \vspace{-0.25in}
\end{figure*}

\subsection{V2I Radio Access Technology Model}
\noindent
We assume that the VEN operates in TDD mode and has a fixed $\mathrm{W}$ Hz BWP for providing RAT connectivity over the universal mobile telecommunications system (UMTS) air interface (Uu interface) for the VEFL.
Note that TDD facilitates channel reciprocity and thus offers less control overhead.
Besides, we assume that all transceivers can mitigate the Doppler effect satisfactorily\footnote{Although Doppler shift is a well-known problem, if the underlying RAT mitigates it adequately, it is less critical for our proposed VEFL framework.}.
For model distribution in the downlink, as the gNB sends the same data to all selected CVs, it can use the entire spectrum and high transmission power to broadcast the model.
Therefore, similar to \cite{richeng2022communication, Yang2021,zeng2022federated}, we ignore the downlink communication.
The $\mathrm{W}$ Hz bandwidth is divided into orthogonal physical resource blocks (pRBs). 
Denote the pRB set of the allocated BWP by the set $\mathcal{Z}=\{z\}_{z=1}^Z$.
Note that due to orthogonal pRBs, there is no intra-cell interference. 
Besides, as we consider a single cell, the proposed VEN is interference-free.

The gNB considers the $5$G-NR frame structure in which the radio frame is $10$ ms long with $10$ subframes.
Within each $1$ ms subframe, there are $2^{\bar{n}}$ slots, where $\bar{n}$ is the sub-carrier spacing numerology \cite{3gpp38_213}.
The pRB allocation \textit{granularity} is the \textit{slot}, i.e., the pRBs can be allocated to different users in each slot $t$.
Each slot carries $14$ OFDM symbols in the time domain and $12$ sub-carriers in the frequency domain.
Moreover, the OFDM symbols can be configured based on the duplexing mode.
As this is uplink-heavy transmission, we consider a downlink-control uplink transmission slot format as shown in Fig. \ref{slotSybolStructure}.
As such, we consider the effective uplink data rate per slot, which will be fleshed out in what follows.

We consider a single-input-multiple-output (SIMO) case, where the gNB has $N$ antennas, and each CV has a single antenna.
Note that our framework can also be extended to other multiple-antenna communication models\footnote{However, the TDD-based massive MIMO requires rigorous channel estimation/equalization, beam management, etc., which deserve separate studies.}.
During slot $t$, denote the channel between CV $v$ and the $n^{\mathrm{th}}$ antenna of the gNB over pRB $z$ by $h_{n,v,z}(t)$. 
Then, we denote the entire channel response at VU $v$ from the gNB over PRB $z$ as 
\begin{equation}
	\begin{aligned}
		\mathbf{h}_{v,z}(t) &= \sqrt{\psi_v(t)} \varrho_v(t) \breve{\mathbf{h}}_{v,z} (t) \in \mathbb{C}^{N \times 1},
	\end{aligned}
\end{equation}
where $\sqrt{\psi_v(t)}$, $\varrho_v(t)$ and $\breve{\mathbf{h}}_{v,z}(t) = [h_{1,v,z}(t), \dots, h_{N,v,z}(t)]^T \in \mathbb{C}^{N \times 1}$ are large scale fading, log-Normal shadowing and fast fading channel responses from the $N$ antennas, respectively.
The path losses are modeled based on the urban macro (UMa) model \cite{3GPP_TR_38_901}.

To that end, denote CV $v$'s unit powered intended signal for the gNB by $s_v(t) \in \mathbb{C}$ and allocated uplink transmission power for pRB $z$ by $P_{v,z}(t)$.
Assuming receiver beamforming vector $\mathbf{g}_{v,z}(t) \in \mathbb{C}^{N \times 1}$, the effective uplink signal received at the gNB, over pRB $z$, is calculated as follows:
\begin{equation}
\label{uplink_RX_signal}
	y_{v,z}(t) = \mathrm{I}_{v,z}(t) \cdot \sqrt{P_{v,z}(t)} {\mathbf{g}_{v,z}(t)}^H \mathbf{h}_{v,z}(t) s_v(t) + \eta, 
\end{equation}
where $\mathrm{I}_{v,z}(t) \in \{0,1\}$ is an indicator function that takes value 1 when pRB $z$ is allocated to CV $v$, and $\eta \sim CN(0, \varsigma^2)$ is the circularly symmetric zero mean Gaussian distributed random variable with variance $\varsigma^2$.

Then, we calculate the received signal-to-noise ratio (SNR) at the gNB for CV $v$'s uplink transmission as follows:
\begin{equation}
\label{uplink_SINR0}
	\Gamma_{v,z}(t) = \big(\mathrm{I}_{v,z}(t) \cdot P_{v,z}(t) \big\vert {\mathbf{g}_{v,z}(t)}^H \mathbf{h}_{v,z}(t) \big\vert^2\big)/(\omega \varsigma^2),
\end{equation}
where $\omega$ is the pRB size.
The gnB can configure CV-specific CSI reference signal (RS) to estimate the channels. 
Since $5$G-NR has the flexibility of configuring CSI-RS periodically, semi-persistently or aperiodically and may also perform uplink channel information multiplexing on the physical uplink shared channel \cite{3gppTS38_214}, this work primarily focuses on the overall VEFL framework considering CSI is known at the gNB\footnote{Channel estimation delay is a part of the RAT and is less critical for our proposed VEFL framework as the server only uses the worst-case estimated channel, discussed in Section \ref{PDPC_Transformation_Solution_Section}, to determine the approximate upper bound for the uplink model offloading delay.}.
Therefore, the gNB can use maximal ratio combining receiver beamforming to get $\mathbf{g}_{v,z}(t) = \mathbf{h}_{v,z}(t)/\left\Vert \mathbf{h}_{v,z}(t) \right\Vert$, which gives the following uplink SNR over pRB $z$.
\begin{equation}
\label{uplink_SINR}
	\Gamma_{v,z}(t) = (\mathrm{I}_{v,z}(t) \cdot P_{v,z}(t) \left\Vert \mathbf{h}_{v,z}(t) \right\Vert^2)/(\omega \varsigma^2).
\end{equation}
To this end, we can calculate the achievable data rate at the gNB from CV $v$'s uplink as follows:
\begin{equation}
\label{uplink_DataRate}
	r_v(t) = \omega (1 - \upsilon) \cdot \mathrm{I}_v(t) \sum\nolimits_{z=1}^Z \mathbb{E}_{\mathbf{h}} \left[\log_2 (1+ \Gamma_{v,z}(t)) \right],
\end{equation}
where $\mathrm{I}_v(t)\in\{0,1\}$ is an indicator function that takes value $1$ when CV $v\in\mathcal{V}_t$ is scheduled for transmissions in slot $t$ and the expectation is over the channel $\mathbf{h}_{v,z}(t)$.
Besides, $\upsilon$ is the loss due to control signaling overhead. 
For our case, if the \textit{flexible} symbol in Fig. \ref{slotSybolStructure} is allocated for uplink, we set $\upsilon=1/14$. 
Moreover, if it is not assigned to uplink, we set $\upsilon=2/14$.

\subsection{Preliminaries of VEFL}

\noindent
Denote the server's global ML model by $\pmb{\omega}$. 
Denote the dataset available at CV $v$ by $\mathcal{D}_v= \{\mathbf{x}_i, y_i \}_{i=1}^{\left\vert \mathcal{D}_v \right\vert}$, where $\mathbf{x}_i$ and $y_i$ are the $i^{\mathrm{th}}$ sample feature and the corresponding label, respectively.
Therefore, during time $t$, the entire dataset can be denoted as $\mathcal{D}_t = \{\mathcal{D}_v\}_{v=1}^{V_t}$.
The edge server aims to optimize 
\begin{equation}
\begin{aligned}
    \underset{\pmb{\omega}}{\text{min }} F(\pmb{\omega}) = \sum\nolimits_{v=1}^{V_t} p_v f_v(\pmb{\omega}),
\end{aligned}
\end{equation}
where $p_v \in [0,1]$ is the linear combination weight for CV $v$ with $\sum\nolimits_{v=1}^{V_t} p_v=1$.
While the typical FedAvg set  $p_v=|\mathcal{D}_v| / |\mathcal{D}_t|$ \cite{mcmahan2017communication}, we will explore more on how to properly set these weights in what follows.
Besides, $f_v(\pmb{\omega})$ is the local empirical loss function of CV $v$.

The server distributes the global model during the VEFL global rounds $k=1,\dots, K$ as shown in Fig. \ref{slotStructure}.
Denote the slots corresponding to the global rounds by the set $\mathcal{T}_{g}=\{\breve{t}_k\}_{k=1}^K$, where $\breve{t}_1$ represents the VEN slot $t$ at which $k=1$.
For example, in Fig. \ref{slotStructure}, $\breve{t}_1=3$ and $\breve{t}_2=10$.
Note that throughout our discussions, we will use the notation $k$ and term \textit{round} to represent \textit{VEFL global round}, while the notation $t$ and term \textit{slot} will represent the \textit{discrete time slot} of the VEN.

Denote the duration between global round $k+1$ and $k$ by 
\begin{equation}
\label{deadlineThreshold}
    \mathrm{t}^{\mathrm{th}}(k) = \kappa \times \left(\breve{t}_{k+1} - \breve{t}_k \right),
\end{equation} 
where $\kappa$ is the transmission time interval (TTI).
Besides, $\mathcal{T}_g \subset \mathcal{T}$ is known to all CVs that have SLAs with the server.
Denote the available CV pool during global round $k$ by $\mathcal{V}_{\breve{t}_k}$.
Without any loss of generality, denote the global model at the server during global round $k$ by $\pmb{\omega}_k$.
The server then broadcasts $\pmb{\omega}_k$ to all CV $v \in \mathcal{V}_{\breve{t}_k}$. 
Denote the local copy of CV $v$'s model at the beginning of round $k$ by $\pmb{\omega}_{v,k}$, i.e., $\pmb{\omega}_{v,k} \gets \pmb{\omega}_k$ for all $v \in \mathcal{V}_{\breve{t}_k}$.
Upon receiving the global model, each $v \in \mathcal{V}_{\breve{t}_k}$ performs local model training to minimize the following objective function
\begin{equation}
\label{clientLossFunction}
    \underset{\pmb{\omega}} {\text{min }}
    f_v(\pmb{\omega}, \pmb{\omega}_k) = F_v(\pmb{\omega}) + (\mu/2) \Vert \pmb{\omega} - \pmb{\omega}_k \Vert^2,  
\end{equation}
where $F_v(\pmb{\omega})$ is CV $v$'s local empirical loss function on its dataset $\mathcal{D}_v$.
An $L_2$ regularization is added to $F_v(\pmb{\omega})$ to tackle heterogeneity that often arises in FL.

Each $v\in\mathcal{V}_{\breve{t}_k}$ trains its local model on its local dataset $\mathcal{D}_v $ for $l_v (k)$ iterations and obtains the following updated model
\begin{equation}
    \pmb{\omega}_{v,k+1} = \pmb{\omega}_{v,k} - \delta \sum\nolimits_{l=1}^{l_v (k)} \nabla f_v(\pmb{\omega}, \pmb{\omega}_{v,k}), 
\end{equation}
where $\delta$ is the step size.
After this local training, the CVs send $\pmb{\omega}_{v,k+1}$'s to the server.
The server then performs weight aggregations, which will be discussed in more detail in the following section, and computes the updated global model $\pmb{\omega}_{k+1}$ for the next round.
The above processes repeat until the globally trained model reaches an expected level of precision.

\subsection{Delay Calculation}
\subsubsection{Model Training Delay}
Denote CV $v$'s CPU computation cycle frequency by $\eta_v(k) \in \{\eta_v^{\mathrm{min}}, \eta_v^{\mathrm{max}}\}$.
Assuming CV $v$ requires $c_v$ CPU cycles to process per-bit data, the local computation delay for one local iteration is calculated as \cite{richeng2022communication}
\begin{equation}
	\mathrm{t}_{v,\text{itr}}^{\text{cmp}} = c_v \mathrm{D}_v/\eta_v(k),
\end{equation}
where $\mathrm{D}_v$ is the dataset size of $v$ in bits.
Then, to perform $l_v (k)$ local iterations during global round $k$, the total local computation delay of CV $v$ is 
\begin{equation}
\label{localComputeDelay}
	\mathrm{t}^{\text{cmp}}_{v}(k) = l_v (k) \cdot \mathrm{t}_{v,\text{itr}}^{\text{cmp}}.
\end{equation}

\subsubsection{Communication Delay}
Let the model parameters be of $M$ dimension, i.e., $\pmb{\omega}\in \mathbb{R}^{M}$.
Then, the required number of bits for transmitting the model parameters is calculated as
\begin{equation}
	S(\pmb{\omega}, \mathrm{FPP}) = \sum\nolimits_{m=1}^M [1+\mathrm{FPP}(m)], 
\end{equation}
where $\mathrm{FPP} (m)$ represents floating point precision. 
Note that besides $\mathrm{FPP} (m)$ bits to send the $m^{\mathrm{th}}$ element of $\pmb{\omega}$, we need $1$ additional bits to represent its sign.

This model payload size $S(\pmb{\omega}, \mathrm{FPP})$ is usually on the scale of megabits (Mbs) and takes more than one coherence time and frequency block. 
Therefore, if a CV starts model offloading from slot $t> \breve{t}_k$ after it finishes the local model training during VEFL global round $k$, the offloading delay is calculated as 
\begin{equation}
\label{comDelay}
    \mathrm{t}^{\text{tx}}_{v}\!(k) \rs=\rs \kappa \cdot \text{min} \big\{\rs T \rs: \rs \kappa (1\!-\!\upsilon) \rs \sum\nolimits_{\bar{t}=t}^T r_v(\bar{t}) \geq S (\pmb{\omega},\text{FPP}), T \in \mathbb{Z}^{+} \rs \big\}\!.\rs\rs
\end{equation}
For CV $v\in \mathcal{V}_{\breve{t}_k}$, the total delay between two VEFL global rounds is, thus, calculated as follows:
\begin{equation}
\label{totalDealy}
\begin{aligned}
    \mathrm{t}_v (k) &= \mathrm{t}^{\text{cmp}}_{v}(k) +  \mathrm{t}^{\text{tx}}_{v}(k).
\end{aligned}
\end{equation}

\subsection{Energy Consumption}
\subsubsection{Local Model Training Energy Consumption}
We calculate the energy consumption of CV $v\in \mathcal{V}_{\breve{t}_k}$ for its local model training during round $k$ as 
\begin{equation}
    \mathrm{e}^{\text{cmp}}_{v}(k) = l_v (k) \cdot (\zeta/2) c_v \mathrm{D}_v {\eta_v(k)}^2,  
\end{equation}
where $\zeta/2$ is the effective capacitance of CV $v$'s CPU chip.

\subsubsection{Communication Energy Consumption}
We calculate the uplink transmission energy consumption of the client $v\in\mathcal{V}_{\breve{t}_k}$ during round $k$ as
\begin{equation}
    \mathrm{e}_{v}^{\text{tx}}(k) = \sum\nolimits_{\bar{t}=t}^{T} \sum\nolimits_{z=1}^Z P_{v,z}(\bar{t}),    
\end{equation}
where $T$ is calculated is (\ref{comDelay}).

Therefore, we calculate the total energy consumption associated with client $v$'s local computation and uplink communication as follows:
\begin{equation}
\label{totalEnergyConsumption}
    \mathrm{e}^{\text{tot}}_{v} (k) = \mathrm{e}^{\text{cmp}}_{v}(k) + \mathrm{e}_{v}^{\text{tx}}(k).
\end{equation}

\section{VEFL Convergence Analysis and Problem Formulation}
\label{convergenceProbFormSection}

\subsection{Convergence Analysis}
\noindent
We make the following standard assumptions \cite{MLSYS2020_38af8613,9252927}.

\textit{Assumption $1$ ($L$-Lipschitz Gradient)}: For all $v\in \mathcal{V}_{\breve{t}_k}$ in all $k$, $F_v(\pmb{\omega})$ is $L$-Lipschitz gradient for any two parameter vectors $\pmb{\omega}$ and $\pmb{\omega}'$, i.e., $\left \Vert \nabla F_v(\pmb{\omega}) - \nabla F_v(\pmb{\omega}') \right \Vert \leq L \left\Vert \pmb{\omega} - \pmb{\omega}' \right \Vert $.

\textit{Assumption $2$ ($B$-Dissimilar Gradients)}: The local  gradient at the CVs are at most $B$-dissimilar from the global gradient $\nabla f (\pmb{\omega})$, i.e.,  $\left\Vert \nabla F_v(\pmb{\omega}) \right\Vert \leq B \left\Vert \nabla f(\pmb{\omega}) \right\Vert$, for all $v$.

\textit{Assumption $3$ ($\sigma$-Bounded Hessian)}: The smallest eigen value of the Hessian matrix is $-\sigma$ for all CVs, i.e., $\nabla^2 F_v \succeq -\sigma \mathbf{I}$. This also implies that the $f_v(\pmb{\omega}, \pmb{\omega}_k)$  in (\ref{clientLossFunction}) is $\mu'=\mu-\sigma$ strongly convex.

\textit{Assumption $4$ ($\gamma$-Inexact Local Solvers)}: Local update of the CV $v$ results in $\gamma$-inexact solution $\pmb{\omega}_{v,k+1}$ of (\ref{clientLossFunction}) in all global round $k$. 
In other words, the local update of CV $v$ yields $\left\Vert \nabla f_v( \pmb{\omega}_{v,k+1}, \pmb{\omega}_k) \right\Vert \leq \gamma \left\Vert \nabla f_v(\pmb{\omega}_k, \pmb{\omega}_k) \right\Vert$, where $\gamma \in [0,1]$. 
Note that $\gamma=0$ means solving (\ref{clientLossFunction}) optimally, while an increased value indicates how the updated model differs from the exact solution.

Additionally, we consider full-batch gradient descent, i.e., each CV performs its model training in its entire dataset, which can be extended to stochastic gradient descent.

Denote the successful trained local model reception from the client $v$ during global round $k$ by
\begin{equation}
\label{successOffloading}
\begin{aligned}
    \rs\rs\rs \mathbf{1}\rs\left(\rs\mathrm{t}_v (k) \rs \leq \rs \mathrm{t}^{\mathrm{th}}(k)|d_v(k) \rs \leq \rs \mathrm{r} \rs\right) 
    &\rs= \! \rs \begin{cases}
            \rs 1, \rs\rs\rs & \rs \text{with probability } p^{\mathrm{suc}}_{v} (k), \\ 
            \rs 0, \rs\rs\rs & \rs \text{with probability } \rs (\! 1 \! - \!p^{\mathrm{suc}}_{v} (k)\!),
            \end{cases}\rs\rs. \rs\rs\rs\rs\rs\rs\rs
\end{aligned}    
\end{equation}
where $d_v(k)$ is the distance of $v$ after performing $l_v (k)$ iterations from the gNB.

We consider two cases for device participation, namely, the FDPC and the PDPC. 
While all CVs with SLAs participate in model training for the former scenario, only a subset of these CVs participates in the latter.
In both cases, for the model aggregation, the server can take a client's locally trained model $\pmb{\omega}_{v,k+1}$ only when it successfully receives it within the deadline threshold $\mathrm{t}^{\mathrm{th}}(k)$.
As such, for FDPC, we express the server's aggregation rule as
\begin{equation}
\label{aggregationRule}
\begin{aligned}
    \rs\pmb{\omega}_{k+1} &\rs = \pmb{\omega}_k \rs + \rs \sum_{v=1}^{V_{\breve{t}_k}} \rs  \Big[\frac{p_v \rs \cdot \rs \mathbf{1}\!\left(\!\mathrm{t}_v (k) \rs \leq \rs \mathrm{t}^{\mathrm{th}}(k)|d_v(k) \rs \leq \rs \mathrm{r} \!\right)}{p^{\mathrm{suc}}_{v} (k)} \rs \left(\! \pmb{\omega}_{v,k+1} \rs - \pmb{\omega}_k \!\right) \! \rs \Big]\!,\rs \rs \rs \rs \rs \rs
\end{aligned}
\end{equation}
where $p_v \in (0,1]$ and $\sum\nolimits_{v=1}^{V_{\breve{t}_k}} p_v = 1$. 
Moreover, $p_v = \bar{p}_v/\sum_{v=1}^{V_{\breve{t}_k}} \bar{p}_v$, where $\bar{p}_v$ is calculated as follows:
\begin{equation}
\label{p_V_Mobility}
\begin{aligned} 
    \rs\bar{p}_v \rs &= \rs (1-\lambda)[ \mathrm{D}_v/ \sum\nolimits_{v=1}^{V_{\breve{t}_k}} \mathrm{D}_v] + \lambda [ \mathrm{t}_{v}^{\mathrm{soj}}(\breve{t}_k)/\sum\nolimits_{v=1}^{V_{\breve{t}_k}} \mathrm{t}_{v}^{\mathrm{soj}}(\breve{t}_k)],\rs \rs \rs \rs 
\end{aligned} 
\end{equation}
where $\lambda \in [0,1]$ is a parameter that balances the weight of the sojourn period and dataset size.

Note that the averaging step in (\ref{aggregationRule}) is unbiased, i.e., $\mathbb{E}\left[\pmb{\omega}_{k+1} \right] = \sum\nolimits_{v=1}^{V_{\breve{t}_k}} p_v \cdot \pmb{\omega}_{v,k+1}$. 
Besides, we adopt this aggregation weight $p_v$ inspired by the complementary cumulative distribution function (CCDF) of the CVs' expected sojourn periods $\mathrm{t}_{v}^{\mathrm{soj}} (\breve{t}_k)$'s.
Particularly, at high speed, $\mathrm{t}_{v}^{\mathrm{soj}} (\breve{t}_k)$ can be very small, as shown in Fig. \ref{ccdf_sojourn_period}. 
For example, at a maximum velocity $\mathrm{u}^{\mathrm{max}}$ of $26.82$ meter/second (m/s), about $28\%$ CVs have an expected sojourn period smaller than $5$ seconds.
Besides, about $13\%$ CVs have smaller than $2.5$ seconds sojourn periods.
As such, the aggregation rule shall also weigh this short sojourn period since these CVs are likely to execute only a few local iterations compared to the other CVs with relatively higher sojourn periods.

\begin{figure}[!t] \vspace{-0.05in}
    \centering
	\includegraphics[trim={15 15 15 15}, clip, width=0.47\textwidth, height=0.2\textheight]{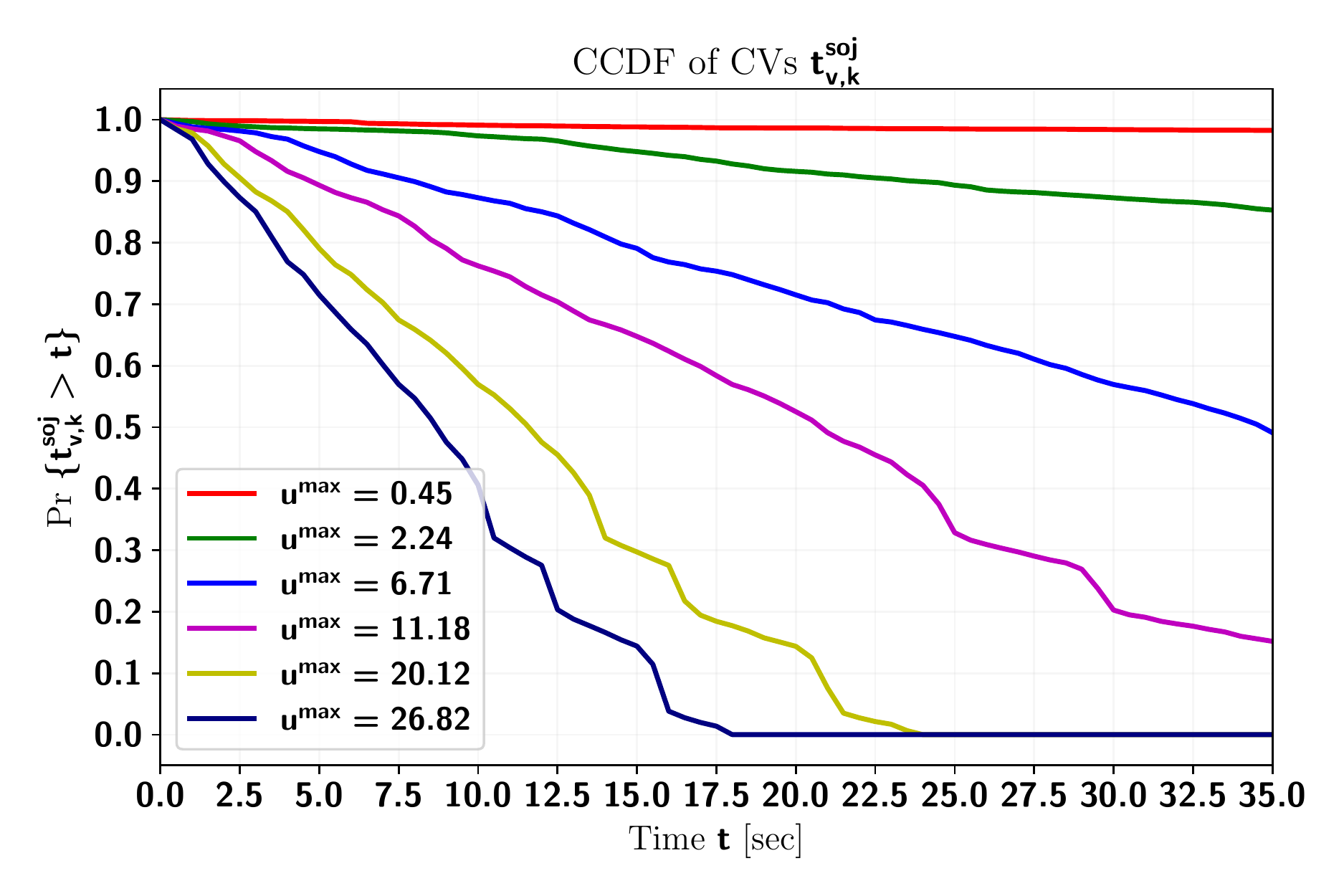} \vspace{-0.05 in}
    \caption{CCDF of CVs expected sojourn period}
    \label{ccdf_sojourn_period}
\end{figure}

One popular way to combat the straggler effect is to aggregate model parameters upon the reception of the first $|\mathcal{C}_k|$, where $\mathcal{C}_k \subseteq \mathcal{V}_{\breve{t}_k}$ \cite{Li2020On}.
In federated edge learning (FEEL), RAT resource is limited. 
Besides, the CVs are required to perform their own computational tasks, while the server is required to pay a fee if it selects a CV for model training.
As such, it is practical to select the subset $\mathcal{C}_k$ at the beginning of VEFL round $k$.
In this pragmatic case, we thus assume that the server uniformly selects $|\mathcal{C}_k|$ CVs out of the original $V_{\breve{t}_k}$ CVs without replacement in each VEFL round $k$.
Moreover, we stress that the server decides how many total CVs it can select based on its limited monetary and bandwidth budgets. 
In other words, we leave this $|\mathcal{C}_k|$ as a design parameter that the server determines contingent on the resource-constraint circumstances.
Denote the event that CV $v$ is selected during VEFL round $k$ by 
\begin{equation}
\label{cvSelectionIndicator}
\begin{aligned}
    & \mathbf{1}\left(v \in \mathcal{C}_k\right)  = \begin{cases}
            1, & \text{with probability $q_v(k)$ } \\
            0, & \text{with probability  } \left(1-q_{v} (k)\right),
            \end{cases}.
\end{aligned}    
\end{equation}

Considering the above factors, the server chooses the following aggregation rule for PDPC.
\begin{align}
\label{aggregationRulePartialNORep}
    \pmb{\omega}_{k+1} & \rs= \pmb{\omega}_k + \rs \sum\nolimits_{v=1}^{V_{\breve{t}_k}} \Big[p_v \cdot \frac{\mathbf{1}\left(v \in \mathcal{C}_k, \mathrm{t}_v (k) \leq \mathrm{t}^{\mathrm{th}}(k)|d_v(k) \leq \mathrm{r} \right)} {q_v(k) p^{\mathrm{suc}}_{v} (k)} \rs \times \rs \rs \rs \nonumber \\
    & \qquad \qquad \qquad \qquad \qquad \qquad \quad \left( \pmb{\omega}_{v,k+1} - \pmb{\omega}_k \right) \Big].
\end{align}
When all selected CVs have the same $p_v$, the averaging step in (\ref{aggregationRulePartialNORep}) is unbiased, i.e., $\mathbb{E}\left[\pmb{\omega}_{k+1} \right] = \sum\nolimits_{v=1}^{V_{\breve{t}_k}} p_v \cdot \pmb{\omega}_{v,k+1}$.
The VEFL algorithm is summarized in Algorithm \ref{vEFLalgorithm}.
The convergence analysis for PDPC is derived in Theorem \ref{theorem2}.

\begin{Theorem}
\label{theorem2}
Using our assumptions and aggregation rule (\ref{aggregationRulePartialNORep}), the expected loss decrease after one VEFL round is
\begin{equation}
\label{expeclosLiters_NoReplace}
\begin{aligned}
    &\mathbb{E} \left[f(\pmb{\omega}_{k+1}) \right] - f(\pmb{\omega}_k)  \leq (B \left(1 + \gamma \right)/\mu' ) \times \hfill \\
    & \qquad \Big(1 + \frac{BL(1+\gamma)}{2 \mu'} \sum\nolimits_{v=1}^{V_{\breve{t}_k}} \frac{p_v} {q_v(k) p^{\mathrm{suc}}_{v} (k)} \Big)  \left \Vert \nabla f(\pmb{\omega}_k) \right\Vert^2.
\end{aligned}
\end{equation}
\end{Theorem}
\begin{proof}
The proof is left in Appendix \ref{proofTheorem2}.
\end{proof}

\begin{Remark}
\label{remarkConvergence}
From the convergence analysis in (\ref{expeclosLiters_NoReplace}), we observe that $p_v$, $q_v(k)$, $p_v^{\mathrm{suc}}(k)$ and $\left\Vert \nabla f(\pmb{\omega}_k) \right\Vert$ may vary in each VEFL round $k$.
Besides, it is generally impossible to find a closed-form expression for $\left\Vert \nabla f(\pmb{\omega}_k) \right\Vert$ by relating it to the learning parameters \cite{salehi2021Federated,wang2019Adaptive}.
As such, in each VEFL round, we aim to sub-optimally minimize the loss in (\ref{expeclosLiters_NoReplace}) by concentrating on $\sum\nolimits_{v=1}^{V_{\breve{t}_k}} p_v/[q_v(k) p_v^{\mathrm{suc}}(k)]$.
\end{Remark}

\begin{Remark}
Note that FDPC is a special case of PDPC, where all CVs participate in the model training. 
Moreover, the convergence analysis and problem formulation are similar. 
\end{Remark}

\subsection{Problem Formulation}
\noindent
In our proposed VEFL, we explicitly consider CVs mobility, system heterogeneity---in terms of CPU frequency and energy budget---across CVs, intermittent Uu links, and the server's limited monetary and radio budgets.
Denote the fixed monetary budget of the server for global round $k$ by $\Xi(k)$ unit. 
Denote the cost for each unit of energy by $\phi_v$ unit/Joule.
The energy consumption of a CV depends on its number of local iterations $l_v (k)$, CPU frequency $\eta_v(k)$ and total uplink transmission power to offload the trained model $\pmb{\omega}_{v,k+1}$.
Particularly, as part of the SLA, each CV first reports its minimum CPU frequency $\eta_v^{\mathrm{min}}$, maximum CPU frequency $\eta_v^{\mathrm{max}}$, maximum transmission power $P_v^{\text{max}}$, energy budget $\mathrm{e}^{\mathrm{bud}}_{v}(k)$, dataset size $\mathrm{D}_v$ and charging policy $\xi_v (l_v(k), \eta_v(k))$ to the server.
More specifically, the CVs have the following charging policy.
\begin{equation}
\label{cost_per_iter}
\begin{aligned}
    \xi_v(l_v (k),\eta_v(k)) &= \mathrm{e}_v^{\mathrm{tot}} (k) \phi_v + \bar{\phi}_v,
\end{aligned}
\end{equation}
where $\mathrm{e}_v^{\mathrm{tot}} (k)$ is calculated in (\ref{totalEnergyConsumption}).
As such, $\mathrm{e}_v^{\mathrm{tot}} (k) \phi_v$ is the CV's expected operational cost, and $\bar{\phi}_v$ is a constant term to ensure its gain upon participating in model training.

\begin{algorithm}[t!]
\fontsize{9}{8}\selectfont
\SetAlgoLined 
\DontPrintSemicolon
\KwIn{Total global round $K$}
\For{$k=1$ to $K$}{
    Broadcast $\pmb{\omega}_k$ to the client CV set $\mathcal{Y}_k$\; \tcc*{$\mathcal{Y}_k=\mathcal{V}_{\breve{t}_k}$ and $\mathcal{Y}_k=\mathcal{C}_k$, respectively for FDPC and PDPC}
    \For {$v \in \mathcal{Y}_k$ in parallel} {
        Train the received global model for $l_v (k)$ local iterations to get the locally trained model $\pmb{\omega}_{v,k+1}$ \;
        Offload $\pmb{\omega}_{v,k+1}$ back to the server using the allocated TTIs and pRBs by the gNB \;}
    Server aggregates client CVs models to get the updated global model $\pmb{\omega}_{k+1}$ using (\ref{aggregationRule}) and (\ref{aggregationRulePartialNORep}), respectively, for FDPC and PDPC \; 
}
\KwOut{Global ML model $\pmb{\omega}_K$}
\caption{Vehicular Edge Federated Learning}
\label{vEFLalgorithm}
\end{algorithm}

To this end, denote 
$\pmb{1}(k) \rs = \rs \big[\mathbf{1}\left(1 \rs \in \rs \mathcal{C}_k\right), \dots, \mathbf{1}\left(V_{\breve{t}_k} \rs \in \rs \mathcal{C}_k\right) \big]^T \rs \rs \in \rs \mathbb{R}^{V_{\breve{t}_k}} \!$,
$\pmb{l} (k) = [l_1 (k), \dots, l_{V_{\breve{t}_k}} (k)]^T \in \mathbb{R}^{V_{\breve{t}_k}}$,
$\pmb{\eta} (k) = \big[\eta_1 (k), \dots, \eta_{V_{\breve{t}_k}} (k) \big]^T \in \mathbb{R}^{V_{\breve{t}_k}}$, 
$\pmb{\mathrm{I}}(t)=[\mathrm{I}_1(t), \dots, \mathrm{I}_{V_{\breve{t}_k}}(t)]^T \in \mathbb{R}^{V_{\breve{t}_k}}$,
$\breve{\pmb{\mathrm{I}}}_v(t) = [\mathrm{I}_{v,1}(t), \dots, \mathrm{I}_{v,Z}(t)]^T \in \mathbb{R}^{Z}$,
$\breve{\pmb{\mathrm{I}}}(t) = [\breve{\pmb{\mathrm{I}}}_1(t), \dots, \breve{\pmb{\mathrm{I}}}_{V_{\breve{t}_k}}(t)]^T \in \mathbb{R}^{V_{\breve{t}_k} \times Z}$, 
$\mathbf{p}_v(t)=[P_{v,1}(t), \dots, P_{v,Z}(t)]^T \in \mathbb{R}^{Z}$ and 
$\mathbf{p}(t)=[\mathbf{p}_1(t), \dots, \mathbf{p}_{V_{\breve{t}_k}}(t)]^T \in \mathbb{R}^{V_{\breve{t}_k} \times Z}$.
Based on Theorem \ref{theorem2} and Remark \ref{remarkConvergence}, in each VEFL round $k$, we want to sub-optimally minimize the loss by minimizing the relaxed objective $\sum\nolimits_{v=1}^{V_{\breve{t}_k}} p_v/[q_v(k) p_v^{\mathrm{suc}}(k)]$.
Note that $q_v(k)$ in the denominator is essentially the expectation of CV selection $\mathbf{1}(v \in \mathcal{C}_k)$, which is a VEFL parameter. 
Besides, $p_v^{\mathrm{suc}}$ is the CV's local model reception successful probability that depends on the RAT configurations.
Instead of solving $\sum\nolimits_{v=1}^{V_{\breve{t}_k}} p_v/[q_v(k) p_v^{\mathrm{suc}}(k)]$ directly, we pose it as relaxed a linear objective $\sum_{v=1}^{V_{\breve{t}_k}} \mathbf{1} (v \in \mathcal{C}_k) \cdot p_v^{\mathrm{suc}}(k)$ that we want to maximize.
Note that maximization of total user participation in a similar fashion is also common in the literature \cite{Ma2020Scheduling,Hamdi2022,Wang2022Federated}. 
Besides, this practical objective function also works well in simulation.
Recall that CSI is not fixed, the server has a limited monetary budget, and the gNB has a limited BWP. 
On the one hand, knowing $p_v^{\mathrm{suc}}(k)$ at the beginning of a VEFL round is impossible. 
On the other hand, the CV selection indicator functions appear in many of our constraints in the following.
Therefore, using the binary CV selection indicator function $\mathbf{1}(v \in \mathcal{C}_k)$ in the objective function serves two primary purposes: congruent problem decomposition in the sequel and elimination of optimization parameter $q_v(k)$.
As such, we want to solve the following relaxed problem.
\begin{subequations}
\label{partialDevProblem}
\begin{align}
    &\underset{\pmb{1} (k), \pmb{l} (k), \pmb{\eta} (k), \pmb{\mathrm{I}}(t), \breve{\pmb{\mathrm{I}}}(t), \mathbf{p}(t) } {\text{maximize }} ~ \sum\nolimits_{v=1}^{V_{\breve{t}_k}} \mathbf{1}(v \in \mathcal{C}_k) \cdot p_{v}^{\mathrm{suc}}(k), \tag{\ref{partialDevProblem}} \\
    &\text{s.t.} ~ \label{PD_cons1} (C_1) \quad \sum\nolimits_{v=1}^{V_{\breve{t}_k}} \mathbf{1}(v \in \mathcal{C}_k) = \left\vert \mathcal{C}_k\right\vert,\\
    & \label{PD_cons2} (C_2) \quad \mathbf{1}(v \in \mathcal{C}_k) \mathrm{t}_v(k) \leq \mathrm{t}^{\mathrm{th}}(k), ~ \forall v, \forall k, \\
    & \label{PD_cons3} (C_3) \quad \mathbf{1}(v \in \mathcal{C}_k) \mathrm{t}_v(k) \leq \mathbf{1}(v \in \mathcal{C}_k) \mathrm{t}_{v}^{\mathrm{soj}}(\breve{t}_k), ~ \forall v, \forall k,\\
    & \label{PD_cons4} (C_4) \quad  \mathbf{1}(v \in \mathcal{C}_k) \mathrm{e}_{v}^{\mathrm{tot}}(k) \leq \mathbf{1}(v \in \mathcal{C}_k) \mathrm{e}^{\mathrm{bud}}_{v}(k), ~\forall v, \forall k,\rs\\
    & \label{PD_cons5} (C_5) \quad \sum\nolimits_{v=1}^{V_{\breve{t}_k}}\mathbf{1}(v \in \mathcal{C}_k) \xi(l_v (k), \eta_v(k)) \leq \Xi(k), ~\forall k, \\
    & \label{PD_cons6} (C_6) \quad \mathbf{1}(v \in \mathcal{C}_k) \cdot l_v (k) \geq l_k^{\mathrm{des}}, ~ l_v (k) \in \mathbb{Z}^{+},~ \forall k, \forall v\\
    & \label{PD_cons7} (C_7) \quad \sum \nolimits_{v \in \mathcal{C}_k} \mathrm{I}_{v,z}(t) = 1,  \forall t, \forall v \in \mathcal{C}_k, \forall z, \\
    & \label{PD_cons8} (C_8) \quad \sum\nolimits_{z=1}^{Z} \sum\nolimits_{v \in \mathcal{C}_k} \mathrm{I}_{v,z}(t) = Z,  \forall t, \forall v \in \mathcal{C}_k, \forall z, \\
    &  \label{PD_cons9} (C_9) \quad  \sum\nolimits_{v \in \mathcal{C}_k} \mathrm{I}_v(t) \leq Z,  \forall t, \forall v \in \mathcal{C}_k, \\
    & \label{PD_cons10} (C_{10}) \quad  \{\mathrm{I}_v(t), \mathrm{I}_{v,z}(t)\} \in \{0,1\},  \forall t, \forall v\in \mathcal{C}_k, \forall z, \\
    & \label{PD_cons11} (C_{11}) \quad 0 \leq \mathrm{I}_{v,z}(t) P_{v,z}(t) \leq P_v^{\mathrm{max}}, ~~ \forall t, \forall z, \forall v \in \mathcal{C}_k\\
    & \label{PD_cons12} (C_{12}) \quad  \sum\nolimits_{z=1}^Z \mathrm{I}_{v,z}(t) P_{v,z}(t) \leq P_v^{\mathrm{max}},  ~~ \forall t, \forall z, v \in \mathcal{C}_k,
\end{align}
\end{subequations}
where constraint (\ref{PD_cons1}) restricts the server to choose $|\mathcal{C}_k|$ CVs in VEFL global round $k$.
The second constraint in $C_2$ restricts the total delay to be within the deadline budget. 
Constraint $C_3$ in (\ref{PD_cons3}) ensures that the total computation and transmission time is within the expected worst-case sojourn period.
Besides, constraint (\ref{PD_cons4}) is taken to restrict the total energy consumption to be within the total energy budgets of the CVs.
$C_5$ in (\ref{PD_cons5}) ensures that the server's total expense cannot exceed its budget $\Xi(k)$.
Furthermore, constraint $C_6$ is for the non-negative integer values of the local iteration numbers $l_v (k)$'s.
Constraint $C_7$ in (\ref{PD_cons7}) ensures that one pRB is allocated to only one CV and constraint $C_8$ in (\ref{PD_cons8}) ensures all pRBs are allocated to all CVs.
$C_9$ ensures that the total number of scheduled CVs in a slot does not exceed the total available pRBs.
Constraint $C_{10}$ is for the binary decision variables.
Finally, constraints $C_{11}$ and $C_{12}$ ensures that the allocated power over pRB $z$ and over all $Z$ pRBs are less than the maximum possible transmission power of the CV, respectively.

\begin{Remark} 
\label{Remark_Problem_Transformation}
The VEFL parameters, i.e., $\mathbf{1}(k)$, $\pmb{l}(k)$ and $\pmb{\eta}(k)$, and the RAT parameters, i.e., $\pmb{\mathrm{I}}(t)$, $\breve{\pmb{\mathrm{I}}}(t)$ and $\mathbf{p}(t)$, are coupled in problem (\ref{partialDevProblem}). 
One of the main challenges of this problem is that the CVs must know their $l_v(k)$'s and $\eta_v(k)$'s when the VEFL round $k$ starts. 
However, to solve (\ref{partialDevProblem}) for these values, the server must also know $\mathrm{t}_v^{\mathrm{tx}}(k)$'s and $\mathrm{e}_v^{\mathrm{tx}}(k)$'s, which is impossible at the beginning of the VEFL round because the Uu links change in each slot.
Moreover, the server also needs to know the charging policies $
\xi_v(l_v(k), \eta_v(k))$'s, which depend on the total energy consumption $\mathrm{e}_{v}^{\mathrm{tot}}(k)$'s.
As such, in the following, we decompose this joint problem into a VEFL parameter optimization sub-problem and a RAT parameter optimization sub-problem. 
More specifically, the VEFL parameter optimization sub-problem uses the worst-case estimation of $\mathrm{t}_v^{\mathrm{tx}}(k)$'s and $\mathrm{e}_v^{\mathrm{tx}}(k)$'s at the beginning of each VEFL round $k$. 
On the other hand, the server solves the RAT parameter optimization sub-problem when it requests for the trained models of the CVs.
\end{Remark}

\section{Problem Transformations and Solutions}
\label{PDPC_Transformation_Solution_Section}

\noindent
This section demonstrates our problem decomposition, followed by the solutions.
First, we illustrate the problem decomposition steps for the two sub-problems discussed in Remark \ref{Remark_Problem_Transformation} through Fig. \ref{probDecomp}.
As the original VEFL parameter optimization sub-problem (c.f. (\ref{iterPlusCpuFreqPD})) is challenging, we perform linear programming (LP) relaxation and transform it into a convex optimization problem (c.f. (\ref{localIterGivenEtaPDCase2})) that we solve to obtain the optimized $\mathbf{1}(k)$, $\pmb{l}(k)$ and $\pmb{\eta}(k)$ values.
Moreover, the server monitors the payload buffers of the CVs and optimizes the RAT parameters to maximize the probability of successful reception $p_v^{\mathrm{suc}}(k)$'s of the CVs' trained models (c.f. (\ref{fdRATProbOriginal})). 
However, since the CSI changes in each slot, the $p_v^{\mathrm{suc}}(k)$'s can only be determined after completing the VEFL round.  
As such, we convert the problem into a per-slot utility maximization problem (c.f. (\ref{fdRATProbTransform})). 
To that end, the earliest deadline first (EDF) based scheduling \cite{buttazzo2011hard} is adopted, followed by the resource allocation optimization. 
When more than one CV is scheduled, the RAT parameter optimization problem (c.f. (\ref{fdRATProbTransform1})) is non-convex. 
Thus, we perform LP relaxation, use the DC techniques and transform it into a relaxed convex optimization problem (c.f. (\ref{fdRATequivalProb2})) that we solve to optimize the pRB and power allocations.

\begin{figure}[!t] \vspace{-0.05 in}
    \centering
    \includegraphics[width=0.42\textwidth, height=0.21\textheight]{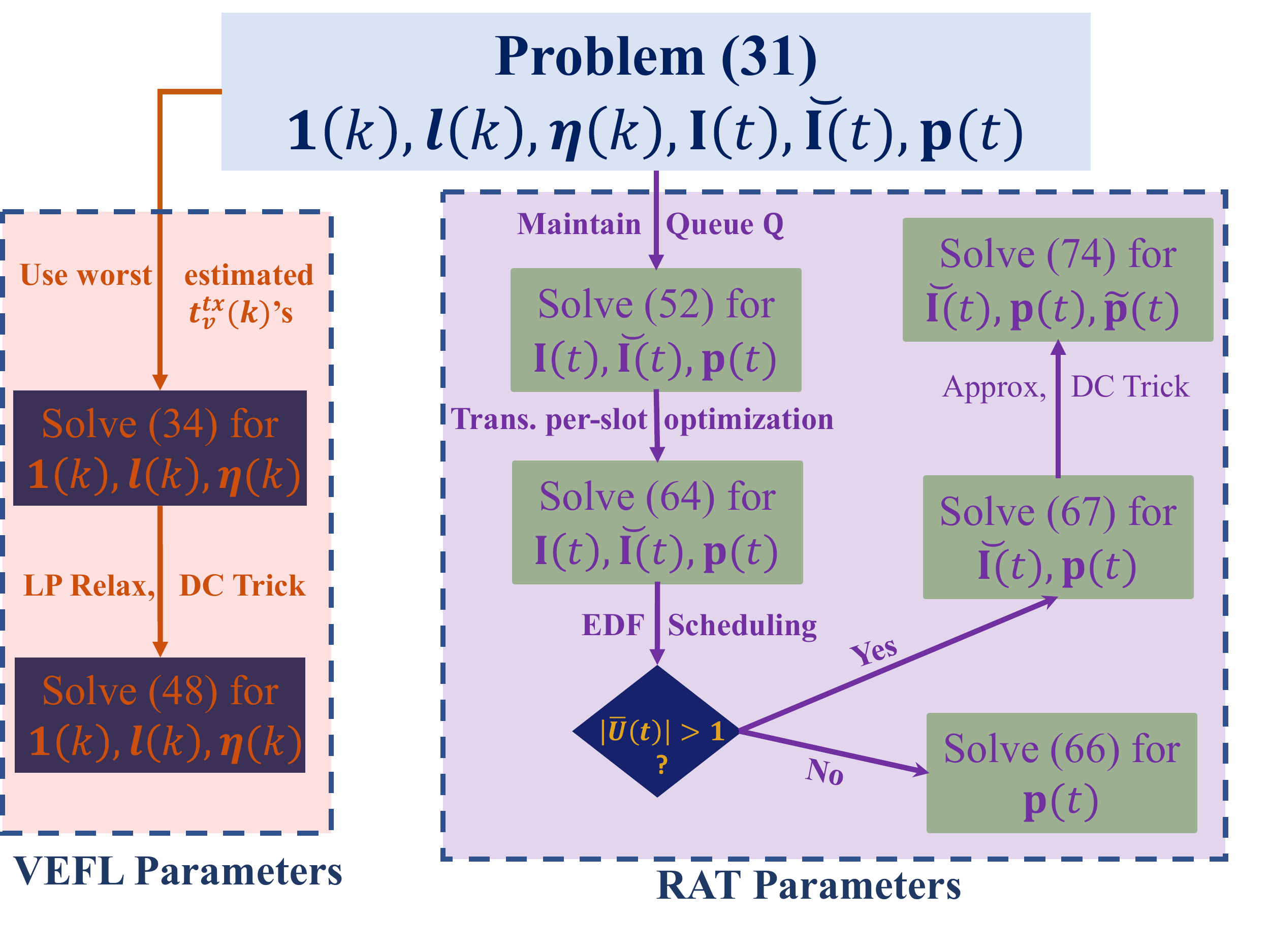} \vspace{-0.1in}
    \caption{Problem decomposition steps}
    \label{probDecomp}
\end{figure}

The server considers the following upper-bounded transmission energy consumption to calculate its per round cost at the beginning of each VEFL round $k$.
\begin{equation}
\begin{aligned}
    \hat{\xi}_v(l_v (k), \eta_v(k)) &= \big[ \mathrm{e}_v^{\mathrm{cmp}}(k) + \kappa P_v^{\mathrm{max}} \cdot \bar{\mathrm{t}}_v^{\mathrm{tx}} \big] \phi_v + \bar{\phi}_v,
\end{aligned}
\end{equation}
where $\bar{\mathrm{t}}_v^{\mathrm{tx}}$ is the worst case total number of required TTIs, which is calculated in (\ref{worstCaseTxTTIEstimate}).
\begin{equation}
\label{worstCaseTxTTIEstimate}
    \bar{\mathrm{t}}_v^{\mathrm{tx}} \rs = \rs \left\lceil S(\pmb{\omega}, \mathrm{FPP}) / (\kappa (1-\upsilon) \omega \cdot  \mathbb{E} \left[\log_2 (1 + \bar{\Gamma}_v) \right] \rs \times \! \tilde{z}(k)) \right\rceil \rs , \rs \rs \rs
\end{equation}
where $\lceil \cdot \rceil$ is the $ceil$ operator and $\bar{\Gamma}_v = \mathrm{min}\{\{\Gamma_{v,z}\}_{z=1}^Z\}$ is the worst-case SNR over the worst channel\footnote{In practical $3$GPP networks, one may use the historical CSI information to get an estimated worst-case channel condition.}. 
Besides, $\tilde{z}(k) = Z/V_{\breve{t}_k}$ if $Z<V_{\breve{t}_k}$ and $\tilde{z}(k) = 1$ otherwise.

\subsection{PDPC: VEFL Parameter Optimizations}

\noindent
Given the upper-bounded communication delay and energy consumption, the server wishes to jointly optimize $\pmb{1}(k)$, $\pmb{l}(k)$ and $\pmb{\eta}(k)$. 
Intuitively, the trained model $\pmb{\omega}_{v,k}$ is expected to deliver better performance if CV $v$ trains it for more local iterations.
Besides, more training samples usually help enhance a trained model's test performance.
However, since the CVs are not stationary, short sojourn periods can also play a critical role. 
As such, we want to configure $\mathbf{1}(k)$, $\pmb{l}(k)$ and $\pmb{\eta}(k)$ under the joint consideration of the sojourn periods and the dataset sizes.
Considering the worst-case communication delay, energy consumption and cost, we optimize the weighted sum of CVs' local training iterations, which is given as
\begin{subequations}
\label{iterPlusCpuFreqPD}
\begin{align}
    &\underset{\pmb{1} (k), \pmb{l} (k), \pmb{\eta} (k)} {\text{maximize }} \quad  \sum\nolimits_{v=1}^{V_{\breve{t}_k}} \mathbf{1}(v \in \mathcal{C}_k) l_v (k) \Theta_v , \tag{\ref{iterPlusCpuFreqPD}} \\
    &\text{s.t.} \quad  C_1, C_6, \\
    &\label{PD_VEFL_cons2} (\tilde{C}_2) ~ \mathbf{1}(v \in \mathcal{C}_k) \mathrm{t}_v^{\mathrm{cmp}}(k) + \nonumber\\
    &\quad \kappa \mathbf{1}(v \in \mathcal{C}_k) \bar{\mathrm{t}}_v^{\text{tx}} \leq  \mathbf{1}(v \in \mathcal{C}_k) \text{min}\{\mathrm{t}^{\mathrm{th}}(k), \mathrm{t}_{v}^{\mathrm{soj}}(\breve{t}_k)\}  ~ \forall v, k, \\
    &\label{PD_VEFL_cons3} (\tilde{C}_3)  \quad \mathbf{1}(v \in \mathcal{C}_k) \mathrm{e}_v^{\mathrm{cmp}}(k) +\nonumber\\
    &\qquad \quad \kappa \mathbf{1}(v \in \mathcal{C}_k) \bar{\mathrm{t}}_v^{\mathrm{tx}} P_v^{\mathrm{max}} \leq \mathbf{1}(v \in \mathcal{C}_k) \cdot \mathrm{e}^{\mathrm{bud}}_{v}(k), \forall v, k,\\
    &\label{PD_VELF_Cons4} (\tilde{C}_5) \quad \sum\nolimits_{v=1}^{V_{\breve{t}_k}} \mathbf{1}(v \in \mathcal{C}_k) \hat{\xi}(l_v (k), \eta_v(k)) \leq \Xi(k), ~~ \forall k,
\end{align}
\end{subequations}
where the constraints are taken for the same reasons as in (\ref{partialDevProblem}) and $\Theta_v = \theta_v / \sum_{v=1}^{V_{\breve{t}_k}} \theta_v$, where $\theta_v$ is calculated in (\ref{thetaV}).
\begin{equation}
\label{thetaV}
\begin{aligned} 
    \rs\rs\rs\theta_v &= (1-\bar{\lambda}) [\mathrm{D}_v / \sum\nolimits_{v=1}^{V_{\breve{t}_k}} \mathrm{D}_v] + \bar{\lambda} [\mathrm{t}_{v}^{\mathrm{soj}}(\breve{t}_k)/\sum\nolimits_{v=1}^{V_{\breve{t}_k}} \mathrm{t}_{v}^{\mathrm{soj}}(\breve{t}_k) ],\rs\rs\rs\rs\rs\rs  
\end{aligned} 
\end{equation}
where $\bar{\lambda} \in [0,1]$ is a weighting parameter.
A higher value of $\bar{\lambda}$ puts more weight on the estimated sojourn period, while a smaller value puts more weight on the dataset size.
Note that if the server solves (\ref{iterPlusCpuFreqPD}), it is expected to have $p_v^{\mathrm{suc}} (k) = 1$ since we consider the worst-case estimate for the transmission delay, energy consumption and cost.
In other words, the sub-problem (\ref{iterPlusCpuFreqPD}) retains the original problem's constraints and is expected to maximize the original objective function in each VEFL round.

Problem (\ref{iterPlusCpuFreqPD}) has binary, integer and multiplicative decision variables and is NP-hard. 
Moreover, the CPU frequency and local iteration numbers are inexorably related, affecting the CV selection process.
To tackle these grand challenges, we do standard LP relaxation on $l_v (k)$, which gives the following modified problem.
\begin{subequations}
\label{iterPlusCpuFreqPDLPrelaxed}
\begin{align}
    \underset{\pmb{1} (k), \pmb{l} (k), \pmb{\eta} (k)} {\text{maximize }} &\quad  \sum\nolimits_{v=1}^{V_{\breve{t}_k}} \mathbf{1}(v \in \mathcal{C}_k) l_v (k) \Theta_v , \tag{\ref{iterPlusCpuFreqPDLPrelaxed}} \\
    \text{s.t.} \quad  & C_1, \tilde{C}_2, \tilde{C}_3, \tilde{C}_5,\\
    \label{PD_VEFL_LPRelax_cons5} &(\tilde{C}_6) \quad \mathbf{1}(v \in \mathcal{C}_k) \cdot l_v (k) \geq l_k^{\mathrm{des}},~~ \forall v, k,
\end{align}
\end{subequations}
To that end, we now handle the multiplicative $\mathbf{1}(v \in \mathcal{C}_k)$ and $l_v (k)$ variables.
To tackle this non-linearity, let us define the following new variable:
\begin{equation}
\label{localIterNewVariables}
    \tilde{l}_v(k) \coloneqq l_v (k) \mathbf{1}(v \in \mathcal{C}_k).
\end{equation}
Since $\mathbf{1}(v \! \in \! \mathcal{C}_k) \rs \in \rs \{\! 0,1 \! \} \!$ and $l_v^{\mathrm{min}} \rs \leq \rs l_v (k) \rs \leq \rs l_v^{\mathrm{max}}\rs$, we can equivalently write this new variable with the following inequalities:
\begin{subequations}
\begin{align}
    l_v^{\mathrm{min}} \mathbf{1}(v \in \mathcal{C}_k) &\leq \tilde{l}_v(k) \leq l_v^{\mathrm{max}} \mathbf{1}(v \in \mathcal{C}_k), \label{ltildeCons1}\\
    \rs\rs \rs l_v^{\mathrm{min}} (1-\mathbf{1}(v \in \mathcal{C}_k)) &\leq l_v (k) - \tilde{l}_v(k) \leq l_v^{\mathrm{max}} (1 - \mathbf{1}(v \rs \in \rs \mathcal{C}_k)), \label{ltildeCons2}\rs\rs\rs\\
    0 \leq & \tilde{l}_v(k) \leq l_v^{\mathrm{max}} \label{ltildeCons3}.
\end{align}
\end{subequations}
Moreover, using the DC trick, the binary CV selection decision variables can equivalently be represented as follows:
\begin{subequations}
\begin{align}
    \sum\nolimits_{v=1}^{V_{\breve{t}_k}} \mathbf{1}(v \in \mathcal{C}_k) - \sum\nolimits_{v=1}^{V_{\breve{t}_k}} (\mathbf{1}(v \in \mathcal{C}_k))^2 & \leq 0,\label{ltildeCVselCons1}\\
    0 \leq \mathbf{1}(v \in \mathcal{C}_k) \leq 1, \quad \forall v \in \mathcal{V}_{\breve{t}_k} \label{ltildeCVselCons2}
\end{align}
\end{subequations}
Equation (\ref{ltildeCVselCons1}) is non-convex. 
However, since it is the difference between two convex functions, we can equivalently write the optimization problem as follows:
\begin{subequations}
\label{localIterGivenEtaPDCase1}
\begin{align}
    &\rs\rs\rs\rs\rs \underset{\pmb{1} (k), \tilde{\pmb{l}}_k, \pmb{l} (k),\pmb{\eta} (k)} {\text{minimize }} \rs -\rs \sum_{v=1}^{V_{\breve{t}_k}} \rs \tilde{l}_v(k) \Theta_v + \bar{\vartheta} \rs \Big[\rs\sum_{v=1}^{V_{\breve{t}_k}} \rs \mathbf{1}(v \rs \in \rs \mathcal{C}_k) - \rs \sum_{v=1}^{V_{\breve{t}_k}} \rs (\mathbf{1}(v \rs \in \rs \mathcal{C}_k))^2\rs\Big], \rs\rs \rs \tag{\ref{localIterGivenEtaPDCase1}} \\
    & \qquad \text{s.t.} \quad \label{pDcVselLocalIterTransform1_Cons1_0} (\ref{ltildeCons1}), (\ref{ltildeCons2}) (\ref{ltildeCons3}), (\ref{ltildeCVselCons2}),\\
    & \label{pDcVselLocalIterTransform1_Cons1} (\tilde{C}_1) \quad \sum\nolimits_{v=1}^{V_{\breve{t}_k}} \mathbf{1}(v \in \mathcal{C}_k) \leq |\mathcal{C}_k|, \forall k,\\
    &\label{PD_VEFL1_cons2} (\tilde{C}_2) ~ (\tilde{l}_v(k) c_v \mathrm{D}_v)/\eta_v(k) + \nonumber\\
    &\quad \mathbf{1}(v \in \mathcal{C}_k) \kappa \bar{\mathrm{t}}_v^{\text{tx}} \leq \mathbf{1}(v \in \mathcal{C}_k)  \text{min}\{\mathrm{t}^{\mathrm{th}}(k), \mathrm{t}_{v}^{\mathrm{soj}}(\breve{t}_k)\}  ~ \forall v, k,\rs\rs \\
    &\label{PD_VEFL1_cons3} (\tilde{C}_3)  ~~ \tilde{l}_v(k) (\zeta/2) c_v \mathrm{D}_v {\eta_v (k)}^2 + \nonumber \\
    &\qquad \qquad \mathbf{1}(v \in \mathcal{C}_k) \kappa \bar{\mathrm{t}}_v^{\mathrm{tx}} P_v^{\mathrm{max}} \leq \mathbf{1}(v \in \mathcal{C}_k)  \mathrm{e}^{\mathrm{bud}}_{v}(k), \forall v, k,\rs\\
    &\label{PD_VELF1_Cons4} (\tilde{C}_5) \sum\nolimits_{v=1}^{V_{\breve{t}_k}} \big([ \tilde{l}_v(k) (\zeta/2) c_v \mathrm{D}_v {\eta_v (k)}^2 + \mathbf{1}(v \rs \in \rs \mathcal{C}_k) \kappa P_v^{\mathrm{max}} \bar{\mathrm{t}}_v^{\mathrm{tx}} ] \phi_v + \nonumber\\
    &\qquad \qquad \qquad \mathbf{1}(v \in \mathcal{C}_k)\bar{\phi}_v \big) \leq \Xi(k), \forall k,\\
    &\label{PD_VEFL_LPRelax1_cons5} (\tilde{C}_6) \quad \tilde{l}_v(k) \geq l_k^{\mathrm{des}},~~ \forall v, k,
\end{align}
\end{subequations}
where $\bar{\vartheta} \gg 0$ is a penalty function that forces $\mathbf{1}(v \in \mathcal{C}_k)$ to be either $0$ or $1$.

Note that (\ref{localIterGivenEtaPDCase1}) is a DC problem.
We use successive convex approximation (SCA) to approximate the second quadratic term inside the penalty function as follows:
\begin{align}
\label{taylorExapnsionCVSel}
    &\sum\nolimits_{v=1}^{V_{\breve{t}_k}} (\mathbf{1}(v \rs \in \rs \mathcal{C}_k))^2 \geq \sum\nolimits_{v=1}^{V_{\breve{t}_k}} (\mathbf{1}(v \rs \in \rs \mathcal{C}_k,i))^2 \rs + \sum\nolimits_{v=1}^{V_{\breve{t}_k}} 2 \cdot \mathbf{1}(v \rs \in \rs  \mathcal{C}_k,i)[\nonumber \\
    &\qquad \qquad \mathbf{1}(v \rs \in \rs \mathcal{C}_k) - \mathbf{1}(v \rs \in \rs \mathcal{C}_k,i)] = H(\mathbf{1} (v \rs \in \rs \mathcal{C}_k,i)),\rs\rs\rs\rs
\end{align}
where $\mathbf{1} (v \! \in \! \mathcal{C}_k,i)$ is an initial feasible point.
Moreover, we can linearize the first term in $\tilde{C}_2$ as follows:
\begin{equation}
\begin{aligned}
    &\frac{\tilde{l}_v(k) c_v \mathrm{D}_v}{\eta_v(k)} \approx \frac{\tilde{l}_v(k,i) c_v \mathrm{D}_v}{\eta_v(k,i)} + \frac{c_v \mathrm{D}_v}{\eta_v(k,i)} [\tilde{l}_v(k) - \tilde{l}_v(k,i)] \rs \rs \\
    &~~ + \frac{\tilde{l}_v(k,i) c_v \mathrm{D}_v}{-(\eta_v(k,i))^2}[\eta_v(k) - \eta_v(k,i)] = A_1(\tilde{l}_v(k), \eta_v(k)),
\end{aligned}
\end{equation}
where $\tilde{l}_v(k,i)$ and $\eta_v(k,i)$ are two initial feasible points.
Similarly, we can linearize the multiplicative term in $\tilde{C}_3$ and $\tilde{C}_5$ as follows:
\begin{align}
\label{A_2_iterEta}
     &\tilde{l}_v(k) (\zeta/2) c_v \mathrm{D}_v {\eta_v (k)}^2 \approx \tilde{l}_v(k,i) (\zeta/2) c_v \mathrm{D}_v {\eta_v(k,i)}^2 +\nonumber\\
     &\qquad\qquad\qquad (\zeta/2) c_v \mathrm{D}_v {\eta_v(k,i)}^2 [\tilde{l}_v(k) - \tilde{l}_v(k,i)] +\\
     & \tilde{l}_v(k,i) \zeta c_v \mathrm{D}_v \eta_v(k,i) [\eta_v(k) - \eta_v(k,i)] = A_2(\tilde{l}_v(k), \eta_v(k)).\nonumber\rs \rs\rs 
\end{align}
Moreover, (\ref{PD_VEFL_LPRelax1_cons5}) may not always be possible due to practical limitations. 
As such, we find the upper-bounded $\tilde{l}_v(k)$ based on the time constraints as follows:
\begin{align}
\label{timeConstrainedIterBound}
    &\tilde{l}_v(k) \leq \big[\mathbf{1}(v \in \mathcal{C}_k,i) \mathrm{min}\{\mathrm{t}^{\mathrm{th}}(k), \mathrm{t}_{v}^{\mathrm{soj}}(\breve{t}_k)\} + \nonumber\\
    & \rs  [(\tilde{l}_v(k,i) c_v \mathrm{D}_v) /(-{\eta_v(k,i)}^2)] [\eta_v(k) - \eta_v(k,i)] - \mathbf{1}(v \rs \in \rs \mathcal{C}_k,i) \kappa \bar{\mathrm{t}}_v^{\mathrm{tx}} \nonumber\\
    &\qquad\qquad \qquad \big] \times [\eta_v(k,i)/(c_v\mathrm{D}_v)] = B_1(\tilde{l}_v(k), \eta_v(k)). \rs  
\end{align} 
Furthermore, CV's energy constraint provides the following upper bound.
\begin{align}
\label{energyConstrainedIterBound}
    \rs\rs \rs \tilde{l}_v(k) & \rs \leq \rs \big[ \mathbf{1}(v \rs\in\rs \mathcal{C}_k,i) \mathrm{e}^{\mathrm{bud}}_{v}(k) + (\zeta/2) c_v \mathrm{D}_v {\eta_v(k,i)}^2 [\tilde{l}_v(k) - \tilde{l}_v(k,i)] \rs \rs\nonumber \\ 
    & - \mathbf{1}(v \in \mathcal{C}_k,i) \kappa \bar{\mathrm{t}}_v^{\mathrm{tx}} P_v^{\mathrm{max}} + \tilde{l}_v(k,i) (\zeta/2) c_v \mathrm{D}_v {\eta_v(k,i)}^2\\
    &\qquad \quad \big]/[\tilde{l}_v(k,i) \zeta c_v \mathrm{D}_v \eta_v(k,i)]  
    = B_2(\tilde{l}_v(k), \eta_v(k)). \nonumber
\end{align}
Therefore, combining both time and energy constraints, the upper-bounded $\tilde{l}_v(k)$ is obtained as follows:
\begin{equation}
\label{PDPC_l_UpperBound}
    \tilde{l}_v(k) \leq \mathrm{min} \{B_1(\tilde{l}_v(k), \eta_v(k)), B_2(\tilde{l}_v(k), \eta_v(k))\}.
\end{equation}
Similarly, we can write the lower bound as follows:
\begin{equation}
\label{PDPC_l_LowerBound}
    \tilde{l}_v(k) \geq \mathrm{min} \left\{\rs l_k^{\mathrm{des}} \rs, \mathrm{min}\{ B_1(\tilde{l}_v(k), \eta_v(k)), B_2(\tilde{l}_v(k), \eta_v(k)) \} \rs \right\}\rs.\rs\rs
\end{equation}

To that end, we optimize the upper bound of problem (\ref{localIterGivenEtaPDCase1}) by successively optimizing the following problem:
\begin{subequations}
\label{localIterGivenEtaPDCase2}
\begin{align}
    &\underset{\pmb{1} (k), \tilde{\pmb{l}}_k, \pmb{l} (k), \pmb{\eta} (k)} {\text{minimize }} \rs -\rs\sum_{v=1}^{V_{\breve{t}_k}} \tilde{l}_v(k) \Theta_v + \bar{\vartheta}\bigg[\sum_{v=1}^{V_{\breve{t}_k}} \mathbf{1}(v \rs \in \rs \mathcal{C}_k) - H(\mathbf{1} (v \rs \in \rs \mathcal{C}_k,i))\bigg]\rs,\rs \tag{\ref{localIterGivenEtaPDCase2}} \\
    & \qquad \text{s.t.} \quad \label{pDcVselLocalIterTransform2_Cons1} (\ref{ltildeCons1}), (\ref{ltildeCons2}), (\ref{ltildeCons3}), (\ref{ltildeCVselCons2}), (\ref{PDPC_l_UpperBound}), (\ref{PDPC_l_LowerBound}) \\
    & \label{pDcVselLocalIterTransform1_1_Cons1} (\tilde{C}_1) \quad \sum\nolimits_{v=1}^{V_{\breve{t}_k}} \mathbf{1}(v \in \mathcal{C}_k,i) \leq |\mathcal{C}_k|, ~\forall k,\\
    &\label{PD_VELF1_1_Cons4} (\tilde{C}_5) ~~ \sum\nolimits_{v=1}^{V_{\breve{t}_k}} \big([ A_2(\tilde{l}_v(k),\eta_v(k)) + \kappa \mathbf{1}(v \in \mathcal{C}_k,i) P_v^{\mathrm{max}} \bar{\mathrm{t}}_v^{\mathrm{tx}} ] \phi_v + \nonumber\\
    &\qquad \qquad \qquad \qquad \qquad  \mathbf{1}(v \in \mathcal{C}_k,i)\bar{\phi}_v \big) \leq \Xi(k), \forall k,
\end{align}
\end{subequations}
where $H(\mathbf{1} (v \! \in \! \mathcal{C}_k,i))$ is the global underestimation of $\sum\nolimits_{v=1}^{V_{\breve{t}_k}}(\mathbf{1}(v \in \mathcal{C}_k))^2$ and is calculated in (\ref{taylorExapnsionCVSel}).
Moreover, $A_2(\tilde{l}_v(k),\eta_v(k))$ is calculated in (\ref{A_2_iterEta}).
The server solves (\ref{localIterGivenEtaPDCase2}) to jointly optimize the VEFL parameters and conveys the $l_v^{*}(k)$'s and $\eta_v^{*}(k)$'s to the selected CVs as part of the SLAs.
Note that problem (\ref{localIterGivenEtaPDCase2}) is convex and can be solved efficiently using existing solvers such as CVX \cite{diamond2016cvxpy}.
We use Algorithm \ref{pdCase_CVSelLocalIter} to solve (\ref{localIterGivenEtaPDCase2})  iteratively. 
Since (\ref{localIterGivenEtaPDCase2}) has $4V_{\breve{t}_k}$ decision variables, $2 + 6V_{\breve{t}_k}$ constraints, and Algorithm \ref{pdCase_CVSelLocalIter} runs for a maximum of $I$ iterations, the computation time complexity of our proposed solution is $\mathcal{O} \big(64IV_{\breve{t}_k}^3 (2 + 6V_{\breve{t}_k})\big)$ \cite{7812683}.

\begin{algorithm} [t!] 
\fontsize{9}{8}\selectfont
\SetAlgoLined 
\DontPrintSemicolon
\KwIn{Initial feasible set $\mathbf{1} (v \in \mathcal{C}_k,i)$, $\tilde{\pmb{l}} (k, i)$, $\pmb{\eta} (k,i)$, $i=0$, maximum iteration $I$, precision level $\epsilon^{\mathrm{prec}}$, initial penalty $\vartheta_0$}
\nl{\textbf{Repeat}:} \;
\Indp {
    $i \gets i+1$ ; $\bar{\vartheta} \gets \vartheta_0+i$ \;
    Solve (\ref{localIterGivenEtaPDCase2}) using $\bar{\vartheta}$, $\mathbf{1} (v \in \mathcal{C}_k, i-1)$, $\tilde{\pmb{l}} (k, i-1)$ and $\pmb{\eta} (k, i-1)$ to find $\mathbf{1} (v \in \mathcal{C}_k,i)$, $\tilde{\pmb{l}} (k, i)$ and $\pmb{\eta} (k, i)$ \;
    }
\Indm \textbf{Until} converge with $\epsilon^{\mathrm{prec}}$ precision or $i=I$ \;
\KwOut{CV selection set $\pmb{1} (k)$, local iteration set $\pmb{l} (k)$ and CPU frequencies $\pmb{\eta} (k)$}
\caption{Iterative Joint CV Selection and Local Iteration Selection Process}
\label{pdCase_CVSelLocalIter}
\end{algorithm}

\begin{Remark}
Note that for FDPC, as all CVs are to be selected, there are no binary $\mathbf{1} (v \in \mathcal{C}_k)$ variables. 
Apart from the DC trick, we follow a similar approach as in (\ref{iterPlusCpuFreqPD}) to jointly optimize $\pmb{l}(k)$ and $\pmb{\eta}(k)$ by maximizing the weighted sum of the CVs' local iterations, where we put the weight based on the aggregation weight $p_v = \bar{p}_v/\sum\nolimits_{v=1}^{V_{\breve{t}_k}} \bar{p}_v$. 
We omitted the optimization problem for brevity. 
Moreover, the following RAT parameter optimization process is also used for FDPC.
\end{Remark}

\subsection{RAT Parameter Optimization}
\label{ratParamOptimSubSection}
\noindent
The server initiates CVs' trained models reception from slot $\tau (k) = \mathrm{min} \big\{\tau_1 (k), \dots, \tau_{|\mathcal{C}_k|} (k) \big\}$, where $\tau_v (k)$ is calculated in (\ref{slotTXStart}), to minimize its operating cost.
\begin{equation}
\label{slotTXStart}
\begin{aligned}
    \tau_v (k) &= 
    \begin{cases}
        \breve{t}_{k} + \big\lfloor \mathrm{t}_{v}^{\mathrm{soj}}(\breve{t}_k)/\kappa \big\rfloor - \bar{\mathrm{t}}_v^{\mathrm{tx}}, & \text{ if  } \mathrm{t}_{v}^{\mathrm{soj}} (\breve{t}_k) < \mathrm{t}^{\mathrm{th}}(k), \rs\rs \\
        \breve{t}_{k+1} - \bar{\mathrm{t}}_v^{\mathrm{tx}}, & \text{ otherwise },
    \end{cases}.\rs
\end{aligned}
\end{equation}
Intuitively, the RAT resources are required only when there are uplink payloads.
As such, (\ref{slotTXStart}) calculates the slot at which a CV finishes its local model training, and $\tau(k)$ ensures that the server does not start the model receptions if no CV finishes its local model training.

Upon solving (\ref{localIterGivenEtaPDCase2}), the server broadcasts the global model $\pmb{\omega}_k$, $\pmb{l} (k)$, $\pmb{\eta}(k)$ and the transmission-reception starting slot $\tau (k)$.
Then, after performing $l_v (k)$ local iterations, CV $v$ needs to offload its trained model $\pmb{\omega}_{v,k+1}$ to the server.
From slot $\tau (k)$ and onward until the $\breve{t}_{k+1}$, the server needs to perform CV scheduling and resource allocations to successfully receive the payloads within each CV's specific deadline constraint.
Denote these optimization slots for the server by the set $\mathcal{T}_k = \{t\}_{t=\tau (k)}^{\breve{t}_{k+1}}$.

Note that all CVs have the same payload of $S(\pmb{\omega}, \mathrm{FPP})$ bits to offload to the server since they receive the same model. 
The server maintains a virtual remaining-local-model-payload buffer of the CVs $\mathbf{Q}(t)\overset{\Delta}{=}[Q_1(t), \dots, Q_{|\mathcal{C}_k|}(t)]$, where CV $v$'s remaining payload buffer $Q_v(t)$ evolves as follows:  
\begin{equation}
\label{payloadQueue}
    Q_v(t+1) = [Q_v(t) - \kappa \cdot r_v(t)]^{+}, ~\forall~ t,
\end{equation}
where $[\cdot]^{+}$ means $\mathrm{max}\{\cdot, 0\}$ and $r_v(t)$ is the uplink data rate calculated in (\ref{uplink_DataRate}).
{\color{black}{
Considering the above factors, we can calculate the probability of successful reception of the CV's trained model $\pmb{\omega}_{v,{k+1}}$ as follows:
\begin{equation}
\label{ProbSuc}
    p_v^{\mathrm{suc}} (k) = 1 - [Q_v(\breve{t}_{k+1})/S(\pmb{\omega}, \mathrm{FPP})].
\end{equation}
We stress that this $p_v^{\mathrm{suc}} (k)$ depends on the CSI, RAT parameters and the quality of the worst-case required number of slots $\bar{\mathrm{t}}_v^{\mathrm{tx}}$, calculated in (\ref{worstCaseTxTTIEstimate}).

Given that the VEFL parameters are solved using (\ref{localIterGivenEtaPDCase2}), we aim to solve the following sub-problem for jointly optimizing the RAT parameters.
\begin{subequations}
\label{fdRATProbOriginal}
\begin{align}
    &\underset{\mathbf{x}(t)} {\text{maximize }}  \quad  \sum\nolimits_{v=1}^{|\mathcal{C}_k|} p_v^{\mathrm{suc}}(k), \tag{\ref{fdRATProbOriginal}} \\
    &  \text{  s.t.} \qquad \qquad  C_7-C_{12}, \label{fdRATCons1_Org} \\
    & \mathrm{t}^{\text{tx}}_{v}(k) \leq \kappa \cdot \mathrm{min}\left\{\breve{t}_k + \big\lfloor \mathrm{t}_{v}^{\mathrm{soj}}(\breve{t}_k)/\kappa \big\rfloor - \tau_v (k), \breve{t}_{k+1} - \tau_v (k) \rs \right\}\rs,\rs \label{fdRATCons2_Org} 
\end{align}
\end{subequations} 
where $\mathbf{x}(t) = [\pmb{\mathrm{I}}(t), \breve{\pmb{\mathrm{I}}}(t), \mathbf{p}(t)]$.
Constraints in (\ref{fdRATCons1_Org}) are taken for the same reasons as in the original problem in (\ref{partialDevProblem}).
Besides, constraint (\ref{fdRATCons2_Org}) ensures that the entire trained model $\pmb{\omega}_{v,k+1}$ has to be received within the remaining $\mathrm{min}\big\{\breve{t}_k + \big\lfloor \mathrm{t}_{v}^{\mathrm{soj}}(\breve{t}_k)/\kappa \big\rfloor - \tau_v (k), \breve{t}_{k+1} - \tau_v (k) \big\}$ slots.
\begin{Remark}
   The maximization of $\sum_{v=1}^{|\mathcal{C}_k|} p_v^{\mathrm{suc}}(k)$ in (\ref{fdRATProbOriginal}) is not straightforward since the CSI varies in each slot $t$, and the CVs require multiple transmission slots to offload their trained models.
   Besides, while $p_v^{\mathrm{suc}}(k)$ is calculated at the end of the current VEFL round $k$, based on (\ref{ProbSuc}), the RAT parameters must be optimized in each slot $t$.
   Note that the radio resources are available for the entire VEFL round. 
   Intuitively, the server should seek a cost-effective way to perform the VEFL training and trained model reception to optimize the associated monetary cost, as defined in (\ref{cost_per_iter}). 
   As such, we aim to devise a per-slot-based long-term energy efficiency (EE) maximization problem to maximize the probability of successful model reception implicitly. 
   In other words, to maximize $p_{v}^{\mathrm{suc}}(k)$, the server aims to maximize the transmission EE, which shall also ensure optimized cost $\xi_v(l_v(k), \eta_v(k))$.
\end{Remark}
}}

Without any loss of generality, we define EE as follows:
\begin{equation}
\label{EE}
    \beta (t) = \big(\sum\nolimits_{v=1}^{|\mathcal{C}_k|} r_v(t)\big)/\big(\sum\nolimits_{v=1}^{|\mathcal{C}_k|} \sum\nolimits_{z=1}^Z P_{v,z}(t)\big).
\end{equation}
Note that (\ref{EE}) is the standard EE of a wireless network that calculates the tradeoff between the total sum rate and the total energy expense \cite{9174768}.
The expected average EE during round $k$ is calculated as 
\begin{equation}
\label{performMetric}
    \bar{\beta} (k) = \bar{\mathrm{r}} (k) / \bar{\mathrm{P}} (k), 
\end{equation}
where $\bar{\mathrm{r}} (k) = \left[1/(\breve{t}_{k+1} - \tau (k)) \right]  \sum\nolimits_{t=\tau (k)}^{\breve{t}_{k+1}} \mathbb{E} \big[ \sum\nolimits_{v=1}^{|\mathcal{C}_k|} r_v(t)\big] $ and $\bar{\mathrm{P}} (k) =  \left[1/(\breve{t}_{k+1} - \tau (k)) \right]  \sum\nolimits_{t=\tau (k)}^{\breve{t}_{k+1}} \mathbb{E} \big[ \sum\nolimits_{v=1}^{|\mathcal{C}_k|} \sum\nolimits_{z=1}^Z P_{v,z}(t) \big]$.
We jointly optimize the RAT parameters to maximize the EE as
\begin{subequations}
\label{fdRATProb}
\begin{align}
    &\underset{\mathbf{x}(t)} {\text{maximize }}  \quad  \bar{\beta} (k), \tag{\ref{fdRATProb}} \\
    &  \text{  s.t.} \quad  C_7-C_{12}, \label{fdRATCons1} \\
    & \mathrm{t}^{\text{tx}}_{v}(k) \leq \kappa \cdot \mathrm{min}\left\{\breve{t}_k + \big\lfloor \mathrm{t}_{v}^{\mathrm{soj}}(\breve{t}_k)/\kappa \big\rfloor - \tau_v (k), \breve{t}_{k+1} - \tau_v (k) \rs \right\}\rs,\rs \label{fdRATCons2}
\end{align}
\end{subequations}
where the constraints are taken for the same reasons as in (\ref{fdRATProbOriginal}).

\begin{Remark}
    By maximizing the EE, problem (\ref{fdRATProb}) ensures that the monetary cost for the uplink transmission is optimized, while constraint (\ref{fdRATCons2}) ensures that all local $\pmb{\omega}_{v,k+1}$'s are fully received. 
    Therefore, if problem (\ref{fdRATProb}) is solved, the server should ensure the maximization of $\sum\nolimits_{v=1}^{|\mathcal{C}_k|} p_{v}^{\mathrm{suc}}(k)$ in a cost effective way, which provides the expected loss minimization of the VEFL.
\end{Remark}

The objective function in (\ref{fdRATProb}) is fractional, non-convex and challenging to solve.
In practice, the Dinkelbach method \cite{dinkelbach1967nonlinear} is widely used to transform the fractional objective function equivalently into a linear one.
Denote the optimal solution for (\ref{fdRATProb}) by $\mathbf{x}^{*}(t)=[\pmb{\mathrm{I}}^{*}(t), \breve{\pmb{\mathrm{I}}}^{*}(t), \mathbf{p}^{*}(t)]$ and the corresponding optimal long-term EE by $\bar{\beta} (k)^{*} = \bar{\mathrm{r}} (k)^{*} / \bar{\mathrm{P}} (k)^{*}$.
Then, using the Dinkelbach approach, we transform the fractional objective function to the following equivalent form \cite{8643543}:
\begin{equation}
\label{equivDinkObj}
    \bar{\beta} (k) = \bar{\mathrm{r}} (k) / \bar{\mathrm{P}} (k) = \bar{\mathrm{r}} (k) - \bar{\beta} (k)^{*} \bar{\mathrm{P}} (k).
\end{equation}
However, since $\bar{\beta} (k)^{*}$ is unknown beforehand, available historical information is usually utilized. 
Denote $\bar{\beta} (k,t)$ as
\begin{equation}
\label{barBetaDink}
    \bar{\beta} (k,t) \rs \coloneqq \rs \!\big[\sum\nolimits_{\hat{t}=\tau (k)}^{t-1} \rs \sum\nolimits_{v=1}^{|\mathcal{C}_k|} \rs r_v(\hat{t}) \! \big] \! \big/ \! \big[\!\sum\nolimits_{\hat{t}=\tau (k)}^{t-1} \rs \sum\nolimits_{v=1}^{|\mathcal{C}_k|}\sum\nolimits_{z=1}^Z \rs P_{v,z}(\hat{t})\big]\rs,\rs\rs  
\end{equation}
where $\bar{\beta} (k, \tau (k))=0$.

To that end, plugging (\ref{equivDinkObj}) and (\ref{barBetaDink}) into the objective function of (\ref{fdRATProb}), we get the modified objective $\bar{\mathrm{r}} (k) - \bar{\beta} (k,t) \bar{\mathrm{P}} (k)$. 
Then, we can rewrite problem (\ref{fdRATProb}) as follows: 
\begin{subequations}
\label{fdRATProbDinkTransformed1}
\begin{align}
    &\underset{\mathbf{x}(t)} {\text{minimize }}  \quad  -\bar{\mathrm{r}} (k) + \bar{\beta} (k,t) \bar{\mathrm{P}} (k), \tag{\ref{fdRATProbDinkTransformed1}} \\
    & \quad  \text{  s.t.} \quad  C_7-C_{12}, (\ref{fdRATCons2}), \label{fdRATtransformed1Cons}
\end{align}
\end{subequations}
Note that problems (\ref{fdRATProb}) and (\ref{fdRATProbDinkTransformed1}) are equivalent \cite{8643543}, and similar treatment is widely used in the literature \cite{9174768,9136761,8643543,9718549}.

Clearly, problem (\ref{fdRATProbDinkTransformed1}) is stochastic due to the intermittent Uu links and is challenging to solve. As such, we leverage the Lyapunov optimization framework to transform (\ref{fdRATProb}) into an online and per-slot-wise deterministic optimization problem.
Let us define a quadratic Lyapunov function as
\begin{equation}
\label{lyapunovFunc}
    L(\mathbf{Q}(t)) \coloneqq (1/2) \sum\nolimits_{v=1}^{|\mathcal{C}_k|} {Q_v(t)}^2.
\end{equation}
Then, we can write the conditional Lyapunov drift as
\begin{equation}
    \Delta (\mathbf{Q}(t)) = \mathbb{E}\left[L(\mathbf{Q}(t+1)) - L(\mathbf{Q}(t))|\mathbf{Q}(t)\right].
\end{equation}
To that end, we express a Lyapunov drift-plus-penalty function as follows:
\begin{equation}
\label{lyapunovDrifPenaltyFunc}
\begin{aligned}
    \Delta_C (\mathbf{Q}(t)) &= \Delta (\mathbf{Q}(t)) + C \cdot \mathbb{E} \Big[- \sum\nolimits_{v=1}^{|\mathcal{C}_k|} r_v(t) +\\
    &\qquad \qquad \quad \bar{\beta} (k,t)\sum\nolimits_{v=1}^{|\mathcal{C}_k|} \sum\nolimits_{z=1}^Z P_{v,z}(t) |\mathbf{Q}(t) \Big],
\end{aligned}
\end{equation}
where $C \in [0, +\infty]$ is a control parameter that adjusts the trade-off between EE and payload buffer backlogs.

\begin{Remark}
    We want to minimize the upper-bounded Lyapunov drift-plus-penalty function in (\ref{lyapunovDrifPenaltyFunc}). 
    The right-hand side suggests minimizing the drift, i.e., payload buffer size, while the added penalty controls the EE of the system.
\end{Remark}

\begin{Lemma}
\label{lyaPunovLemma}
For an arbitrary $\mathbf{x}(t) = [\pmb{\mathrm{I}}(t), \breve{\pmb{\mathrm{I}}}(t), \mathbf{p}(t)]$, the Lyapunov drift-plus-penalty function is upper-bounded as
\begin{equation}
\label{drifPenaltyEq}
\begin{aligned}
    &\Delta_C (\mathbf{Q} (t))  \leq \varpi -\mathbb{E} \left[\sum\nolimits_{v=1}^{|\mathcal{C}_k|} \kappa r_v(t) Q_v(t) | \mathbf{Q}(t) \right] + \\
    &~~  C \cdot \mathbb{E} \left[-\sum\nolimits_{v=1}^{|\mathcal{C}_k|} r_v(t) + \bar{\beta} (k,t)\sum\nolimits_{v=1}^{|\mathcal{C}_k|} \sum\nolimits_{z=1}^Z P_{v,z}(t) |\mathbf{Q}(t) \right],
\end{aligned}
\end{equation}
where $\varpi \geq 0$ does not depend on the queue state.
\end{Lemma}
\begin{proof}
The proof is left in Appendix \ref{proofLyaPunovLemma}.
\end{proof}

From Lemma \ref{lyaPunovLemma}, we find the following per-slot utility function that we want to minimize:
\begin{equation}
\begin{aligned}
    &\rs \rs \mathrm{U}(t) \! = \! C \bar{\beta} (k,t) \rs \sum\nolimits_{v=1}^{|\mathcal{C}_k|} \sum\nolimits_{z=1}^Z \rs P_{v,z}(t) \! - \rs \sum\nolimits_{v=1}^{|\mathcal{C}_k|} r_v(t) \! \big[\!\kappa Q_v(t) \! + \! C \big]\rs.\rs \rs \rs \rs\rs \rs\rs \rs
\end{aligned}
\end{equation}
Using (\ref{uplink_SINR}) and (\ref{uplink_DataRate}), we can write
\begin{align*}
    &\mathrm{U}(t) = C \bar{\beta} (k,t) \sum\nolimits_{v=1}^{|\mathcal{C}_k|} \sum\nolimits_{z=1}^Z \rs P_{v,z}(t) - \omega (1 - \upsilon) \sum\nolimits_{v=1}^{|\mathcal{C}_k|} \rs \mathrm{I}_v(t) \big[\\
    &\sum\nolimits_{z=1}^Z  \log_2 ( 1 + (\mathrm{I}_{v,z}(t) P_{v,z}(t) \left\Vert \mathbf{h}_{v,z}(t) \right\Vert^2)/(\omega\varsigma^2) ) \big] \left(\kappa Q_v(t) + C\right).\rs \rs
\end{align*}
As such, we can pose the following per-slot online optimization problem that can be solved opportunistically.
\begin{subequations}
\label{fdRATProbTransform}
\begin{align}
    \underset{\mathbf{x}(t)} {\text{minimize }}  & \quad  \mathrm{U} (t), \tag{\ref{fdRATProbTransform}} \\
    \text{  s.t.} \quad  \quad & C_7-C_{12}, (\ref{fdRATCons2}), \label{fdRATtransformCons1} 
\end{align}
\end{subequations}
where the constraints are taken for the same reasons as in (\ref{fdRATProbDinkTransformed1}).

Recall that CV $v$ starts offloading its trained model from slot $\tau_v (k)$.
Therefore, in slot $t \in \mathcal{T}_k$, denote the CVs that need to offload their trained models by the set $\mathcal{U}(t)$.
Since the gNB can schedule at most $Z$ CVs in a slot for their uplink transmissions and all $\pmb{\omega}_{v,k+1}$'s need to be received within slot $\breve{t}_{k+1}$ to meet the deadline threshold, the gNB has to perform a scheduling decision when $|\mathcal{U}(t)| > Z$.
As such, to ensure (\ref{fdRATCons2}), the medium access control (MAC) scheduler schedules the CVs based on the remaining deadlines and payload buffer status $\mathbf{Q}_v(t-1)$.
Particularly, we adopt the widely used real-time operating system's  EDF-based scheduling\footnote{Similar scheduling is also widely used in deadline-constrained applications \cite{Guan2020Maximize, Guo2019Enabling, pervej2022efficient}.} \cite{buttazzo2011hard}. 
Note that if EDF cannot guarantee zero deadline violation, no other algorithm can \cite{Guan2020Maximize}. 
Algorithm \ref{cvScheduling} provides the scheduled CV set $\bar{\mathcal{U}}(t)$ during slot $t$.

Per this scheduling policy, a CV is scheduled in at least $\big\lfloor (Z/|\mathcal{C}_k|) \times \mathrm{min}\left\{\breve{t}_k + \big\lfloor \mathrm{t}_{v}^{\mathrm{soj}}(\breve{t}_k)/\kappa \big\rfloor, \breve{t}_{k+1} \right\} - \tau_v (k) \big\rfloor$ slots.
As such, to satisfy constraint (\ref{fdRATCons2}), we enforce
\begin{equation}
\label{UlReqTxBits}
    r_v(t) \geq \rs 
    \begin{cases}
        \rs \frac{Q_v(t-1)} {\rs \kappa \times \mathrm{min}\left\{\breve{t}_k + \big\lfloor \mathrm{t}_{v}^{\mathrm{soj}}(\breve{t}_k)/\kappa \big\rfloor - \tau_v (k), \breve{t}_{k+1} - \tau_v (k) \right\}}, \rs\rs \rs\rs & \rs\rs\rs \rs \rs\rs \text{if } |\mathcal{C}_k| \leq Z,\\
        \rs \frac{Q_v(t-1)} {\rs \kappa  \times \rs \big\lfloor \rs (Z/|\mathcal{C}_k|) \times \mathrm{min}\left\{\breve{t}_k + \big\lfloor \mathrm{t}_{v}^{\mathrm{soj}}(\breve{t}_k)/\kappa \big\rfloor, \breve{t}_{k+1} \right\} - \tau_v (k) \big\rfloor}, & \rs \text{else},
    \end{cases}\rs\rs, 
\end{equation}
where $t=\tau_v (k), \dots, \mathrm{min}\{\breve{t}_k + \big\lfloor \mathrm{t}_{v}^{\mathrm{soj}}(\breve{t}_k)/\kappa \big\rfloor, \breve{t}_{k+1}$\}.

To that end, note that in a particular slot $t$, if only $1$ CV requires offloading its trained model, it shall have all pRBs, i.e., $\mathrm{I}_{v,z}(t)=1, \forall z \in \mathcal{Z}$.
The server only needs to optimize the transmission power of the scheduled CV. 
As such, in this special case, we pose the following optimization problem: 
\begin{subequations}
\label{fdRATProbTransform_Special1CV}
\begin{align}
    \underset{\mathbf{p}_v(t)} {\text{minimize }} \quad & \quad  \mathrm{U} (t), \tag{\ref{fdRATProbTransform_Special1CV}} \\
    \text{s.t. } \qquad & \quad  C_{11}, C_{12}, (\ref{UlReqTxBits}), \mathrm{I}_{v,z} (t) = 1, \quad  \forall z.
\end{align}
\end{subequations}

\begin{algorithm} [t!]
\fontsize{9}{8}\selectfont
\SetAlgoLined 
\DontPrintSemicolon
Determine the CV set $\mathcal{U}(t)$ eligible for scheduling \;
Set empty set $\bar{\mathcal{U}}(t)$ \;
\uIf{$|\mathcal{U}(t)| \leq Z$}{
    $\bar{\mathcal{U}}(t) \gets \mathcal{U}(t)$ \;
    Set $\mathrm{I}_v(t) = 1, \forall v \in \mathcal{U}(t)$} 
\Else{
    $\bar{\mathcal{U}}(t) \gets \mathrm{indmin}_{[Z]} \{T_1^{\mathrm{thr}}, \dots, T_{|\mathcal{U}(t)|}^{\mathrm{thr}} \}$, where $T_v^{\mathrm{thr}} = \mathrm{min}\left\{\breve{t}_k + \big\lfloor \mathrm{t}_{v}^{\mathrm{soj}}(\breve{t}_k)/\kappa \big\rfloor - t, \breve{t}_{k+1} - t \right\}$ \tcc{\textit{$\mathrm{indmin}_{[Z]}\{\mathcal{Q}\}$ means the index of the first $Z$ smallest entry of $\mathcal{Q}$}}
    Set $\mathrm{I}_v(t)=1$ $\forall v \in \bar{\mathcal{U}}(t)$ 
}
\KwOut{Scheduled CV set $\bar{\mathcal{U}}(t)$}
\caption{CV Scheduling in Each TTI}
\label{cvScheduling}
\end{algorithm}

However, when $|\mathcal{U}(t)|>1$, the MAC scheduler schedules the CVs in the set $\bar{\mathcal{U}}(t)$ based on Algorithm \ref{cvScheduling}.
We, therefore, reformulate optimization problem (\ref{fdRATProbTransform}) as
\begin{subequations}
\label{fdRATProbTransform1}
\begin{align}
    &\underset{\bar{\mathbf{x}}(t)} {\text{minimize }} \quad  \bar{\mathrm{U}} (t), \tag{\ref{fdRATProbTransform1}} \\
    &\text{s.t. }  C6, C7, C10, C11, (\ref{UlReqTxBits}), \mathrm{I}_{v,z}(t) \rs \in \rs \{0,1\}, \forall v \in \bar{\mathcal{U}}(t), z, \label{fdRATtransformCons1_1}
\end{align}
\end{subequations}
where $\bar{\mathbf{x}} (t) = [\breve{\pmb{\mathrm{I}}}(t), \mathbf{p}(t)] $ is the decision variables set for all $v\in \bar{\mathcal{U}}(t)$ and $\bar{\mathrm{U}}(t)$ is calculated as 
\begin{align*}
\label{barUt}
    &\bar{\mathrm{U}}(t) = C \bar{\beta} (k,t) \rs \sum\nolimits_{v \in \bar{\mathcal{U}}(t)} \! \sum\nolimits_{z=1}^Z \rs P_{v,z}(t) \!-\! \omega(1-\upsilon) \rs \sum\nolimits_{v \in \bar{\mathcal{U}}(t)} \rs \big[\rs\\
    & \sum\nolimits_{z=1}^Z \log_2 ( 1 + (\mathrm{I}_{v,z}(t) P_{v,z}(t) \left\Vert \mathbf{h}_{v,z}(t) \right\Vert^2)/(\omega\varsigma^2) ) \big] [ C + \kappa Q_v(t)]. \rs
\end{align*}    
Problem (\ref{fdRATProbTransform1}) has binary and multiplicative decision variables, making it non-convex.
To handle these complexities, we introduce the following new variable.
\begin{equation}
\label{newVarPower}
    \tilde{P}_{v,z} (t) \coloneqq \mathrm{I}_{v,z}(t) P_{v,z}(t).
\end{equation}
Now, since $\mathrm{I}_{v,z} (t) \in\{0,1\}$ and $0\leq P_{v,z}(t) \leq P_v^{\mathrm{max}}$, $\tilde{P}_{v,z} (t)$ is equivalent to the following inequalities:
\begin{subequations}
\begin{align}
    0 \cdot \mathrm{I}_{v,z}(t) & \leq \tilde{P}_{v,z} (t) \leq P_v^{\mathrm{max}} \cdot \mathrm{I}_{v,z}(t), \label{transformCons1}\\
    0\cdot[1 - \mathrm{I}_{v,z}(t)] & \leq P_{v,z}(t) - \tilde{P}_{v,z} (t) \leq P_v^{\mathrm{max}} \cdot [1 - \mathrm{I}_{v,z}(t)], \label{transformCons2}\\
    0 \leq & \tilde{P}_{v,z} (t) \leq P_v^{\mathrm{max}},\label{transformCons3}
\end{align}
\end{subequations}
Note that the binary decision variable $\mathrm{I}_{v,z}(t) \in \{0,1\}$ is equivalent to $\mathrm{I}_{v,z}(t) - (\mathrm{I}_{v,z}(t))^2 = 0$.
Then, using the DC trick, we equivalently represent this binary decision constraint as \cite{6816086,7812683}:
\begin{subequations}
\begin{align}
    &\sum\nolimits_{v\in\bar{\mathcal{U}}(t)} \sum\nolimits_{z=1}^Z \mathrm{I}_{v,z}(t) - \sum\nolimits_{v\in\bar{\mathcal{U}}(t)} \sum\nolimits_{z=1}^Z (\mathrm{I}_{v,z}(t))^2 \leq 0, \label{transformCons4}\\ 
    & \qquad 0 \leq \mathrm{I}_{v,z}(t) \leq 1, \quad \forall v \in \bar{\mathcal{U}}(t), z \in \mathcal{Z}. \label{transformCons5}
\end{align}
\end{subequations}
To that end, using (\ref{newVarPower}-\ref{transformCons5}), we equivalently rewrite (\ref{fdRATProbTransform1}) as 
\begin{subequations}
\label{fdRATequivalProb}
\begin{align}
    &\underset{\tilde{\mathbf{x}}(t)} {\text{minimize }} \quad  \tilde{\mathrm{U}} (t), \tag{\ref{fdRATequivalProb}} \\
    &\text{  s.t.} \quad  C_7, C_8, (\ref{transformCons1}), (\ref{transformCons2}), (\ref{transformCons3}), (\ref{transformCons4}), (\ref{transformCons5}), \label{fdRATequivalCons1_1}\\ 
    &\qquad \quad  \sum\nolimits_{z=1}^Z \tilde{P}_{v,z} (t) \leq P_v^{\mathrm{max}},  ~~ \forall t, \forall z, \label{fdRATequivalCons2_2}\\
    &\rs \tilde{r}_v(t) \geq \frac{Q_v(t-1)}{\rs \rs \kappa \rs \times \rs \big\lfloor \rs (Z/|\mathcal{C}_k|) \rs \times \rs (\mathrm{min}\{\breve{t}_k + \big\lfloor \mathrm{t}_{v}^{\mathrm{soj}}(\breve{t}_k)/\kappa \big\rfloor, \breve{t}_{k+1} \} - t) \big\rfloor},\rs\rs \label{fdRATequivalCons2_3}
\end{align}
\end{subequations}
where $\tilde{\mathbf{x}}(t) = [\breve{\pmb{I}}(t),\mathbf{p}(t),\tilde{\mathbf{p}}(t)]$, $\tilde{\mathbf{p}}(t)=[\tilde{\mathbf{p}}_1(t), \dots, \tilde{\mathbf{p}}_{|\mathcal{C}_k|}(t)]^T$, $\tilde{\mathbf{p}}_v(t)=[\tilde{P}_v^1(t),\dots, \tilde{P}_{v,z} (t)]^T$, 
$\tilde{\mathrm{U}}(t) = C \bar{\beta} (k,t) \sum\nolimits_{v \in \bar{\mathcal{U}}(t)}\sum\nolimits_{z=1}^Z \tilde{P}_{v,z} (t) - \omega(1-\upsilon) \sum\nolimits_{v \in \bar{\mathcal{U}}(t)} [\sum\nolimits_{z=1}^Z \log_2 ( 1 + (\tilde{P}_{v,z} (t) \left\Vert \mathbf{h}_{v,z}(t) \right\Vert^2)/(\omega\varsigma^2) ) ] [C + \kappa Q_v(t)] $ 
and $\tilde{r}_v(t) = \omega(1-\upsilon)\sum\nolimits_{z=1}^Z [\log_2 (1 + (\tilde{P}_{v,z} (t) \left\Vert \mathbf{h}_{v,z}(t) \right\Vert^2)/(\omega \varsigma^2) ) ]$.

Notice that constraint (\ref{transformCons4}) makes problem (\ref{fdRATequivalProb}) non-convex. 
Particularly, constraint (\ref{transformCons4}) is the difference between two convex functions.
Thus, this constraint can be incorporated into the objective function to act as a penalty when violated. 
Particularly, for a sufficiently large constant $\vartheta$, optimization problem (\ref{fdRATequivalProb}) can be equivalently \cite{6816086} represented as
\begin{subequations}
\label{fdRATequivalProb1}
\begin{align}
    &\underset{\tilde{\mathbf{x}}(t)} {\text{minimize }}  \tilde{\mathrm{U}} (t) \! + \! \vartheta \! \Big[\rs \sum_{v\in\bar{\mathcal{U}}(t)} \rs \rs \rs \sum\nolimits_{z=1}^Z \mathrm{I}_{v,z}(t) - \rs \rs \rs \rs \sum_{v\in\bar{\mathcal{U}}(t)} \rs \rs \rs \sum\nolimits_{z=1}^Z (\mathrm{I}_{v,z}(t))^2 \! \Big]\rs, \tag{\ref{fdRATequivalProb1}} \rs \rs \rs \rs \\
    &\text{  s.t.} \quad  C_7, C_8, (\ref{transformCons1}), (\ref{transformCons2}), (\ref{transformCons3}), (\ref{transformCons5}), (\ref{fdRATequivalCons2_2}), (\ref{fdRATequivalCons2_3}). \label{fdRATequivalCons2_1}
\end{align}
\end{subequations}
Intuitively, $\vartheta \gg 0$ penalizes the objective function when any $\mathrm{I}_{v,z}(t)$ is not either $0$ or $1$.
Therefore, for a hefty $\vartheta$, the pRB allocation parameter will need to be either $0$ or $1$ to minimize the objective function \cite{6816086, 7812683}.

\begin{Remark}
    The optimization problem (\ref{fdRATequivalProb1}) belongs to the DC function programming because all four terms in the objective functions are convex, and the constraints span a convex set.
\end{Remark}
To that end, we use SCA to approximate the second term of (\ref{transformCons4}).
Using Taylor expansion, we write the following for a feasible point $\mathrm{I}_{v,z}(t, j)$.
\begin{align}
\label{approxRATsquaredTerm}
    &\sum\nolimits_{v\in\bar{\mathcal{U}}(t)} \sum\nolimits_{z=1}^Z (\mathrm{I}_{v,z}(t))^2 \geq \sum\nolimits_{v\in\bar{\mathcal{U}}(t)} \sum\nolimits_{z=1}^Z (\mathrm{I}_{v,z}(t, j))^2 + \nonumber\\
    &~\sum\nolimits_{v\in\bar{\mathcal{U}}(t)} \sum\nolimits_{z=1}^Z 2\mathrm{I}_{v,z}(t, j) [\mathrm{I}_{v,z}(t) -\mathrm{I}_{v,z}(t, j)] = H(\tilde{\pmb{\mathrm{I}}}(t, j)),
\end{align}
where $\tilde{\pmb{\mathrm{I}}}(t, j)$ is the set of all feasible points for all $v$ and $z$.
As such, given an initial feasible $\tilde{\mathbf{x}}(t,j)$, we obtain the upper bound of (\ref{fdRATequivalProb1}) via optimizing the following convex optimization problem.
\begin{subequations}
\label{fdRATequivalProb2}
\begin{align}
    &\underset{\tilde{\mathbf{x}}(t, j)} {\text{minimize }}~~  \tilde{\mathrm{U}} (t) + \vartheta \Big[\rs \sum\nolimits_{v\in\bar{\mathcal{U}}(t)} \sum\nolimits_{z=1}^Z \mathrm{I}_{v,z}(t) - H(\tilde{\pmb{\mathrm{I}}}(t, j)) \Big], \tag{\ref{fdRATequivalProb2}} \rs\\
    &\text{  s.t.} \quad  C_7, C_8, (\ref{transformCons1}), (\ref{transformCons2}), (\ref{transformCons3}), (\ref{transformCons5}),(\ref{fdRATequivalCons2_2}), (\ref{fdRATequivalCons2_3}), \label{fdRATequival1Cons1_1} 
\end{align}
\end{subequations}
where $\tilde{\mathbf{x}}(t, j) = [\tilde{\pmb{\mathrm{I}}}(t, j), \mathbf{p}(t, j), \tilde{\mathbf{p}}(t, j)]$ and $H(\tilde{\pmb{\mathrm{I}}}(t, j))$ is the global underestimation of $\sum\nolimits_{v\in\bar{\mathcal{U}}(t)} \sum\nolimits_{z=1}^Z (\mathrm{I}_{v,z}(t))^2$.

Note that (\ref{fdRATequivalProb2}) is convex and can be solved efficiently using existing solvers such as CVX \cite{diamond2016cvxpy}.
Particularly, we use Algorithm \ref{rATOptimization} and CVX to solve (\ref{fdRATequivalProb2}) iteratively.
Note that (\ref{fdRATequivalProb2}) has $3(|\bar{\mathcal{U}}(t)| \times Z )$ decision variables and $4(|\bar{\mathcal{U}}(t)| \times Z ) + 2|\bar{\mathcal{U}}(t)| + Z + 1$ constraints.
Besides, in the worst-case, $|\bar{\mathcal{U}}(t)| = Z$.
Moreover, we need to run Algorithm \ref{cvScheduling} to get the required scheduling decisions before solving (\ref{fdRATequivalProb2}).
As such the time complexity of running Algorithm \ref{rATOptimization} for $J$ iterations is $\mathcal{O}\big(27JZ^6 (4Z^2 + 3Z + 1) + |\mathcal{U}(t)|^2 + Z \big)$ \cite{7812683}.
Moreover, our SCA-based solutions of (\ref{localIterGivenEtaPDCase2}) and (\ref{fdRATequivalProb2}) converge to locally optimal solutions of the original problems (\ref{iterPlusCpuFreqPD}) and (\ref{fdRATProbTransform}), respectively \cite{7812683,6816086}.

\begin{algorithm} [t!]
\fontsize{9}{8}\selectfont
\SetAlgoLined 
\DontPrintSemicolon
\KwIn{Initial feasible set $\tilde{\mathbf{x}} (t, j)$, $j=0$, maximum iteration $J$, precision level $\epsilon^{\mathrm{prec}}$, initial penalty $\vartheta_0$}
\nl{\textbf{Repeat}:} \;
\Indp {
    $j \gets j+1$ \;
    $\vartheta \gets \vartheta_0+j$ \;
    Solve (\ref{fdRATequivalProb2}) using $\vartheta$ and $\tilde{\mathbf{x}}(t, j-1)$ to find $\tilde{\mathbf{x}}(t, j)$ \;
    }
\Indm \textbf{Until} converge with $\epsilon^{\mathrm{prec}}$ precision or $j=J$ \;
\KwOut{pRB allocation and power allocation set $\tilde{\mathbf{x}}^{*}(t)$}
\caption{Iterative pRB and Power Allocation Process}
\label{rATOptimization}
\end{algorithm}

\begin{Remark}[VEFL Summary]
    At the beginning of a VEFL round $k$, each CV reports its $\eta_v^{\mathrm{min}}$, $\eta_v^{\mathrm{max}}$, $P_v^{\text{max}}$, $\mathrm{e}^{\mathrm{bud}}_{v}(k)$, $\mathrm{D}_v$ and $\xi_v (l_v(k), \eta_v(k))$ to the server as part of the SLA.
    The server solves (\ref{localIterGivenEtaPDCase2}) to jointly optimize $\pmb{1} (k)$, $\pmb{l} (k)$ and $\pmb{\eta} (k)$. 
    It then sends the most recent global model $\pmb{\omega}_k$, $l_v^{*}(k)$'s and $\eta_v^{*}(k)$'s to the selected CVs.
    Each CV uses the optimized $\eta_v^{*}(k)$ to train the received ML model for $l_v^{*}(k)$ local rounds on its local dataset.
    The server and the CVs work together to receive the CVs' locally trained models $\pmb{\omega}_{v,k+1}$'s. 
    More specifically, the server keeps track of the remaining payload buffer $\mathbf{Q}(t)$, which essentially tells it how much of $\pmb{\omega}_{v,k+1}$'s are yet to be received.
    Recall that the server estimated the worst-case required number of offloading time slots $\bar{\mathrm{t}}_v^{\mathrm{tx}}$ and determined the VEN slot $\tau(k)$ from which it shall start receiving the $\pmb{\omega}_{v,k+1}$'s. 
    The server makes the payload offloading scheduling decisions based on Algorithm \ref{cvScheduling} from this slot $\tau(k)$.
    It then solves (\ref{fdRATProbTransform_Special1CV}) to optimize the transmission power if only one CV is scheduled. 
    Otherwise, it solves (\ref{fdRATequivalProb2}) to optimize the pRB and power allocation jointly.
    The scheduled CVs then receives their $\mathrm{I}_{v,z}^{*}(t)$'s and $P_{v,z}^{*}(t)$'s from the server, and continue offloading their $\pmb{\omega}_{v,k+1}$'s.
    To that end, the server updates the global model based on its aggregation rule defined in (\ref{aggregationRulePartialNORep}) and repeats the above processes for the next round.
\end{Remark}

\begin{figure*} \vspace{-0.05 in}
\begin{minipage}{0.33\textwidth}
	\centering
	\includegraphics[trim=12 5 15 10, clip, width=\textwidth, height=0.15\textheight]{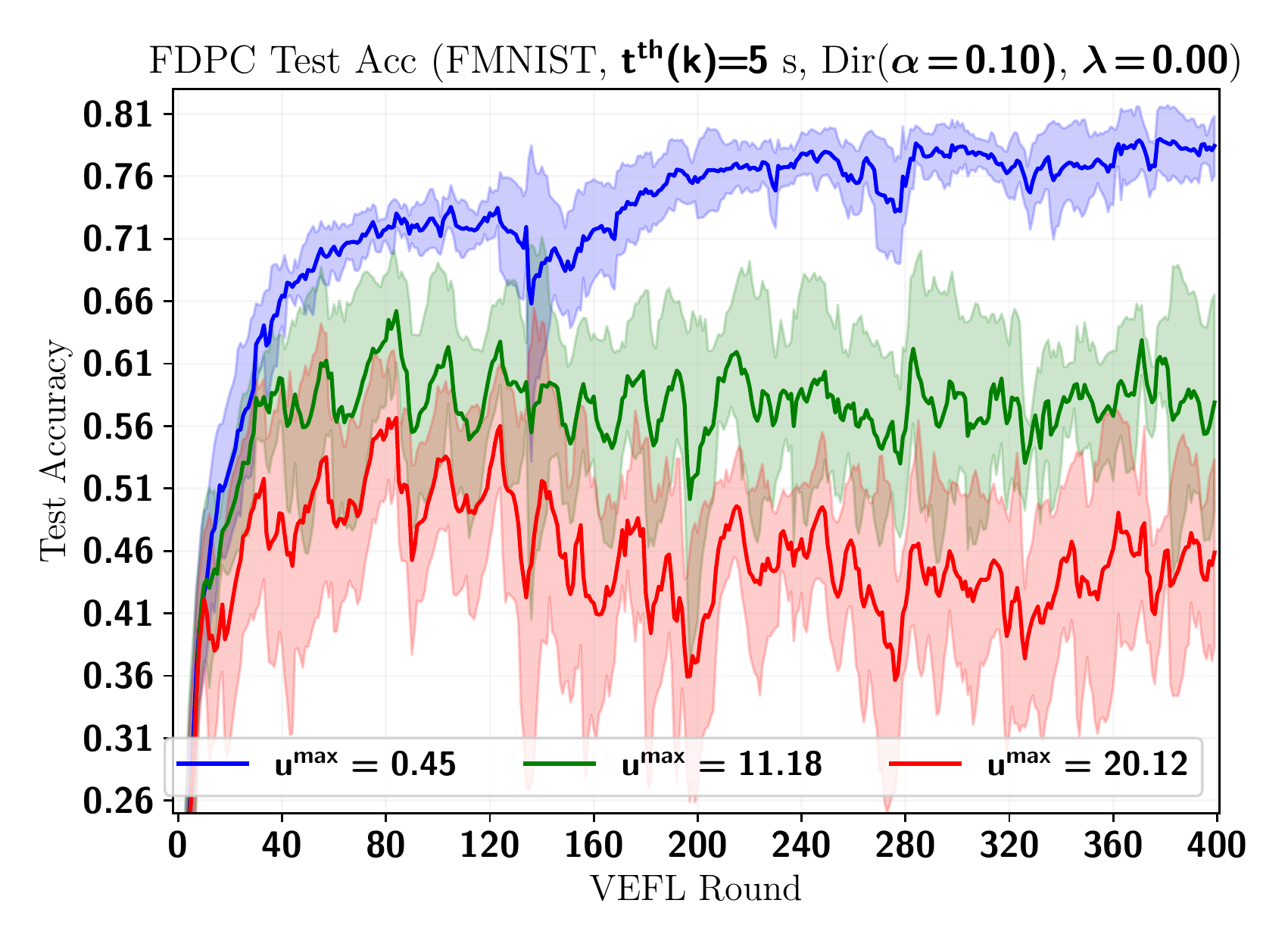} \vspace{-0.15 in}
	\caption{FDPC - impact of CV velocity: $\mathrm{t}^{\mathrm{th}}(k)=5$, Dir($\alpha=0.1$), $\lambda=0$}
	\label{veloImpact_FMNIST_alpha_0_1}
\end{minipage} \hspace{0.00001in}
\begin{minipage}{0.33\textwidth}
	\centering
	\includegraphics[trim=12 5 8 10, clip, width=\textwidth, height=0.15\textheight]{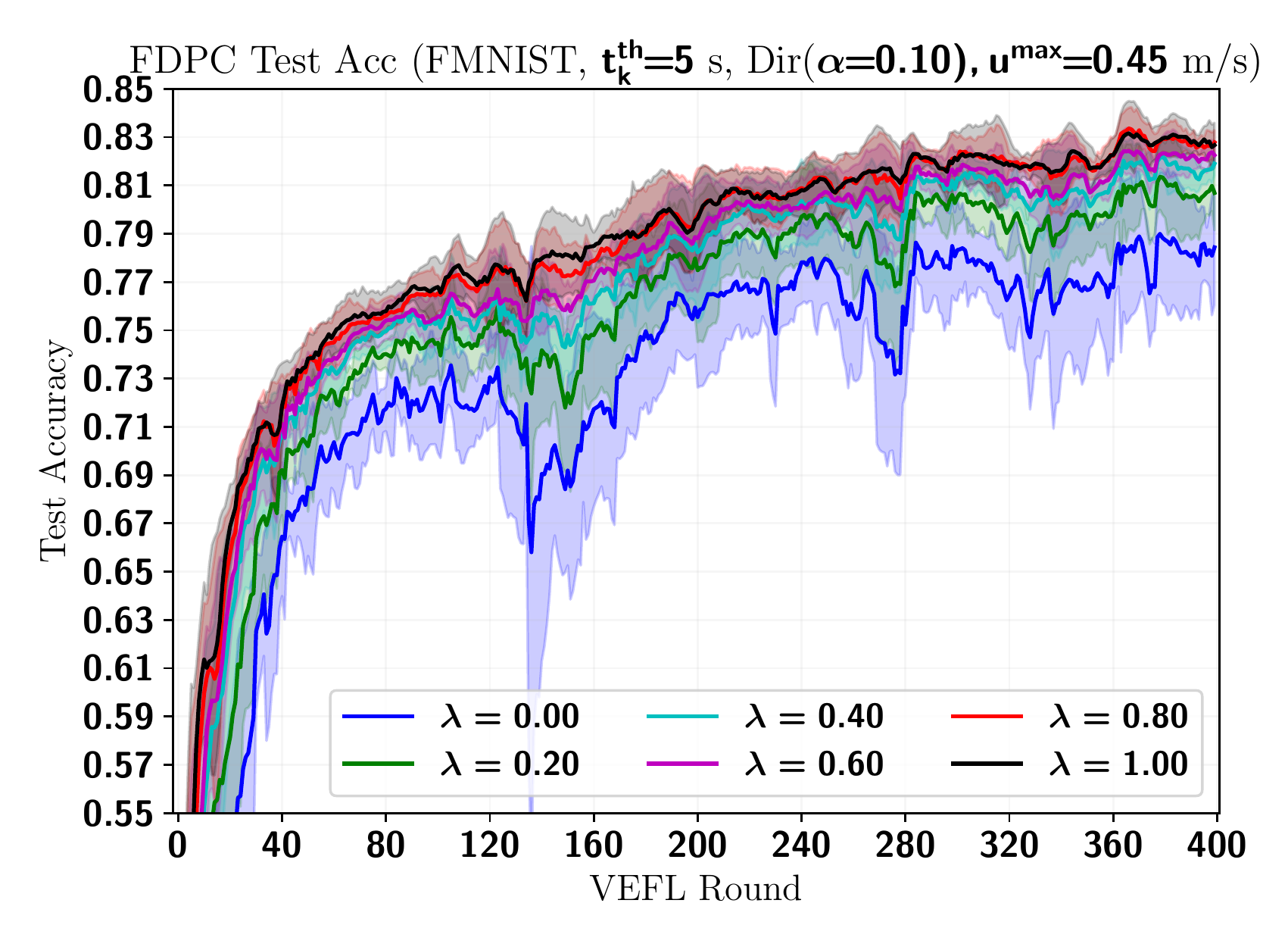} \vspace{-0.15 in}
	\caption{FDPC - impact of $\lambda$: $\mathrm{t}^{\mathrm{th}}(k)=5$, Dir($\alpha=0.1$), $\mathrm{u}^{\mathrm{max}}=0.45$ m/s}
	\label{sojAlphaImpact_FMNIST_alpha_0_1_velo_0_4_5}
\end{minipage} \hspace{0.00001in}
\begin{minipage}{0.33\textwidth}
	\centering
	\includegraphics[trim=12 5 8 10, clip, width=\textwidth, height=0.15\textheight]{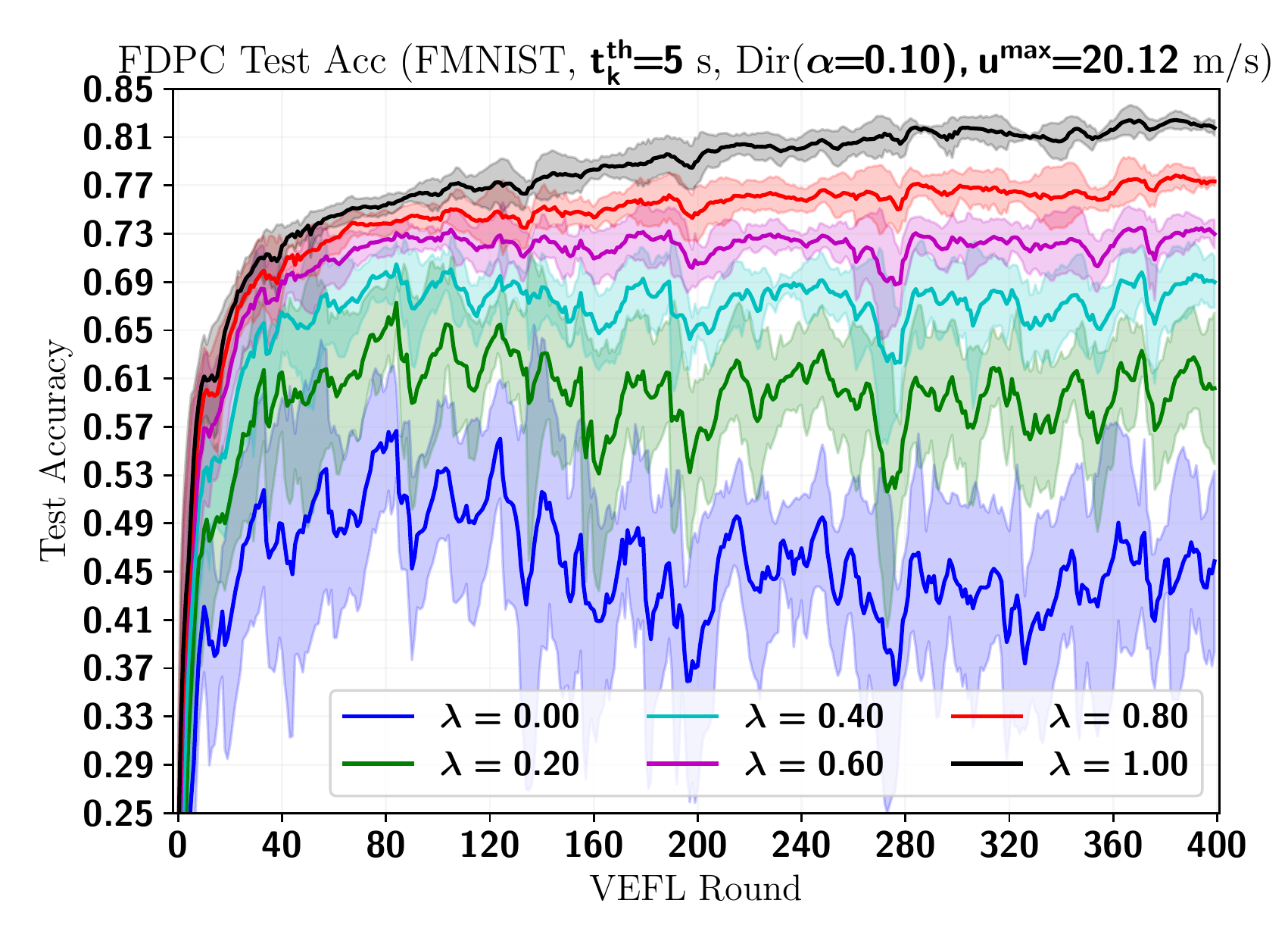} \vspace{-0.15 in}
	\caption{FDPC - impact of $\lambda$:  $\mathrm{t}^{\mathrm{th}}(k)=5$, Dir($\alpha=0.1$), $\mathrm{u}^{\mathrm{max}}=20.12$ m/s}
	\label{sojAlphaImpact_FMNIST_alpha_0_1_velo_20_1_2}
\end{minipage} \hspace{0.00001in}
\begin{minipage}{0.33\textwidth}
	\centering
	\includegraphics[trim=12 5 12 10, clip, width=\textwidth, height=0.15\textheight]{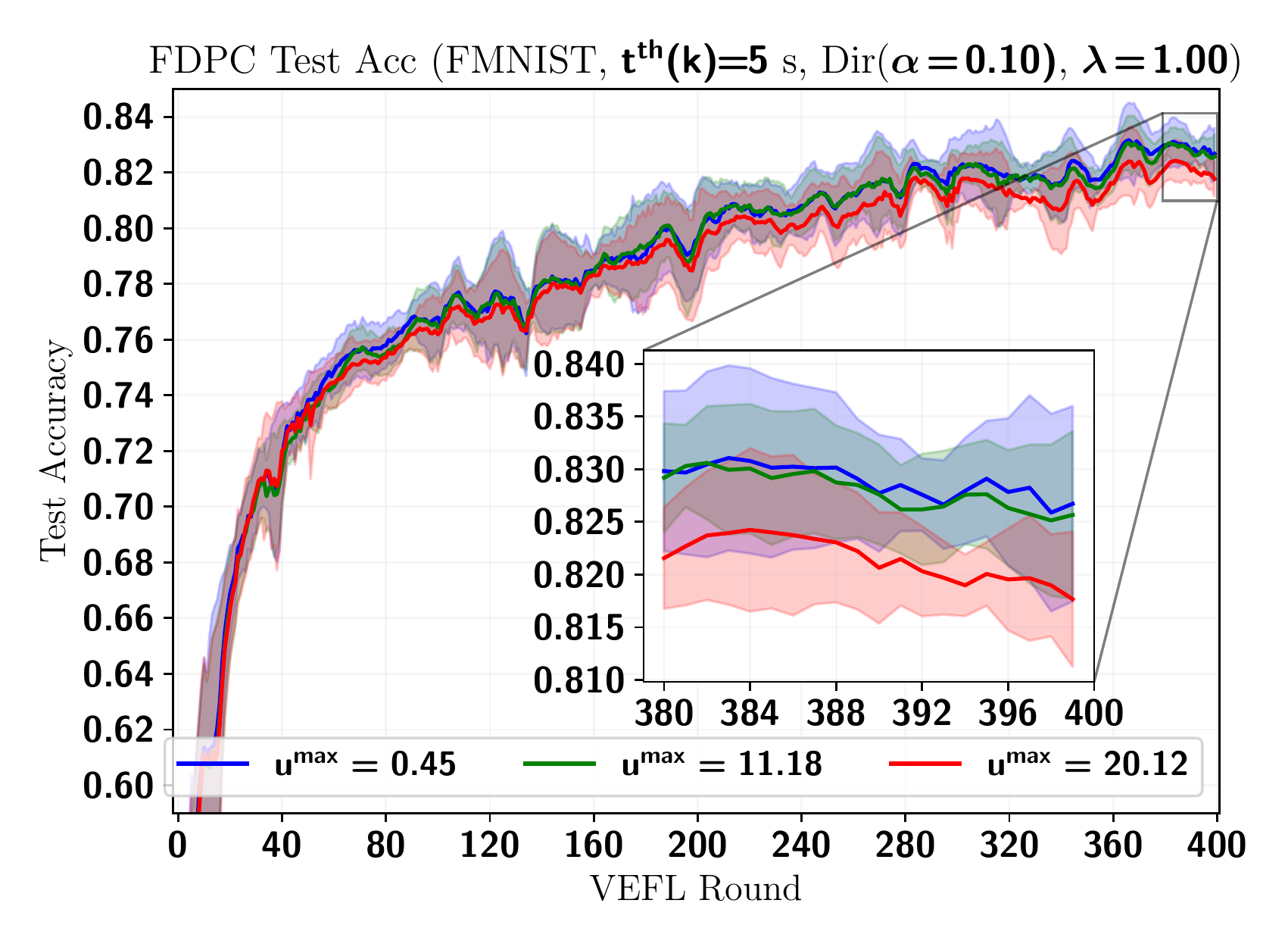} \vspace{-0.15 in}
	\caption{FDPC - performance when $\lambda=1$: $\mathrm{t}^{\mathrm{th}}(k)=5$, Dir($\alpha=0.1$)}
	\label{sojAlphaImpact_FMNIST_alpha_0_1_velo_20_12_SojAlp_1}
\end{minipage} \hspace{0.00001in}
\begin{minipage}{0.33\textwidth} \hspace{0.001 in}
\centering
	\includegraphics[trim=8 2 12 10, clip, width=\textwidth, height=0.15\textheight]{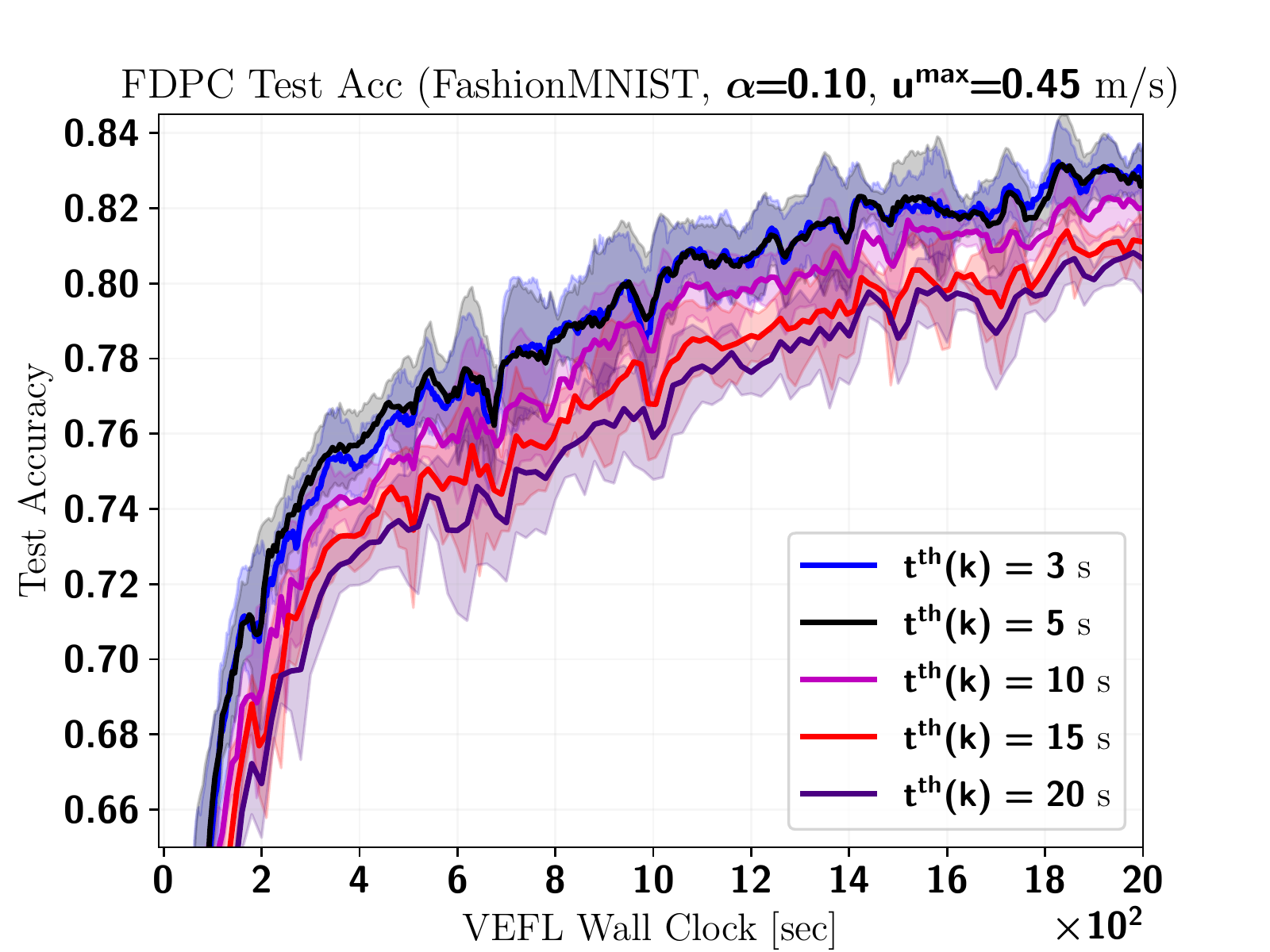} \vspace{-0.15 in}
	\caption{FDPC - Impact of $\mathrm{t}^{\mathrm{th}}(k)$: Dir$(\alpha=0.1)$, $\lambda=1$, $\mathrm{u}^{\mathrm{max}}=0.45$ m/s}
	\label{deadlineImpact_FMNIST_FDPC_velo_0_45}
\end{minipage} \hspace{0.00001in}
\begin{minipage}{0.33\textwidth}
\centering
	\includegraphics[trim=8 2 12 10, clip, width=\textwidth, height=0.15\textheight]{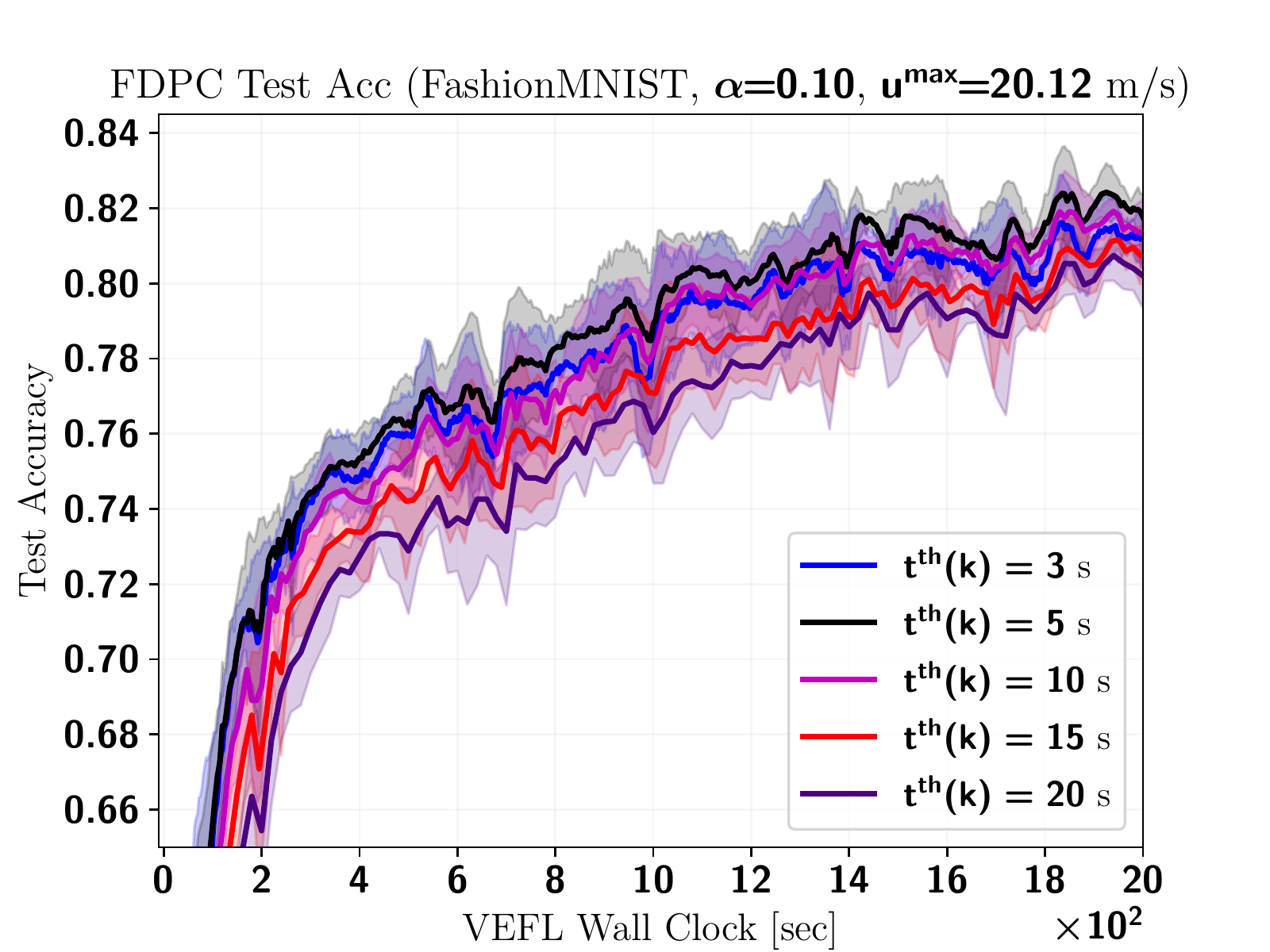} \vspace{-0.15 in}
	\caption{FDPC - Impact of $\mathrm{t}^{\mathrm{th}}(k)$: Dir$(\alpha=0.1)$, $\lambda=1$, $\mathrm{u}^{\mathrm{max}}=20.12$ m/s}
	\label{deadlineImpact_FMNIST_FDPC_velo_20_12}
\end{minipage}
\end{figure*}

\section{Simulation Results and Discussions}
\label{simulationResults_Section}

\subsection{Simulation Setting}
\noindent
For our RoI, we use the real-world map information of downtown Raleigh, NC, USA, from OpenStreetMap\footnote{\url{https://www.openstreetmap.org}}.
We then use SUMO to model practical microscopic mobility as presented in Section \ref{CVs_Mobility_Model}.
We repeat our simulation process for $5$ times to get the average performance.
The VEN is simulated for $2000$ seconds in each of these repeats.
There are $293$, $297$, $299$, $292$ and $287$ unique CVs, respectively in these $5$ repeats.
Besides, the total CVs present in the VEN slots varies. 
Each CV's minimum and maximum CPU frequencies are selected randomly from $[1\times 10^3, 5\times 10^3]$ Hz and $[1.9\times 10^9, 2.8\times 10^9]$ Hz, respectively.
The energy budget, required CPU cycle to process per bit data, per unit energy cost and $\bar{\phi}_v$s are randomly chosen from $[20,30]$ Joules, $[20, 30]$, $[5, 10]$ units and $[10,20]$ units, respectively.
For the RAT parameters, we consider $r=500$ meter, $Z=10$, $\omega=1.8$ MHz\footnote{With numerology $1$, the pRB size is 360 KHz. However, for the ease of faster simulation, we considered that the pRB size is $5 \times 360$ KHz.}, $\bar{n}=1$, subcarrier spacing $30$ KHz and $\kappa=0.5$ ms.
We assume the gNB has $4$ antennas, whereas the CVs have single antennas. 
For model training and testing, we use a) MNIST \cite{lecun2010mnist}, b) FashionMNIST \cite{xiao2017FMINST} c) German Traffic Sign Recognition Benchmark (GTSRB) \cite{GTSRB2011} and d) CIFAR-$10$ \cite{krizhevsky2009learning} datasets.
We use a simple convolutional neural network (CNN) with two convolutional layers, each followed by max-pooling layers, one fully connected layer and the output layer.
For the non-IID data distribution, we use symmetric Dirichlet distribution $\mathrm{Dir}(\alpha)$ with the concentration parameter $\alpha$ and use a similar approach as in \cite{richeng2022communication} to distribute the labels across the unique CVs.
Note that a smaller $\alpha$ means the label distribution across CVs is more skewed.

To this end, we study the impact of different system parameters using the FashionMNIST dataset for FDPC and PDPC. 
We will also use the other three datasets to compare the performances later in Section \ref{simResultsComparisons}.
Note that we have omitted the results of the VEFL and RAT parameter optimizations for brevity.

\subsection{Simulation Results: FDPC}

\begin{figure*}\vspace{-0.05in}
\begin{minipage}{0.33\textwidth}
\centering
	\includegraphics[trim=12 5 12 10, clip, width=\textwidth, height=0.15\textheight]{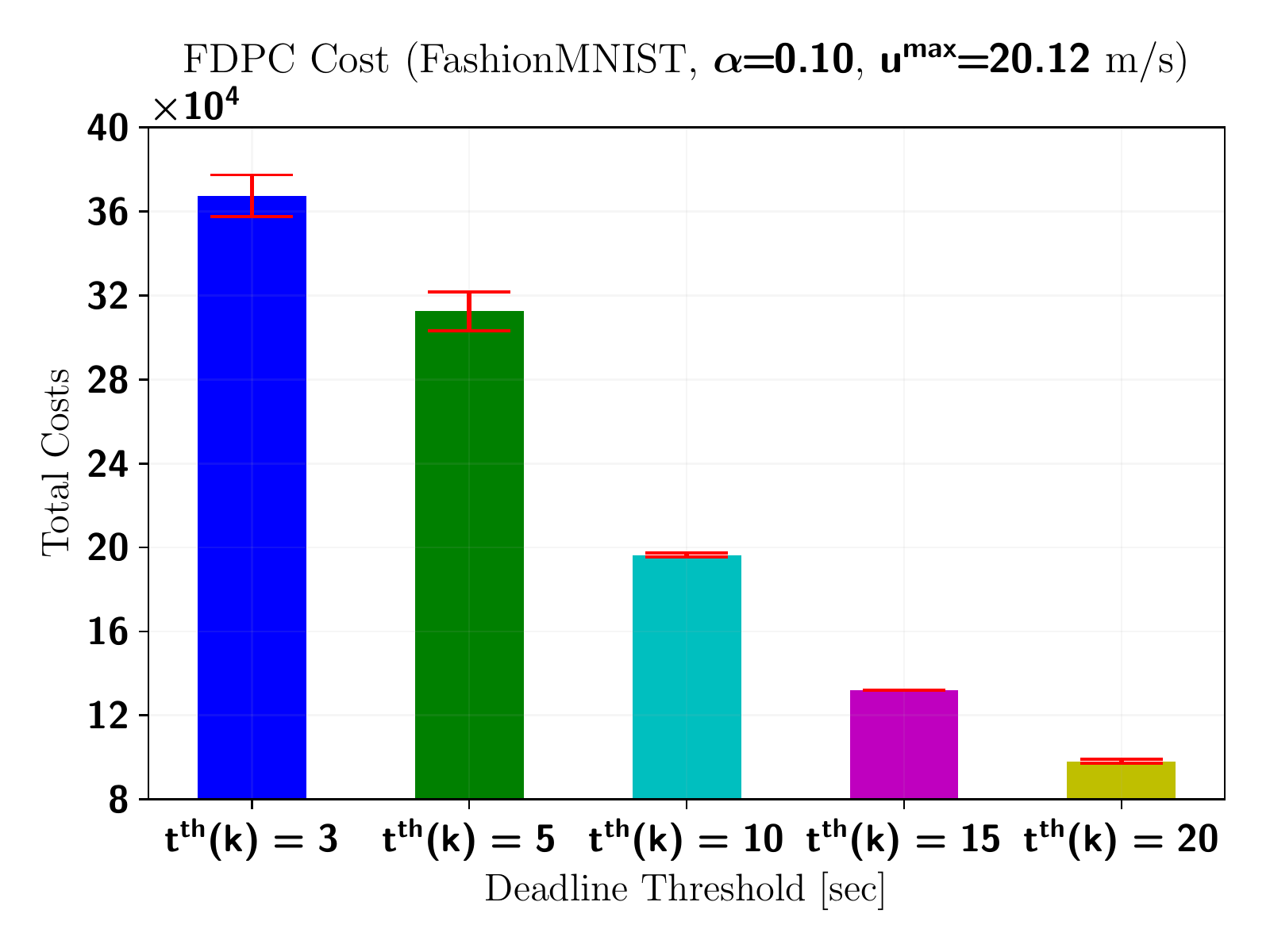} \vspace{-0.15 in}
	\caption{FDPC: total costs for different $\mathrm{t}^{\mathrm{th}}(k)$}
	\label{deadlineVsCost_FMINST_FDPC_velo_20_12}
\end{minipage} \hspace{0.00001in}
\begin{minipage}{0.33\textwidth}
\centering
	\includegraphics[trim=12 5 12 10, clip, width=\textwidth, height=0.15\textheight]{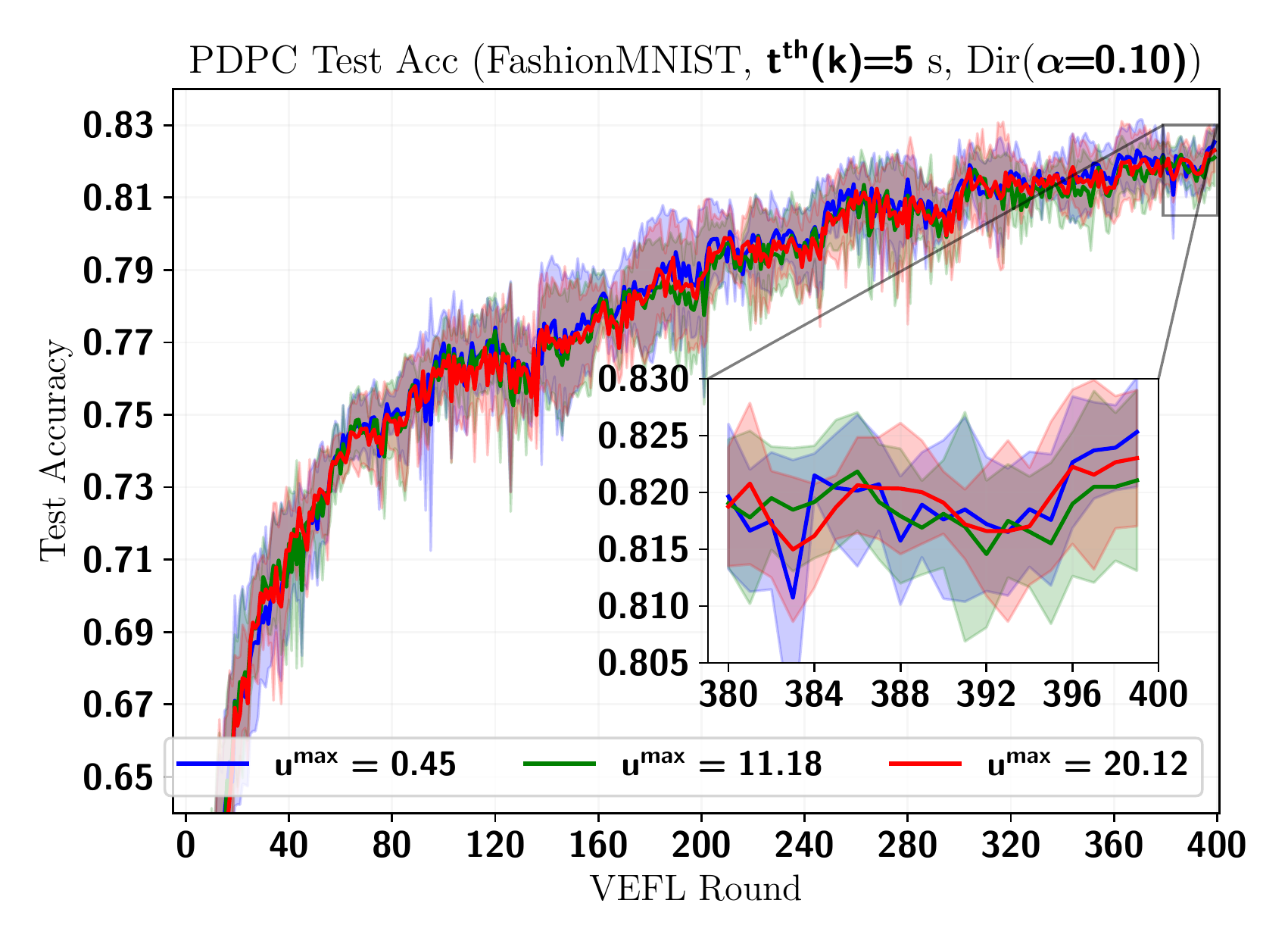} \vspace{-0.15 in}
	\caption{PDPC: Impact of CV velocity: $\mathrm{t}^{\mathrm{th}}(k)=5$, Dir($\alpha=0.1$)}
	\label{veloImpact_PDPC_FMNIST_alpha0_10}
\end{minipage} 
\begin{minipage}{0.33\textwidth}
\centering
	\includegraphics[trim=12 5 12 10, clip, width=\textwidth, height=0.15\textheight]{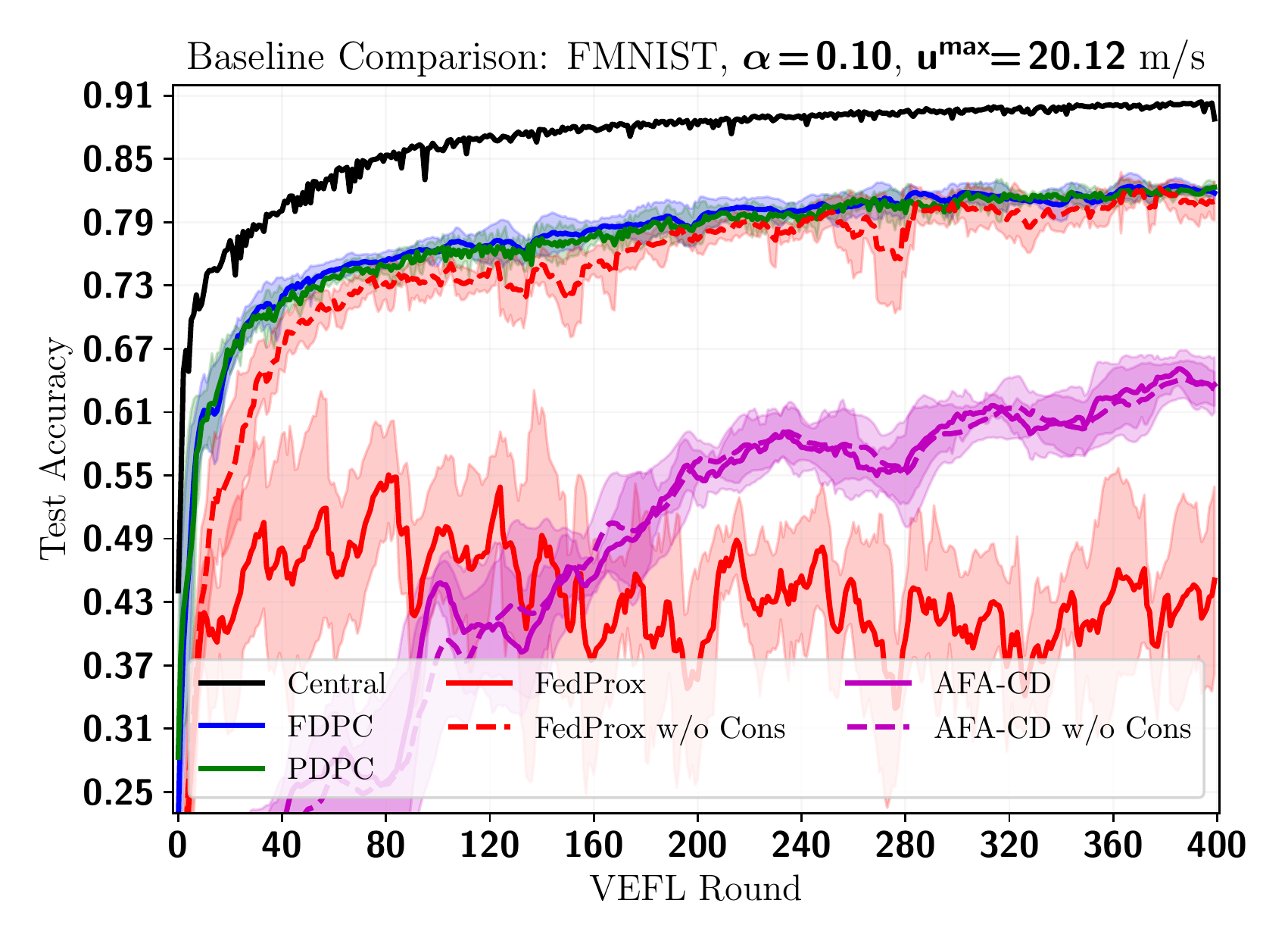} \vspace{-0.15 in}
	\caption{Performance comparison: $\mathrm{t}^{\mathrm{th}}(k)=5$s, $\mathrm{u}^{\mathrm{max}}=20.12$m/s, Dir($\alpha=0.1$)}
	\label{performanceCompFMNIST_alpha0_10_velo_20_12}
\end{minipage} \hspace{0.001 in}
\begin{minipage}{0.33\textwidth}
\centering
	\includegraphics[trim=12 5 12 10, clip, width=\textwidth, height=0.15\textheight]{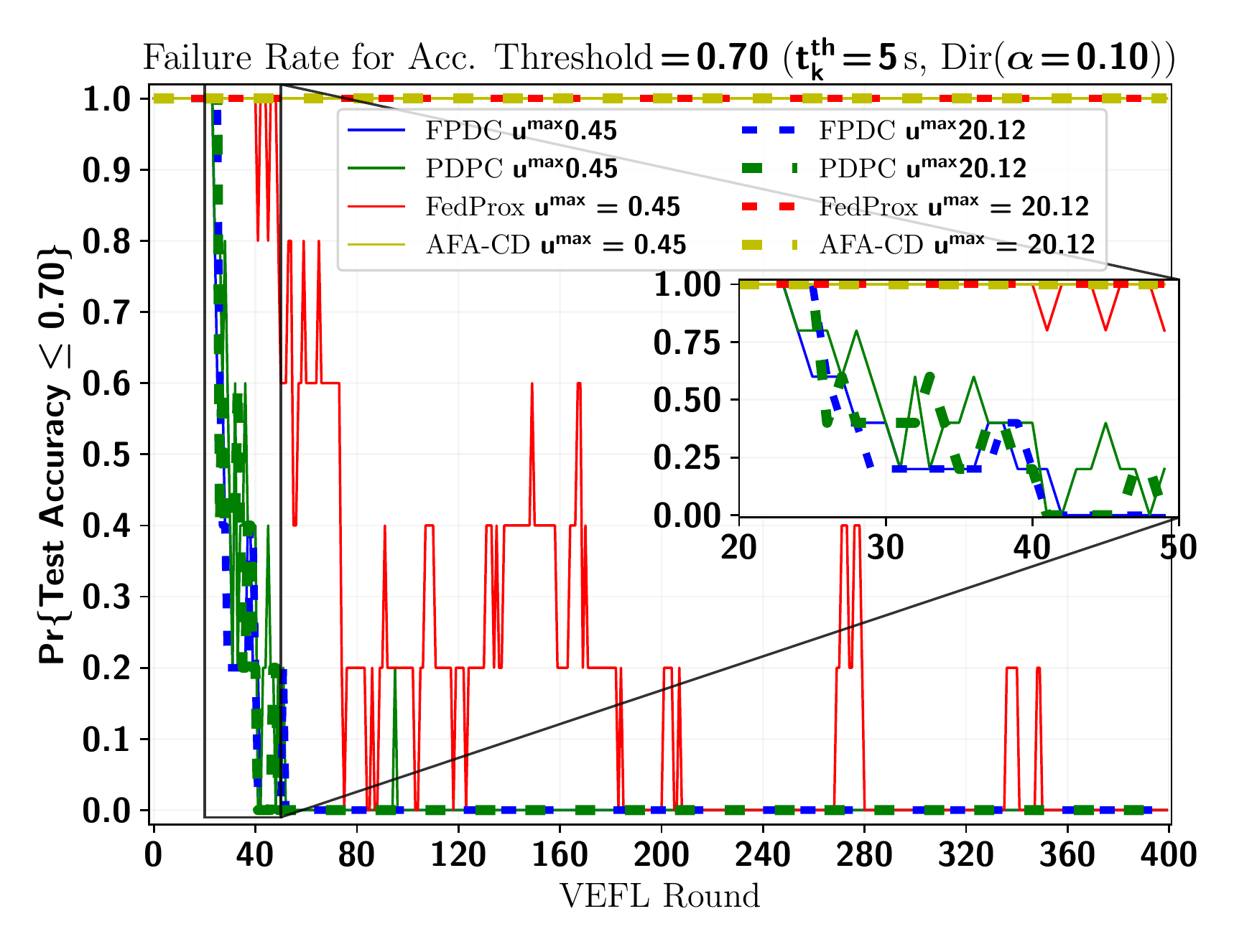} \vspace{-0.15 in}
	\caption{Failure rate of achieving $70\%$ test accuracy: $\mathrm{t}^{\mathrm{th}}(k)=5$ s, $\mathrm{u}^{\mathrm{max}}=0.45$ m/s, Dir($\alpha=0.1$)}
	\label{performanceCompFMNIST_alpha0_10_velo_0_45}
\end{minipage}\hspace{0.0001 in}
\begin{minipage}{0.33\textwidth}
\centering
	\includegraphics[trim=12 5 12 10, clip, width=\textwidth, height=0.15\textheight]{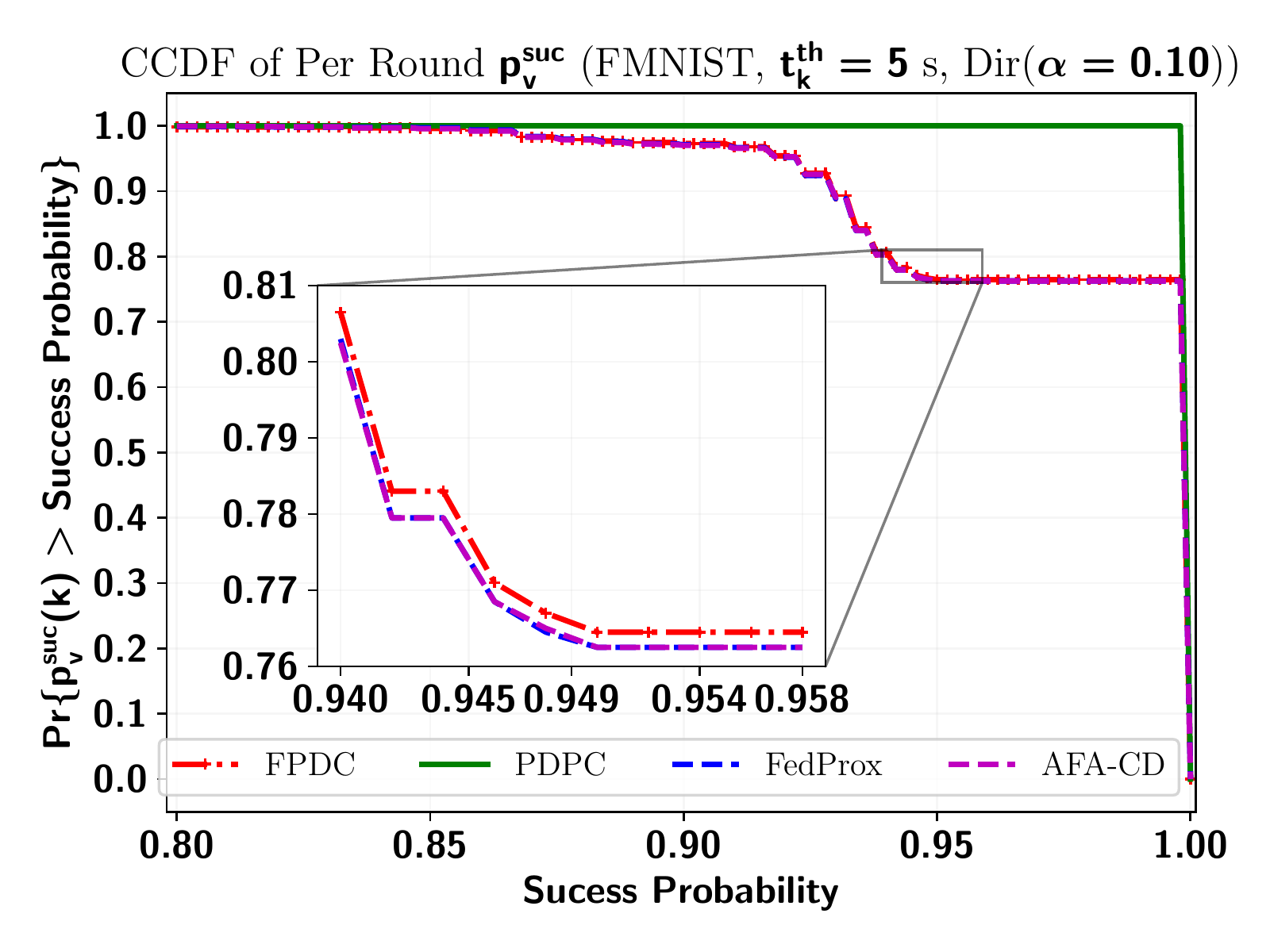} \vspace{-0.15 in}
	\caption{CCDF of success probability comparison: $\mathrm{t}^{\mathrm{th}}(k)=5$, Dir($\alpha=0.1$)}
	\label{performanceCompSuccFMINST_alpha0_10}
\end{minipage} \hspace{0.0001 in}
\begin{minipage}{0.33\textwidth}
\centering
	\includegraphics[trim=12 5 12 10, clip, width=\textwidth, height=0.15\textheight]{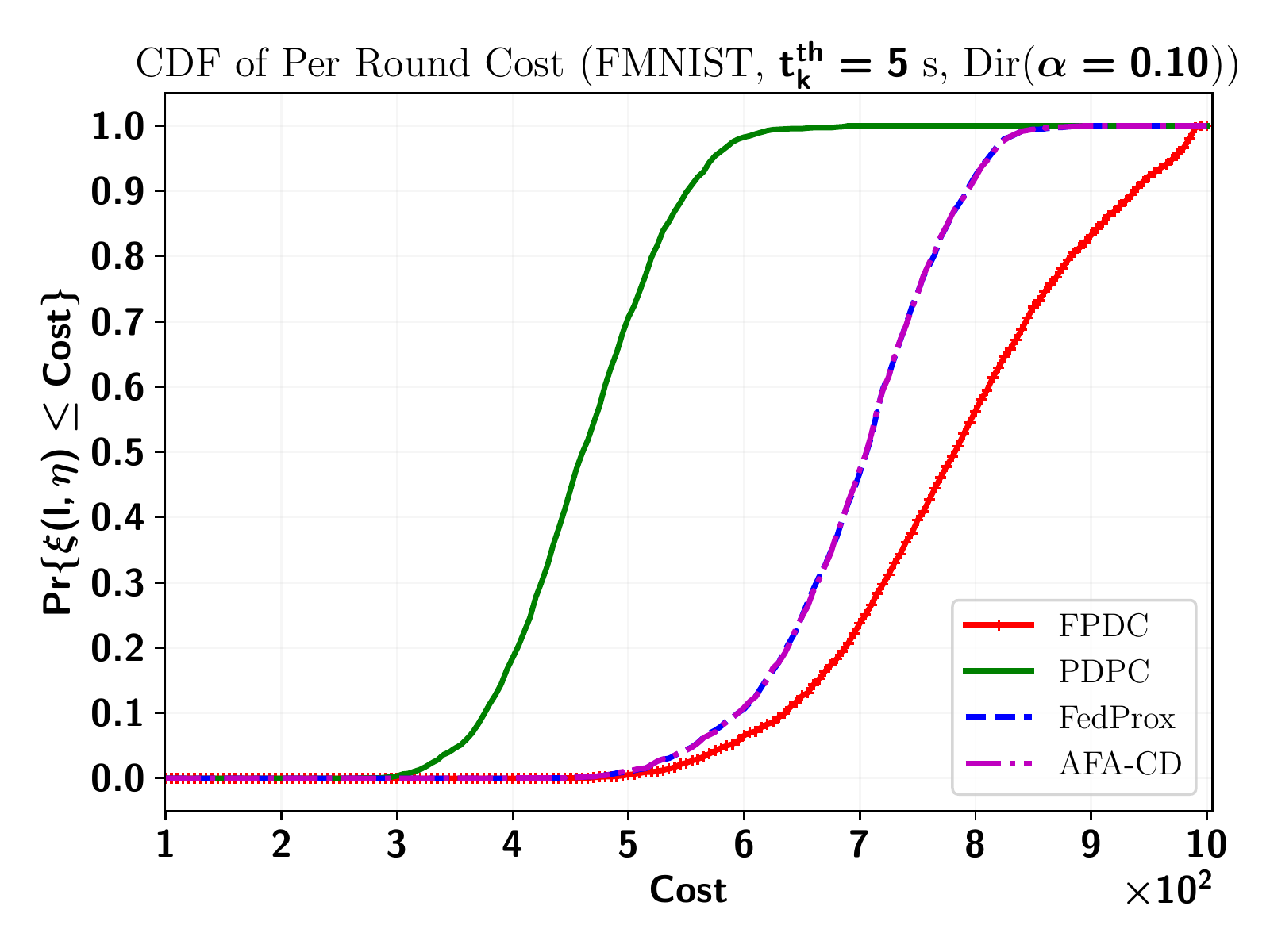} \vspace{-0.15 in}
	\caption{CDF of per round cost comparison: $\mathrm{t}^{\mathrm{th}}(k)=5$, Dir($\alpha=0.1$)}
	\label{performanceCompCostFMNIST_alpha0_10}
\end{minipage}
\end{figure*}

\subsubsection[]{FDPC - Impact of Velocity and $\lambda$}
First, we validate our aggregation rule's necessity for different velocities. 
Note that in (\ref{p_V_Mobility}), if we put zero weight on the estimated sojourn period, then the server essentially aggregates the model weights solely based on the dataset sizes of the CVs.
When the FEEL clients are stationary, it makes sense to aggregate the model parameters based on the dataset sizes. 
However, when the clients are mobile, which is the case for CVs, solely aggregating based on dataset size may not yield the best results. 
This is particularly true when mobility is relatively high. 
The expected sojourn period of the CV decreases as the velocity increases. 
Therefore, a CV has less time to perform its local iterations as it must offload the trained model before moving out of the gNB's coverage area. 
This essentially means that the performance of the trained global model may not yield good test accuracy.
A similar trend is also observed in our simulation results.
Fig. \ref{veloImpact_FMNIST_alpha_0_1} shows how test accuracy varies across the VEFL rounds for different CV velocities for $\alpha=0.1$.
We observe that the performance decreases when $\mathrm{u}^{\mathrm{max}}$ increases.
Note that the solid lines in the figures show the mean values, while the shaded strip is the standard deviation.
These simulation results also reveal that the deviation also increases with increased velocity.

In addition, it is beneficial to put more weight on the expected sojourn period.
Our results in Fig. \ref{sojAlphaImpact_FMNIST_alpha_0_1_velo_0_4_5} and Fig. \ref{sojAlphaImpact_FMNIST_alpha_0_1_velo_20_1_2} also validate this claim.
When $\mathrm{u}^{\mathrm{max}}=0.45$ m/s, we observe that $\lambda=0.8$ and $\lambda=1$ yield nearly identical performance in Fig. \ref{sojAlphaImpact_FMNIST_alpha_0_1_velo_0_4_5}.
Besides, when $\mathrm{u}^{\mathrm{max}}=20.12$ m/s, $\lambda=1$ clearly provides the best test accuracy, which is reflected in Fig. \ref{sojAlphaImpact_FMNIST_alpha_0_1_velo_20_1_2}.
Fig. \ref{sojAlphaImpact_FMNIST_alpha_0_1_velo_20_12_SojAlp_1} verifies that our proposed aggregation policy works for all $\mathrm{u}^{\mathrm{max}}$s.
Furthermore, the standard deviations in the test accuracy are also negligible for all velocities.

\subsubsection[]{FDPC - Impact of $\mathrm{t}^{\mathrm{th}} (k)$}
The impact of the deadline threshold $\mathrm{t}^{\mathrm{th}}(k)$ between two VEFL rounds can also affect the learning performance.
Intuitively, having a smaller $\mathrm{t}^{\mathrm{th}}(k)$ means the server can run more VEFL rounds within a fixed time duration.
Therefore, the trained model's test accuracy is expected to be better with a smaller $\mathrm{t}^{\mathrm{th}}(k)$.
However, although this holds at low velocity, it is not straightforward at a higher $\mathrm{u}^{\mathrm{max}}$ since the short sojourn period may lead to a smaller $l_v(k)$.
Therefore, under high mobility, the server shall carefully decide the deadline threshold $\mathrm{t}^{\mathrm{th}}(k)$ to get a reasonably trained model from the participating CVs. 
Our simulation results also reveal similar trends in Fig. \ref{deadlineImpact_FMNIST_FDPC_velo_0_45} and Fig. \ref{deadlineImpact_FMNIST_FDPC_velo_20_12}.
Particularly, when $\mathrm{u}^{\mathrm{max}}=0.45$ m/s, we observe that $\mathrm{t}^{\mathrm{th}}=3$ s and $\mathrm{t}^{\mathrm{th}}=5$ s yield almost similar performances.
However, at $\mathrm{u}^{\mathrm{max}}=20.12$ m/s, we observe that $\mathrm{t}^{\mathrm{th}}(k)=3$ s clearly does not provide better results over $\mathrm{t}^{\mathrm{th}}(k)=5$ s and $\mathrm{t}^{\mathrm{th}}(k)=10$ s.
Furthermore, a shorter $\mathrm{t}^{\mathrm{th}} (k)$ also causes more VEFL rounds, which causes more expenses for the server.
Fig. \ref{deadlineVsCost_FMINST_FDPC_velo_20_12} shows how the total cost $\sum_{k=1}^K \sum_{v=1}^{V_{\breve{t}_k}} \xi(l_v(k), \eta_v(k))$ vary for different $\mathrm{t}^{\mathrm{th}}(k)$.
Here, we stress that a smaller $\mathrm{t}^{\mathrm{th}}(k)$ can be desirable for the server if it requires a trained global model that delivers a certain level of accuracy quickly at the cost of higher expenses.
As such, depending on its operational needs, it may choose the desired deadline threshold.

\subsection{Simulation Results: PDPC}
\noindent
Note that the length of subset $\mathcal{C}_k$ is a design parameter that the system administrator can choose.
Based on our simulation setting, we found that $|\mathcal{C}_k|=8$ provides the best results for MNIST and GTSRB datasets, while $|\mathcal{C}_k|=10$ and $|\mathcal{C}_k|=12$ work best for the FashionMNIST and CIFAR-$10$ datasets, respectively.
Moreover, our framework is general, where the server can adjust this value in each VEFL round.

\subsubsection{PDPC - Impact of Velocity}
We now validate that our PDPC solution is robust against different velocities.
Recall that we maximize a weighted summation of $l_v(k)$s, where we choose the weight based on expected sojourn periods and dataset sizes of the CVs.
Particularly, since PDPC allows the server to choose the CVs at the beginning of the VEFL rounds, we put equal weights on the sojourn periods and the dataset sizes, i.e., $\bar{\lambda}=0.5$ in (\ref{thetaV}).
Therefore, a proper joint optimization for $\mathbf{1}(v \in \mathcal{C}_k)$, $l_v(k)$ and $\eta_v(k)$ can combat the effect of short sojourn periods under high mobility.
Our simulation results also validate this. 
In Fig. \ref{veloImpact_PDPC_FMNIST_alpha0_10}, we observe that the test accuracies are very similar for $\mathrm{u}^{\mathrm{max}}=0.45$ m/s, $\mathrm{u}^{\mathrm{max}}=11.18$ m/s and $\mathrm{u}^{\mathrm{max}}=20.12$ m/s.
Particularly, when $\mathrm{u}^{\mathrm{max}} = 0.45$ m/s and the VEFL rounds are $k=100$, $k=200$, $k=300$ and $k=400$, the test accuracies are $76.71\%$, $79.19 \%$, $81.17\%$ and $82.53\%$, respectively.
For the same VEFL rounds, when $\mathrm{u}^{\mathrm{max}}=11.18$ m/s, the test accuracies are $76.17\%$, $78.56\%$, $80.67\%$ and $82.11\%$, respectively. 
Besides, when $\mathrm{u}^{\mathrm{max}}=20.12$ m/s, the test accuracies are $76.08\%$, $78.74\%$, $81.12\%$ and $82.3\%$, respectively.

Note that PDPC performs similarly to FDPC with respect to different $\mathrm{t}^{\mathrm{th}}(k)$, which we skip presenting for brevity.

\begin{table*}[!t] \vspace{-0.05in}
\caption{Test Accuracy with Trained $\pmb{\omega}_K$}
\small
\fontsize{7}{10}\selectfont
\centering
\begin{tabular}{|C{0.9cm} |C{0.6cm} | C{1.6cm}| C{1.6cm} | C{2cm} | C{2cm}| C{1.9cm} | C{1.9cm} | C{1.5cm} |}
	\hline 
	\rowcolor{gray!15} \textbf{Dataset} & Dir($\alpha$)& \textbf{FDPC (Proposed)} & \textbf{PDPC (Proposed)} & \textbf{FedProx \cite{MLSYS2020_38af8613} with Constraints} & \textbf{AFA-CD \cite{yang2022anarchic} with Constraints} & \textbf{FedProx \cite{MLSYS2020_38af8613} W/O Constraints}  & \textbf{AFA-CD \cite{yang2022anarchic} W/O Constraints} & \textbf{Centralized ML} \\ \hline
	\multirow{3}{*}{MNIST} &  $0.1$ & $0.9677 \pm 0.0023$  & $\mathbf{0.9691} \pm 0.0024$ & $0.4566 \pm 0.1497$ & $0.7543 \pm 0.112$ & $0.9723 \pm 0.0012$ & $0.8285 \pm 0.0209$ & \multirow{3}{*}{$0.9905$} \\ \cline{2-8}
	& $0.9$ & $0.9779 \pm 0.0022$  & $\mathbf{0.9787} \pm	0.0013$ & $0.8906 \pm 0.0191 $ & $0.8248 \pm 0.0994$ & $0.982 \pm 0.0017$ & $0.8994 \pm 0.0049$ &\\ \cline{2-8}
	& $10$ & $\mathbf{0.9816} \pm 0.0009$  & $\mathbf{0.9816} \pm 0.0006$ & $0.9216 \pm 0.0033$ & $0.8294 \pm 0.1025$ & $0.9838 \pm 0.0011$ & $0.9136 \pm 0.0031$ & \\ \hline
	\multirow{3}{*}{FMNIST} &  $0.1$ & $0.8177 \pm	0.0064$  & $\mathbf{0.823} \pm 0.006$  & $0.4507 \pm 0.0888$ & $0.6359 \pm 0.0262$ & $0.8098 \pm 0.0181$ & $0.6319 \pm 0.0161$ & \multirow{3}{*}{$0.8879$}\\ \cline{2-8}
	& $0.9$ & $\mathbf{0.857} \pm	0.0024$  & $0.8527 \pm 0.0032$  & $0.6457 \pm	0.0251$ & $0.6933 \pm 0.0058$ & $0.8649 \pm 0.0033$ & $0.6943 \pm 0.0136$ &\\ \cline{2-8} 
	& $10$ & $\mathbf{0.8641} \pm 0.0019$  & $0.8526 \pm 0.0033$ & $0.6846 \pm 0.0142$ & $0.7061 \pm 0.003 $ & $0.8744 \pm 0.0021$ & $0.7057 \pm 0.0033$ & \\ \hline
	\multirow{3}{*}{GTSRB} &  $0.1$ & $0.4838 \pm	0.0123$ & $\mathbf{0.488} \pm 0.0049$ & $0.104 \pm 0.026$ & $0.1247 \pm 0.0144 $ & $0.5568 \pm 0.0118$ & $0.1246 \pm 0.0139$  & \multirow{3}{*}{$0.9166$} \\ \cline{2-8}
	& $0.9$ & $0.614 \pm 0.0078$  & $\mathbf{0.6165} \pm 0.0187$ & $0.2045 \pm 0.0191$ & $0.1628 \pm 0.0146$ & $0.6622 \pm 0.0122$ & $0.1592 \pm 0.0106$ & \\ \cline{2-8}
	& $10$ & $0.6159 \pm 0.0164$  & $\mathbf{0.6438} \pm 0.0182$ & $0.2195 \pm 0.0183$ & $0.1602 \pm 0.0144$ & $0.6873 \pm 0.0085$ & $0.1524 \pm 0.0077$ & \\ \hline
    \multirow{3}{*}{CIFAR-$10$} & $0.1$ & $\mathbf{0.5043} \pm 0.0122$ & $0.4966 \pm 0.0081$ & $0.2047 \pm 0.0211$ & $0.256 \pm 0.02141$ & $0.5043 \pm 0.0213$ & $0.2484 \pm 0.0168$ & \multirow{3}{*}{$0.7394$}\\ \cline{2-8}
    & $0.9$ & $0.589 \pm 0.006$ & $\mathbf{0.595} \pm 0.0081 $ & $0.2803 \pm 0.022 $ & $0.3229 \pm 0.0045$ & $0.6124 \pm 0.0129$ & $0.3282 \pm 0.0038$ & \\ \cline{2-8}
    & $10$ & $0.6175 \pm 0.0022$  & $\mathbf{0.6202} \pm 0.009 $ & $0.3305 \pm 0.0097$ & $0.3304 \pm 0.0097 $ & $0.6293 \pm 0.0026$ & $0.3274 \pm 0.0026$&\\ \hline
	\end{tabular}
	\label{performanceComparison}
\end{table*}

\subsection{Performance Comparisons}
\label{simResultsComparisons}
\noindent
We now compare the performances of our proposed FDPC and PDPC VEFL schemes with the state-of-the-art FedProx\cite{MLSYS2020_38af8613} and anarchic federated averaging - cross-device (AFA-CD) \cite{yang2022anarchic} baselines.
For FedProx and AFA-CD, we consider two cases, namely, (1) all constraints are present, and (2) no constraints are enforced.
Note that the latter is the ideal case that does not consider any delay, energy, monetary or radio resource constraints.
Besides, for the baselines with system constraints, we assume that the server expends an equal amount of its budget for all CVs. 
Furthermore, taking the maximum CPU frequencies $\eta_v^{\mathrm{max}}$'s, the $l_v(k)$'s are chosen to satisfy the time and energy constraints.
Moreover, we use the same RAT solution for offloading the CVs' trained models.
Note that the baselines with constraints are expected to perform worse when the velocity increases since they do not consider mobility in the model weights aggregation.
We also compare the results with the centralized ML, which serves as the performance upper bound as the server can access all training data.

To that end, we compare these schemes with $\mathrm{u}^{\mathrm{max}}=20.12$ m/s in Fig. \ref{performanceCompFMNIST_alpha0_10_velo_20_12}.
When VEFL rounds are $k=200$, $k=300$ and $k=400$, the test accuracies with FDPC are $78.91\%$, $81.43\%$ and $81.77 \%$, respectively.
On the other hand, PDPC delivers $78.74\%$, $81.12\%$ and $82.30\%$ test accuracies for the same respective VEFL rounds.
By contrast, for the same VEFL rounds, when all constraints are present, FedProx returns $35.89\%$, $39.94\%$ and $45.07\%$, while AFA-CD yields $55.03\%$, $60.51\%$ and $63.59 \%$ test accuracies.
Moreover, in the ideal case, i.e., without (W/O) constraints, FedProx delivers $77.75\%$, $80.09\%$ and $80.98\%$, and AFA-CD returns $56.26\%$, $59.04\%$ and $63.19\%$ test accuracies. 
These results suggest that our proposed FDPC and PDPC schemes significantly outperform the baselines.
It is worth noting that FedProx W/O constraints performs similarly to our solutions with constraints, while the AFA-CD fails to achieve good performance. 
This is due to the fact that FedProx also incorporates a similar local loss function for the clients, whereas AFA-CD does not incorporate any proximal terms and employs a different aggregation rule.
Besides, the gap with the centralized ML is expected due to the inherent nature of the system model. 
More specifically, all training data samples are distributed among the CVs without repetitions. 
As such, once a CV moves out of the communication area, its training data samples are no longer available.
Moreover, these CVs may only perform a few local rounds before moving out of the gNB's coverage.

Besides, unlike in the ideal case, under all system constraints, FedProx suffers from potential model divergence.
In such a constrained case, notice that the best test accuracy of $55.06\%$ is achieved with FedProx when $k=81$.
However, $\pmb{\omega}_k$ does not converge and fluctuates rapidly, as shown in Fig. \ref{performanceCompFMNIST_alpha0_10_velo_20_12}.
It is even more evident in Fig. \ref{performanceCompFMNIST_alpha0_10_velo_0_45}, which shows the failure rate of achieving a test accuracy of $70\%$ for different $\mathrm{u}^{\mathrm{max}}$.
We can see that FDPC achieves a higher test accuracy than the accuracy threshold for all velocities after about $51$ VEFL rounds with probability $1$.
On the other hand, we notice a few oscillations with PDPC. 
Particularly, the probability of having a higher test accuracy than the threshold becomes one after about $51$ VEFL rounds, notwithstanding a $20\%$ drop at VEFL round $k=95$.
However, from $k=96$ onward, PDPC delivers a higher test accuracy than the threshold with probability $1$.
Contrary to these guaranteed performances with FDPC and PDPC, the performance of FedProx with constraints is unreliable.
Even at low mobility, $\mathrm{u}^{\mathrm{max}}=0.45$ m/s, we observe that FedProx's performance oscillates.
This baseline shows that the test accuracy can be less than the threshold $20\%$ of the time, even at VEFL round $k=349$ with relatively low mobile CVs.
Furthermore, AFA-CD fails to achieve the desired $70\%$ accuracy.

Although PDPC yields slightly higher performance oscillations than FDPC initially, it is still a practical solution in resource-constrained circumstances.
To further study its benefits, we now examine the CCDF of $\mathrm{p}_v^{\mathrm{suc}}(k)$.
Recall that for FDPC, FedProx and AFA-CD, all CVs with SLAs receive the global model and participate in the training process.
However, due to mobility, some CVs may not finish even a single local iteration within the allocated deadline. 
Besides, the energy constraint can also potentially restrict some CVs from training and offloading their models.
On the other hand, in PDPC, this can be mitigated by appropriately selecting the subset $\mathcal{C}_k$.
In any case, if the server does not receive a CV's local model within $\mathrm{t}^{\mathrm{th}}(k)$, it considers that as a failure. 
Our simulation results in Fig. \ref{performanceCompSuccFMINST_alpha0_10} suggest that in all VEFL rounds, a CV successfully offloads its locally trained model to the server with at least $0.9$ probability $97.3\%$, $100\%$, $97.1\%$ and $97\%$ percent of the times for FDPC, PDPC, FedProx and AFA-CD, respectively.
Moreover, a CV has a success probability of at least $0.99$ for $76.45\%$ and $100\%$ of the VEFL rounds, respectively, in FDPC, PDPC, and $76.25 \%$ of the VEFL rounds for FedProx and AFA-CD, with all constraints present.

PDPC not only guarantees a $100\%$ trained models reception but also yields a lesser expense for the server.
This is due to the fact that the server only needs to pay the fees to the selected CVs in the subset $\mathcal{C}_k$.
Our simulation results in Fig. \ref{performanceCompCostFMNIST_alpha0_10} show that, in a VEFL round, the server's cost is less than $700$ units for about $23.75\%$, $100\%$, $46.8\%$ and $47.3\%$ of the times, respectively, for FDPC, PDPC, FedProx and AFA-CD.
Note that PDPC not only requires a lesser expense but also delivers near identical test accuracies to FDPC.

Table \ref{performanceComparison} shows the performance comparisons on MNIST, FashionMNIST, GTSRB and CIFAR-$10$ datasets with the obtained global model after $K=400$ VEFL rounds when $\mathrm{u}^{\mathrm{max}}=20.12$ m/s, $\mathrm{t}^{\mathrm{th}}=5$ s and $\Xi(k)=1000$ units.
Under the resource-constrained scenario, it can be observed that both FDPC and PDPC perform similarly while the baselines lag significantly.
Furthermore, even without constraints, AFA-CD fails to perform reasonably in GTSRB and CIFAR-$10$ datasets.
Moreover, the proposed schemes achieve performance comparable to FedProx W/O constraints, which validates the effectiveness of our proposed solutions under extreme resource constraints.

\section{Conclusion}
\label{conclusion}
\noindent
A novel vehicular edge federated learning framework that leverages a $5$G-NR RAT and the limited computation powers of the moving CVs is proposed in this work.
Under delay, energy and cost constraints, first, subset CV selection, local iterations and CPU frequencies are jointly optimized given the estimated worst-case sojourn period and communication delay and cost.
Then, the RAT parameters are jointly optimized using an online per-slot-based stochastic optimization technique.
Extensive simulation results suggest that the server can combat high mobility by aggregating the trained model parameters using a weighted combination of sojourn periods and dataset sizes in FDPC, whereas by appropriately selecting a subset of CVs in PDPC.
Moreover, the proposed method outperforms the state-of-the-art FedProx and AFA-CD algorithms under the same constraints in all the examined scenarios.

\appendices
\section{Proof of Theorem \ref{theorem2}} 
\label{proofTheorem2}

\noindent
Using Taylor expansion, we can write the following:
\begin{align}
\label{lossTayLorExapand}
	& f(\pmb{\omega}_{k+1}) = f(\pmb{\omega}_k) + \left<\nabla f(\pmb{\omega}_k), \pmb{\omega}_{k+1}-\pmb{\omega}_k \right> + \nonumber \\
	&\qquad \frac{1}{2!} \left( \pmb{\omega}_{k+1}-\pmb{\omega}_k\right)^T \nabla^2 f(\pmb{\omega}_k) \left(\pmb{\omega}_{k+1}-\pmb{\omega}_k\right) + h(o), \\
	& \overset{(a)}{\leq} f(\pmb{\omega}_k) + \underbrace{\left<\nabla f(\pmb{\omega}_k), \pmb{\omega}_{k+1}-\pmb{\omega}_k \right>}_\text{Term $1$} + \underbrace{\frac{L}{2} \left \Vert \pmb{\omega}_{k+1}-\pmb{\omega}_k \right\Vert^2}_\text{Term $2$}, \nonumber
\end{align}
where $h(o)$ represents higher-order terms.
We reach to ($a$) by ignoring $h(o)$ and using \textit{Assumption} $1$.

\textbf{Bound Term $2$}: 
Let $\hat{\pmb{\omega}}_{v,k+1} = \underset{\pmb{\omega}}{\text{arg min }} f_v({\pmb{\omega}, \pmb{\omega}_k})$. 
Then, due to $\mu'$-strong convexity of $f_v({\pmb{\omega}, \pmb{\omega}_k})$, we can write
\begin{align} 
\label{eq1}
		\rs \rs & \left\Vert \hat{\pmb{\omega}}_{v,k+1} - \pmb{\omega}_{v,k+1} \right \Vert  \leq \frac{1}{\mu'} \left\Vert\nabla f_v(\hat{\pmb{\omega}}_{v,k+1}, \pmb{\omega}_k) - \nabla f_v(\pmb{\omega}_{v,k+1}\rs, \pmb{\omega}_k) \right\Vert\rs,\rs\rs\rs \nonumber\\
		& = (1/\mu') \left\Vert \nabla f_v(\pmb{\omega}_{v,k+1}, \pmb{\omega}_k) \right \Vert 
		\overset{(a)}{\leq} (\gamma/\mu') \left\Vert \nabla f_v (\pmb{\omega}_k, \pmb{\omega}_k) \right \Vert, \\
		& = (\gamma/\mu') \left\Vert \nabla F_v(\pmb{\omega}_k) \right\Vert 
		\overset{(b)}{\leq} [(B \gamma)/\mu'] \left\Vert \nabla f(\pmb{\omega}_k) \right\Vert, \nonumber 
\end{align}
where we obtain ($a$) and ($b$) using \textit{Assumption} $4$, i.e., $\gamma$-inexact local solvers and \textit{Assumption} $2$, i.e., $B$-dissimilarity, respectively.
Similarly, we can write 
\begin{equation}
\label{eq2}
    \left\Vert \hat{\pmb{\omega}}_{v,k+1} - \pmb{\omega}_k \right \Vert  \leq (1/\mu') \left\Vert \nabla F_v(\pmb{\omega}_k) \right \Vert = (B/\mu') \left \Vert \nabla f(\pmb{\omega}_k) \right \Vert\rs. \rs\rs 
\end{equation}
Then, using triangle inequality we can write 
\begin{equation}
\label{ind_cl_diff}
    \left \Vert \pmb{\omega}_{v,k+1} - \pmb{\omega}_k \right \Vert  \leq [(B(1+\gamma))/\mu'] \left\Vert \nabla f(\pmb{\omega}_k) \right\Vert.
\end{equation}
Now using (\ref{ind_cl_diff}) and (\ref{aggregationRulePartialNORep}), we get 
\begin{align}
\label{term2Simplified_Partial}
\displaybreak[0]
	& \left \Vert \pmb{\omega}_{k+1} - \pmb{\omega}_k \right \Vert^2 = \bigg \Vert  \sum_{v=1}^{V_{\breve{t}_k}} \bigg[p_v \cdot \nonumber\\
	&\quad \frac{ \mathbf{1}\left(v \in \mathcal{C}_k, \mathrm{t}_v (k) \leq \mathrm{t}^{\mathrm{th}}(k)| d_v(k) \leq \mathrm{r} \right)}{q_v (k) p^{\mathrm{suc}}_{v} (k)} \left( \pmb{\omega}_{v,k+1} - \pmb{\omega}_k \right) \bigg] \bigg \Vert^2 \nonumber\\
	&\overset{(a)}{\leq} \sum_{v=1}^{V_{\breve{t}_k}} p_v \cdot \bigg \Vert \frac{\mathbf{1}\left(v \in \mathcal{C}_k, \mathrm{t}_v (k) \leq \mathrm{t}^{\mathrm{th}}(k)| d_v(k) \leq \mathrm{r} \right)} {q_v(k) p^{\mathrm{suc}}_{v} (k)} \times \nonumber\\
	&\qquad \qquad \qquad \qquad \qquad \qquad \left( \pmb{\omega}_{v,k+1} - \pmb{\omega}_k \right) \bigg \Vert^2 \rs , \rs \\
	&\overset{(b)}{=} \sum_{v=1}^{V_{\breve{t}_k}} p_v \cdot \bigg(\frac{\mathbf{1}\left(v \in \mathcal{C}_k, \mathrm{t}_v (k) \leq \mathrm{t}^{\mathrm{th}}(k)|d_v(k) \leq \mathrm{r} \right)} {q_v(k) p^{\mathrm{suc}}_{v} (k)}\bigg)^2 \rs \rs \times\rs \rs\rs \rs\nonumber\\
	&\qquad\qquad\qquad\qquad\qquad\qquad\qquad \left\Vert \pmb{\omega}_{v,k+1} - \pmb{\omega}_k \right\Vert^2 \rs \rs, \rs\rs \nonumber\\
	&\overset{(c)}{=} \frac{B^2 (1 + \gamma)^2} {{\mu'}^2} \left\Vert \nabla f(\pmb{\omega}_k)\right\Vert^2 \times \nonumber\\ 
	& \qquad \sum_{v=1}^{V_{\breve{t}_k}} p_v \cdot \bigg( \frac{\mathbf{1}\left(v \in \mathcal{C}_k, \mathrm{t}_v (k) \leq \mathrm{t}^{\mathrm{th}}(k) | d_v(k) \leq \mathrm{r} \right)}{q_v(k) p^{\mathrm{suc}}_{v} (k)} \bigg)^2,\nonumber
\end{align} 
where ($a$) follows from Jensen's inequality due to the convexity of $\left\Vert \cdot \right\Vert^2$. 
($b$) stems using the fact that $\frac{\mathbf{1}\left(v \in \mathcal{C}_k, \mathrm{t}_v (k) \leq \mathrm{t}^{\mathrm{th}}(k)| d_v(k) \leq \mathrm{r} \right)} {q_v(k)p^{\mathrm{suc}}_{v} (k)}$ is scalar. 
Moreover, we get to ($c$) using (\ref{ind_cl_diff}).

\textbf{Bound Term $1$}: 
Using the aggregation rule defined in (\ref{aggregationRulePartialNORep}), we can write
\begin{align}
\label{term1Simplified_NoReplace}
	& \left<\nabla f(\pmb{\omega}_k), \pmb{\omega}_{k+1}-\pmb{\omega}_k \right> = \big<\nabla f(\pmb{\omega}_k), \sum\nolimits_{v=1}^{V_{\breve{t}_k}} \big[p_v \cdot \nonumber\\
	& \frac{\mathbf{1}\left(v \in \mathcal{C}_k, \mathrm{t}_v (k) \leq \mathrm{t}^{\mathrm{th}}(k)| d_v(k) \leq \mathrm{r} \right)}{q_v(k) p^{\mathrm{suc}}_{v} (k)} \left(\pmb{\omega}_{v,k+1} - \pmb{\omega}_k \right) \big] \big>,\rs\rs\nonumber\\
	&= \sum_{v=1}^{V_{\breve{t}_k}} p_v \cdot \frac{\mathbf{1}\left(v \in \mathcal{C}_k, \mathrm{t}_v (k) \leq \mathrm{t}^{\mathrm{th}}(k)|d_v(k) \leq \mathrm{r} \right)}{q_v(k)p^{\mathrm{suc}}_{v} (k)} \times \\
	& \qquad\qquad\qquad\qquad\qquad \left<\nabla f(\pmb{\omega}_k), \pmb{\omega}_{v,k+1} - \pmb{\omega}_k \right>, \nonumber\\
	& \leq B \left(1 + \gamma \right)/(\mu') \left \Vert \nabla f(\pmb{\omega}_k) \right\Vert^2 \sum\nolimits_{v=1}^{V_{\breve{t}_k}} p_v \times \nonumber \\
	&\qquad\qquad\qquad \frac{\mathbf{1}\left(v \in \mathcal{C}_k, \mathrm{t}_v (k) \leq \mathrm{t}^{\mathrm{th}}(k)| d_v(k) \leq \mathrm{r} \right)}{q_v(k) p^{\mathrm{suc}}_{v} (k)}.\nonumber
\end{align}
Then, plugging (\ref{term1Simplified_NoReplace}) and (\ref{term2Simplified_Partial}) in (\ref{lossTayLorExapand}), we get the following
\begin{align}
\label{oneIterLossDecrease_NoReplace}
	&f(\pmb{\omega}_{k+1}) \leq f(\pmb{\omega}_k) + B \left(1 + \gamma \right)/(\mu') \left \Vert \nabla f(\pmb{\omega}_k) \right\Vert^2 \times \nonumber\\
	&\rs  \sum_{v=1}^{V_{\breve{t}_k}} p_v \cdot \frac{\mathbf{1}\left(v \rs \in \rs \mathcal{C}_k, \mathrm{t}_v (k) \rs \leq \rs \mathrm{t}^{\mathrm{th}}(k)|d_v(k) \leq \mathrm{r} \right)} {q_v(k)p^{\mathrm{suc}}_{v} (k)} + \frac{L B^2 (1+\gamma)^2}{2{\mu'}^2} \times \rs\rs \\
	& \left\Vert \nabla f(\pmb{\omega}_k)\right\Vert^2 \sum_{v=1}^{V_{\breve{t}_k}} p_v \cdot \bigg( \frac{\mathbf{1}\left(v \rs \in \rs \mathcal{C}_k, \mathrm{t}_v (k) \rs \leq \rs \mathrm{t}^{\mathrm{th}}(k) | d_v(k) \leq \mathrm{r} \right)}{q_v(k)p^{\mathrm{suc}}_{v} (k)} \bigg)^2\rs\rs.\nonumber
\end{align}
Therefore, we can calculate the expected training loss decrease in one global round by taking the expectation of (\ref{oneIterLossDecrease_NoReplace}) as
\begin{align}
    & \mathbb{E} \left[f(\pmb{\omega}_{k+1})\right] - f(\pmb{\omega}_k) \leq  \mathbb{E} \bigg\{B \left(1 + \gamma \right)/(\mu') \left \Vert \nabla f(\pmb{\omega}_k) \right\Vert^2 \times \nonumber\\
	&\qquad \sum_{v=1}^{V_{\breve{t}_k}} p_v \cdot \frac{\mathbf{1}\left( v \in\mathcal{C}_k, \mathrm{t}_v (k) \leq \mathrm{t}^{\mathrm{th}}(k)| d_v(k) \leq \mathrm{r} \right)}{q_v(k)p^{\mathrm{suc}}_{v} (k)} \bigg\} + \nonumber\\
	&\qquad \mathbb{E} \bigg\{\frac{L B^2 (1+\gamma)^2}{2{\mu'}^2 } \left \Vert \nabla f(\pmb{\omega}_k)\right\Vert^2 \times \nonumber\\
	&\qquad \sum_{v=1}^{V_{\breve{t}_k}} p_v \cdot \bigg( \frac{\mathbf{1}\left(v \in \mathcal{C}_k, \mathrm{t}_v (k) \leq \mathrm{t}^{\mathrm{th}}(k) | d_v(k) \leq \mathrm{r} \right)}{q_v(k) p^{\mathrm{suc}}_{v} (k)} \bigg)^2 \bigg\},\nonumber\\
	&= \frac{B \left(1 + \gamma \right)}{\mu' } \left \Vert \nabla f(\pmb{\omega}_k) \right\Vert^2 +  \frac{L B^2 (1+\gamma)^2}{2{\mu'}^2 } \left \Vert \nabla f(\pmb{\omega}_k)\right\Vert^2 \rs\rs \times \rs\rs\rs\rs\rs\rs \nonumber\\
	&\quad \sum_{v=1}^{V_{\breve{t}_k}} p_v \cdot  \frac{\mathbb{E} \left\{[\mathbf{1}\left(v \in \mathcal{C}_k, \mathrm{t}_v (k) \leq \mathrm{t}^{\mathrm{th}}(k) | d_v(k) \leq \mathrm{r} \right)]^2\right\}}{{q_v(k)}^2{p^{\mathrm{suc}}_{v} (k)}^2} , \nonumber \\
	&\overset{(a)}{=} \frac{B \left(1 + \gamma \right)}{\mu'} \left \Vert \nabla f(\pmb{\omega}_k) \right\Vert^2 +  \frac{L B^2 (1+\gamma)^2}{2{\mu'}^2} \left \Vert \nabla f(\pmb{\omega}_k)\right\Vert^2 \times \nonumber \\
	&\qquad \sum\nolimits_{v=1}^{V_{\breve{t}_k}} [p_v/( {q_v(k)}^2{p^{\mathrm{suc}}_{v} (k)}^2)] \times \\
	&\mathbb{E} \big\{\big[\mathbf{1}\left(v \in \mathcal{C}_k \right) \times \mathbf{1} \big( \mathrm{t}_v (k) \leq \mathrm{t}^{\mathrm{th}}(k) | d_v(k) \leq \mathrm{r} \big) \big]^2\big\},\nonumber\\
	&\overset{(b)}{=} \frac{B \left(1 + \gamma \right)}{\mu'} \left \Vert \nabla f(\pmb{\omega}_k) \right\Vert^2 +  \frac{L B^2 (1+\gamma)^2}{2{\mu'}^2} \left \Vert \nabla f(\pmb{\omega}_k)\right\Vert^2 \times \nonumber \\
	&\qquad \sum\nolimits_{v=1}^{V_{\breve{t}_k}} [p_v /( {q_v(k)}^2{p^{\mathrm{suc}}_{v} (k)}^2)] \times \nonumber \\
	&\mathbb{E} \big\{\big[\mathbf{1}\left(v \in \mathcal{C}_k \right) \big]^2 \times \mathbb{E} \big[\mathbf{1} \big( \mathrm{t}_v (k) \leq \mathrm{t}^{\mathrm{th}}(k) | d_v(k) \leq \mathrm{r} \big) \big]^2\big\},\nonumber\\
	&=\frac{B \big(1 + \gamma \big)}{\mu'} \left \Vert \nabla f(\pmb{\omega}_k) \right\Vert^2 +  \frac{L B^2 (1+\gamma)^2}{2{\mu'}^2} \left \Vert \nabla f(\pmb{\omega}_k)\right\Vert^2 \times \nonumber\\
	&\qquad\qquad\qquad\qquad\qquad \sum\nolimits_{v=1}^{V_{\breve{t}_k}} [p_v/(q_v(k) p^{\mathrm{suc}}_{v} (k))],\nonumber\\
	&=\frac{B \left(1 + \gamma \right)}{\mu'} \bigg[1 + \frac{B L (1+\gamma)}{2 \mu' } \sum_{v=1}^{V_{\breve{t}_k}} \frac{p_v} {q_v(k) p^{\mathrm{suc}}_{v} (k)} \bigg]  \left \Vert \nabla f(\pmb{\omega}_k) \right\Vert^2, \nonumber
\end{align}
where we write $(a)$ using the fact that $\mathbf{1}\left(v \in \mathcal{C}_k \right)$ and $\mathbf{1}\left(\mathrm{t}_v (k) \leq \mathrm{t}^{\mathrm{th}}(k) | d_v(k) \leq \mathrm{r} \right)$ are independent.
Besides, in $(b)$ the outer and inner expectations are with respect to CV sampling and successful trained model receptions, respectively.

\section{Proof of Lemma \ref{lyaPunovLemma}} 
\label{proofLyaPunovLemma}

\noindent
Using the remaining payload buffer in (\ref{payloadQueue}), we can write
\begin{align}
\label{lyaPunovLem}
    &{Q_v(t+1)}^2 = \left([Q_v(t) - \kappa \cdot r_v(t)]^{+}\right)^2,\nonumber\\
    &\qquad \qquad \leq {Q_v(t)}^2 + \kappa^2 {r_v(t)}^2 - 2\kappa r_v(t) Q_v(t),\nonumber\\
    &{Q_v(t+1)}^2 - {Q_v(t)}^2 \leq \kappa^2 {r_v(t)}^2 - 2\kappa r_v(t) Q_v(t),\\
    &(1/2)\sum\nolimits_{v=1}^{|\mathcal{C}_k|} \left[ {Q_v(t+1)}^2 - {Q_v(t)}^2 \right]  \nonumber\\
    &\qquad \quad \leq - \sum\nolimits_{v=1}^{|\mathcal{C}_k|} \kappa r_v(t) Q_v(t) +  (1/2) \sum\nolimits_{v=1}^{|\mathcal{C}_k|} \kappa^2 {r_v(t)}^2, \nonumber
\end{align}
where the last inequality is obtained by dividing both sides by $2$ and summing up the inequalities for all $v \in \mathcal{V}_{\breve{t}_k}$.
Then, using (\ref{uplink_DataRate}), we can write 
\begin{align}
    & L(\mathbf{Q}(t+1)) -  L(\mathbf{Q}(t)) \leq - \sum\nolimits_{v=1}^{|\mathcal{C}_k|} \kappa r_v(t) Q_v(t) +\nonumber\\
    & (\kappa^2\omega^2 (1-\upsilon)^2 /2) \sum_{v=1}^{|\mathcal{C}_k|} \Big[ {\mathrm{I}_v(t)}^2 \Big(\sum_{z=1}^Z \log_2 (1+\Gamma_{v,z}(t)) \Big)^2 \Big],\nonumber\\
    & \overset{(a)}{\leq} - \sum\nolimits_{v=1}^{|\mathcal{C}_k|} \kappa r_v(t) Q_v(t) + (\kappa^2\omega^2 (1-\upsilon)^2/2) \sum\nolimits_{v=1}^{|\mathcal{C}_k|} \Big[ \nonumber\\
    &\qquad {\mathrm{I}_v(t)}^2 \cdot Z \sum\nolimits_{z=1}^Z \left(\log_2 (1+\Gamma_{v,z}(t)) \right)^2 \Big],\\
    & \overset{(b)} {\leq} 
    - \sum\nolimits_{v=1}^{|\mathcal{C}_k|} \kappa r_v(t) Q_v(t) + [(Z \kappa^2 \omega^2 (1-\upsilon)^2)/ (\ln 2)^2] \times \nonumber \\ &\quad \Big[\sum\nolimits_{v=1}^{|\mathcal{C}_k|}  {\mathrm{I}_v(t)}^2 \cdot \sum\nolimits_{z=1}^Z (\mathrm{I}_{v,z}(t) \cdot P_{v,z}(t) \left\Vert \mathbf{h}_{v,z}(t) \right\Vert^2)/ (\omega \varsigma^2)\Big],\nonumber
\end{align}
where ($a$) comes from Cauchy-Schwarz inequality on real numbers $\left(\sum_{z=1}^{Z} a_z \cdot 1 \right)^2 \leq \left(\sum_{z=1}^{Z} a_z^2\right) \cdot \left(\sum_{z=1}^{Z} 1^2\right)$ and ($b$) stems from $(\log_2(1+a))^2 \leq (2a /(\ln 2)^2)$.

Now, since both $\mathrm{I}_v(t)$ and $\mathrm{I}_{v,z}(t)$ are binary indicator functions with maximum value of $1$, we can write the following:
\begin{align}
    &\frac{Z\kappa^2 \omega^2 (1-\upsilon)^2} {(\ln 2)^2} \sum_{v=1}^{|\mathcal{C}_k|}  {\mathrm{I}_v(t)}^2 \cdot \sum_{z=1}^Z \frac{\mathrm{I}_{v,z}(t) \cdot P_{v,z}(t) \left\Vert \mathbf{h}_{v,z}(t) \right\Vert^2} {\omega \varsigma^2} \nonumber\\
    & \overset{(a)}{\leq} \frac{Z\kappa^2 \omega^2 (1-\upsilon)^2} {(\ln 2)^2} \sum_{v=1}^{|\mathcal{C}_k|}  1^2 \cdot \sum_{z=1}^Z \frac{1 \cdot \frac{P_v^{\text{max}}}{Z} \left\Vert \mathbf{h}_{v,z}(t) \right\Vert^2} {\omega \varsigma^2},\\
    & = \frac{\kappa^2 \omega^2(1-\upsilon)^2} {(\ln 2)^2} \sum_{v=1}^{|\mathcal{C}_k|}  \sum_{z=1}^Z (P_v^{\text{max}} \left\Vert \mathbf{h}_{v,z}(t) \right\Vert^2)/ (\omega \varsigma^2) = \varpi, \nonumber
\end{align}
where, in ($a$), we use the fact that $P_v^{\mathrm{max}}$ is equally distributed among all pRBs.

Then, using $\varpi$, adding $C \cdot (- \sum\nolimits_{v=1}^{|\mathcal{C}_k|} r_v(t) + \bar{\beta} (k,t)\sum\nolimits_{v=1}^{|\mathcal{C}_k|} \sum\nolimits_{z=1}^Z P_{v,z}(t))$ on both sides of (\ref{lyaPunovLem}) and taking expectation on both sides conditioned on the queue state $Q_v(t)$, we reach to (\ref{drifPenaltyEq}).

\bibliography{Reference.bib}
\bibliographystyle{IEEEtran}

\end{document}